\documentclass[journal,twoside]{IEEEtran}

\usepackage[T1]{fontenc}
\usepackage{times}
\usepackage[pdftex]{hyperref} %pagebackref
\usepackage{amsmath, amsthm, amssymb}
\usepackage{graphicx,amsfonts, color, multirow, marvosym, latexsym,array,booktabs}
\usepackage[table]{xcolor}% http://ctan.org/pkg/xcolor
\input{Qcircuit.tex}

\def\leqclean{\stackrel{\!{\text{\small clean}}}{\leq}\!}

\def\leqlo{\mbox{\boldmath$\leq$}_{\scriptstyle}\,}
\def\leqclo{\mbox{\boldmath$\leq$}_{\text{CL\,}}}
\def\geqlo{\mbox{\boldmath$\geq$}_{\scriptstyle}\,}
\def\geqclo{\mbox{\boldmath$\geq$}_{\text {CL}}\,}
\def\eqlo{\mbox{\boldmath$=$}_{\scriptstyle}\,}
\def\eqclo{\mbox{\boldmath$=$}_{\text{CL}}\,}
\DeclareMathOperator{\Emb}{[\text{\EUR\EUR}]}
\DeclareMathOperator{\Hom}{Hom}
\DeclareMathOperator{\poly}{poly}
\DeclareMathOperator{\rank}{rank}
\DeclareMathOperator{\spec}{spec}
\DeclareMathOperator{\Tr}{Tr}

\def\deficit{\Delta_{\mathrm{sim}}}
\def\bbC{\mathbb{C}}
\def\bbE{\mathbb{E}}
\def\bbR{\mathbb{R}}
\def\bbZ{\mathbb{Z}}

\def\cI{{\cal I}}

\def\cM{{\cal M}}
\def\cN{{\cal N}}

\def\cP{{\cal P}}
\def\cQ{{\cal Q}}
\def\cS{{\cal S}}
\def\cT{{\cal T}}
\def\cU{{\cal U}}
\def\cW{{\cal W}}
\def\cX{{\cal X}}
\def\cY{{\cal Y}}
\def\bp{{\bf p}}
\def\bq{{\bf q}}
\def\br{{\bf r}}
\def\bP{{\bf P}}
\def\bQ{{\bf Q}}

\def\<{\langle}
\def\>{\rangle}
\def\L{\left}
\def\R{\right}
\def\ot{\otimes}
\newcommand{\proj}[1]{\ket{#1}\bra{#1}}
\newcommand{\braket}[2]{\left\langle #1 \middle| #2 \right\rangle}

\def\bl{{\bar{\lambda}}}
\def\eps{\epsilon}

\def\tr{\Tr}
\newcommand{\supp}{{\mathrm{supp}}}
\newcommand{\id}{{\operatorname{id}}}

\newcommand\bea{\begin{eqnarray}}
\newcommand\eea{\end{eqnarray}}
\newcommand{\be}{\begin{equation}}
\newcommand{\ee}{\end{equation}}
\def\besp#1\eesp{\begin{equation}\begin{split}#1\end{split}\end{equation}}
\newcommand{\bmu}{\begin{multline}}
\newcommand{\emu}{\end{multline}}
\def\non{\nonumber}
\def\ba#1\ea{\begin{align}#1\end{align}}
\def\bas#1\eas{\begin{align*}#1\end{align*}}
\def\bmu#1\emu{\begin{multline}#1\end{multline}}

\newtheorem*{rep@theorem}{\rep@title}
\newcommand{\newreptheorem}[2]{%
\newenvironment{rep#1}[1]{%
 \def\rep@title{#2 \ref{##1}}%
 \begin{rep@theorem}}%
 {\end{rep@theorem}}}
\makeatother

\newtheorem{lemma}{Lemma}
\newreptheorem{lemma}{Lemma}

\newtheorem{theorem}[lemma]{Theorem}
\newtheorem{definition}[lemma]{Definition}
\newtheorem{corollary}[lemma]{Corollary}

\newcommand{\smfrac}[2]{\mbox{$\frac{#1}{#2}$}}
\def\half{\smfrac{1}{2}}

\def\ra{\rightarrow}
\def\la{\leftarrow}
\def\bit#1\eit{\begin{itemize}#1\end{itemize}}
\def\benum#1\eenum{\begin{enumerate}#1\end{enumerate}}
\newcommand{\eq}[1]{Eq.~(\ref{eq:#1})}
\newcommand{\eqs}[2]{Eqs.~(\ref{eq:#1}) and (\ref{eq:#2})}
\newcommand{\fig}[1]{Fig.~\ref{fig:#1}}
\newcommand{\defref}[1]{Definition~\ref{def:#1}}
\newcommand{\secref}[1]{Sec.~\ref{sec:#1}}
\newcommand{\thmref}[1]{Theorem~\ref{thm:#1}}
\newcommand{\lemref}[1]{Lemma~\ref{lem:#1}}

\def\Usch{U_{\text{Sch}}}

\def\oldcomment#1{}
\newcommand{\markov}[3]{\ensuremath{I(#1;#3|#2)=0}}

\begin{document}

\title{The quantum reverse Shannon theorem
and resource tradeoffs for simulating quantum channels}
\author{Charles H. Bennett, Igor Devetak, Aram W. Harrow, Peter W. Shor
  and Andreas Winter
  \thanks{Charles H. Bennett is with the IBM T.J. Watson Research
    Center, Yorktown Heights, NY 10598 (USA). This work was funded in
    part by ARDA contract DAAD19-01-0056 and DARPA QUEST contract
    HR0011-09-C0047.  {Email:} \tt{chdbennett@gmail.com}}
  \thanks{Igor Devetak is with IMC Financial Markets, Poststrasse 20,
    6300 Zug, Switzerland.  This work was performed while he was at
    IBM T.J. Watson Research Center and the Department of Electrical
    Engineering at USC.  {Email:} {\tt igor.devetak@gmail.com}}
\thanks{Aram W. Harrow is with the Department of Physics,  Massachusetts Institute of Technology, 77 Massachusetts Avenue, Cambridge, MA 02139, USA.
This work was also performed while he was at University of Bristol
and the University of Washington.
He was funded by NSF grants CCF-0916400 and CCF-1111382 and ARO contract
W911NF-12-1-0486.  {Email:} {\tt aram@mit.edu}}
\thanks{Peter W. Shor is with the Department of Mathematics,
  Massachusetts Institute of Technology, 77 Massachusetts Avenue,
  Cambridge, MA 02139, USA and was supported in part by NSF grants
  CCF-0431787 (``Quantum Channel Capacities and Quantum Complexity'')
  and CCF-0829421 (``Physics Based Approaches to Quantum
  Algorithms''), as well as the NSF STC on Science of Information.
  Email: {\tt shor@math.mit.edu}}
\thanks{Andreas Winter is with ICREA and F\'{\i}sica Te\`{o}rica:
  Informaci\'{o} i Fenomens Qu\`{a}ntics, Universitat Aut\`{o}noma de
  Barcelona, ES-08193 Bellaterra (Barcelona), Spain.
  During preparation of this
  paper he was also affiliated with the Department of Mathematics,
  University of Bristol and the Centre for Quantum Technologies,
  National University of Singapore. He acknowledges support by the
  U.K. EPRSC grant ``QIP IRC'', the Royal Society, the Philip Leverhulme
  Trust, EC integrated project QAP (contract IST-2005-15848),  as well as
  STREPs QICS and QCS, and finally the ERC Advanced Grant ``IRQUAT''.
  {Email:} {\tt der.winter@gmail.com}.}
}

\date{\today}

\maketitle

\begin{abstract}
  Dual to the usual noisy channel coding problem, where a noisy
  (classical or quantum) channel is used to simulate a noiseless one,
  reverse Shannon theorems concern the use of noiseless channels to
  simulate noisy ones, and more generally the use of one noisy channel
  to simulate another.  For channels of nonzero capacity, this
  simulation is always possible, but for it to be efficient, auxiliary
  resources of the proper kind and amount are generally required. In
  the classical case, shared randomness between sender and receiver is
  a sufficient auxiliary resource, regardless of the nature of the
  source, but in the quantum case the requisite auxiliary resources
  for efficient simulation depend on both the channel being simulated,
  and the source from which the channel inputs are coming. For tensor
  power sources (the quantum generalization of classical IID sources),
  entanglement in the form of standard ebits (maximally entangled
  pairs of qubits) is sufficient, but for general sources, which may
  be arbitrarily correlated or entangled across channel inputs,
  additional resources, such as entanglement-embezzling states or
  backward communication, are generally needed.  Combining existing
  and new results, we establish the amounts of communication and
  auxiliary resources needed in both the classical and quantum cases,
  the tradeoffs among them, and the loss of simulation efficiency when
  auxiliary resources are absent or insufficient.  In particular we
  find a new single-letter expression for the excess forward
  communication cost of coherent feedback simulations of quantum
  channels (i.e. simulations in which the sender retains what would
  escape into the environment in an ordinary simulation), on
  non-tensor-power sources in the presence of unlimited ebits but no
  other auxiliary resource. Our results on tensor power sources
  establish a strong converse to the entanglement-assisted capacity
  theorem.
\end{abstract}

\vfill          %To move TOC to other column (AW)...
\tableofcontents

\section{Introduction}
\label{sec:intro}
\subsection{Motivation}
In classical information theory, Shannon's celebrated noisy channel
coding theorem~\cite{Shannon48} establishes the ability of any noisy memoryless
channel $N$ to simulate an ideal noiseless binary channel, and shows
that its asymptotic efficiency or capacity for doing so is given by
a simple expression \ba C(N) &= \max_{p} I(X;Y) \non \\&= \max_{p}
\bigl\{ H(X) + H(Y) - H(XY) \bigr\}, \label{eq:mutinfo} \ea where
$H$ is the entropy, $X$ the input random variable and $Y=N(X)$ the
induced output variable. The capacity, in other words, is equal to
the maximum, over input distributions $p$, of the input-output
mutual information for a single use of the channel.  Somewhat more
recently, a dual theorem, the classical ``reverse Shannon theorem''
was proved~\cite{BSST01}, which states that for any channel
$N$ of capacity $C$, if the sender and receiver share an unlimited
supply of random bits, an expected $C n + o(n)$ uses of a noiseless
binary channel are sufficient to exactly simulate $n$ uses of the
channel. In~\cite{Winter:RST} a version of this construction is
given which achieves asymptotically perfect simulation, works on a
uniform blocksize $C n + o(n)$, and uses an amount of shared
randomness increasing linearly with $n$, in contrast to the
exponential amount used in~\cite{BSST01}. These simulations do not
depend on the nature of the source, and work for arbitrarily varying
as well as IID sources.

Together with the original Shannon theorem, these theorems show that
in the presence of shared randomness, the asymptotic properties of a
classical channel can be characterized by a single parameter, its
capacity; with all channels of equal capacity being able to simulate
one another with unit asymptotic efficiency in the presence of
shared randomness. In~\cite{BSST01} a quantum analog of the reverse
Shannon theorem was conjectured, according to which quantum channels
should be characterizable by a single parameter in the presence of
unlimited shared entanglement between sender and receiver.

A (discrete memoryless) quantum channel can be viewed physically
as a process wherein a quantum system, originating with a sender
Alice, is split into a component for a receiver Bob and another for an
inaccessible environment (commonly referred to as Eve).  Mathematically it can be viewed as an isometric
embedding $\cN^{A\ra BE}$ of Alice's Hilbert space ($A$) into the joint Hilbert space
of Bob ($B$) and Eve $(E)$.  Tracing out Eve yields a completely positive,
trace-preserving linear map on density operators from $A$ to $B$,
which we denote $\cN^{A\ra B}$. Operationally, the two pictures are
equivalent, but we will sometimes find it convenient mathematically to
work with one or the other.

The theory of quantum channels is richer
and less well understood than that of classical channels. Unlike
classical channels, quantum channels have multiple inequivalent
capacities, depending on what one is
trying to use them for, and what additional resources are brought
into play. These include
\begin{itemize}
\item The ordinary classical capacity $C$, defined as the maximum
asymptotic rate at which classical bits can be transmitted
reliably through the channel, with the help of a quantum encoder
and decoder.
\item The ordinary quantum capacity $Q$, which is the maximum
asymptotic rate at which qubits can be transmitted under similar
circumstances.
\item The private classical capacity $P$, which is the maximum rate at
  which classical bits can be transmitted to Bob while remaining
  private from Eve, who is assumed to hold the channel's environment $E$.
\item The classically assisted quantum capacity $Q_2$, which is the
maximum asymptotic rate of reliable qubit transmission with the
help of unlimited use of a 2-way classical side channel between
sender and receiver.
\item The entanglement-assisted classical capacity
$C_E$~\cite{BSST99,BSST01}, which is the maximum asymptotic rate of
reliable bit transmission with the help of unlimited pure state
entanglement shared between the sender and receiver.
\item Similarly, one can define the entanglement-assisted quantum
  capacity $Q_E$~\cite{BSST99,BSST01}, which is simply
  $\frac{1}{2}C_E$, by teleportation~\cite{BBCJPW98} and
  super-dense coding~\cite{BW92}.\footnote{Another powerful assistive
    resource, unlimited noiseless quantum back-communication from
    receiver to sender, turns out to be equivalent to unlimited shared
    entanglement~\cite{Bowen-feedback}.  Thus the capacity of a channel assisted by such
    back-communication is $C_E$ for classical messages and $Q_E$ for
    quantum messages.}
\
\end{itemize}

Somewhat unexpectedly, the entanglement assisted capacities are the simplest to calculate, being given by an expression analogous to
\eq{mutinfo}. In \cite{BSST01} (see also \cite{Holevo02}) it
was shown that
\be C_E(\cN) = \max_\rho
 \bigl\{  H(\rho) + H({\cal N}(\rho)) - H({I} \otimes {\cal N} (\Phi_\rho)) \bigr\}
\label{eq:cedef},
\ee
where the optimization is over all density matrices $\rho$ on $A$ and
$\Phi_\rho^{RA}$ is a purification of $\rho$ by a reference system $R$
(meaning that $\Phi_\rho^{RA}$ is a pure state and $\Tr_R
\Phi_\rho^{RA} = \rho^{A}$).  The entanglement-assisted capacity
formula \eq{cedef} is formally identical to \eq{mutinfo}, but with
Shannon entropies replaced by von Neumann entropies.  It shares the
desirable property with \eq{mutinfo} of being a concave function of
$\rho$, making it easy to compute~\cite{AC97}.
We can
alternately write the RHS of \eq{cedef} as
\be \max_\rho I(R;B)_\rho, \ee
using the definitions
 \bas \ket{\Psi} &=(I^R \ot \cN^{A\ra BE})\ket{\Phi_{\rho}^{RA}} \\
I(R;B)_\rho & = I(R;B)_\Psi = H(R)_\Psi + H(B)_\Psi - H(RB)_\Psi \\
&  =
H(\Psi^R) + H(\Psi^B) - H(\Psi^{RB}).\eas
  We will use $I(R;B)_\rho$ and
$I(R;B)_\Psi$ interchangeably, since the mutual information and other
entropic properties of $\Psi$ are uniquely determined by $\rho$.

%Here the first term is analogous to a classical channel's input
%entropy, the second to its output entropy, and the third to the joint
%input-output entropy.

Aside from the constraints $Q \leq P \leq C \leq C_E$, and $Q \leq Q_2$,
which are obvious consequences of the definitions, and $Q_2\leq
Q_E=\frac{1}{2}C_E$, which follows from \cite{TGW13}, the five capacities appear to
vary rather independently (see for example \cite{BDSS04} and \cite{SS09}).
Except in special cases, it is not possible,
without knowing the parameters of a
channel, to infer any one of these capacities from the other four.

This complex situation naturally raises the question of how many
independent parameters are needed to characterize the important
asymptotic, capacity-like properties of a general quantum channel. A
full understanding of quantum channels would enable us to calculate
not only their capacities, but more generally, for any two channels
${\cal M}$ and ${\cal N}$, the asymptotic efficiency (possibly zero)
with which ${\cal M}$ can simulate ${\cal N}$, both alone and in
the presence of auxiliary resources such as classical communication
or shared entanglement.

One motivation for studying communication in the presence of
auxiliary resources is that it can simplify the classification of
channels' capacities to simulate one another.  This is so because if
a simulation is possible without the auxiliary resource, then the
simulation remains possible with it, though not necessarily vice
versa. For example, $Q$ and $C$ represent a channel's asymptotic
efficiencies of simulating, respectively, a noiseless qubit channel
and a noiseless classical bit channel. In the absence of auxiliary
resources these two capacities can vary independently, subject to
the constraint $Q\leq C$, but in the presence of unlimited prior
entanglement, the relation between them becomes fixed: $C_E=2Q_E$,
because entanglement allows a noiseless 2-bit classical channel to
simulate a noiseless 1-qubit channel and vice versa (via
teleportation~\cite{BBCJPW98} and superdense coding~\cite{BW92}).
Similarly the auxiliary resource of shared randomness simplifies the
theory of classical channels by allowing channels to simulate one
another efficiently according to the classical reverse
Shannon theorem.

\subsection{Terminology}
The various capacities of a quantum channel ${\cal N}$ may be defined
within a framework where asymptotic communication resources and
conversions between them are treated abstractly \cite{DHW05}. Many
independent uses of a noisy channel $\cN$, i.e.~$\cN^{\otimes n}$,
corresponds to an asymptotic resource $\<\cN\>$, while
standard resources such as ebits (maximally-entangled pairs of qubits,
also known as EPR pairs), or instances of a noiseless qubit channel
from Alice to Bob are denoted $[qq]$ and $[q\ra q]$ respectively.
Their classical analogues are $[cc]$ and $[c\ra c]$, which stand for
bits of shared randomness (rbits), and uses of noiseless classical bit
channels (cbits).  Communication from Bob to Alice is denoted by
$[q\la q]$ and $[c\la c]$.  Within this framework, coding theorems can
be thought of as transformations from one communication resource to
another, analogous to reductions in complexity theory, but involving
resources that are quantitative rather than qualitative, the rate
(if other than 1) being indicated by a coefficient preceding the
resource expression.  We consider two kinds of asymptotic {\em
resource reducibility} or {\em resource inequality}~\cite{DHW05}: viz.
asymptotic reducibility via local operations $\leqlo_{\text{L}}$, usually abbreviated $\leqlo$, and asymptotic reducibility via
clean local operations $\leqclo$.  A resource $\beta$ is said to be
locally asymptotically reducible to $\alpha$ if there is an
asymptotically faithful transformation from $\alpha$ to $\beta$ via
local operations: that is, for any $\eps,\delta>0$ and for all
sufficiently large $n$, $n(1+\delta)$ copies of $\alpha$ can be
transformed into $n$ copies of $\beta$ with overall error
$<\eps = o(1)$.  Here, and throughout the paper, we use $o(1)$
to mean a quantity that approaches zero as $n\ra\infty$.
We use ``error'' to
refer to the trace distance in the context of states, which is defined as
$$\frac{1}{2}\|\rho-\sigma\|_1 = \frac{1}{2}\tr|\rho-\sigma|.$$
For channels, ``error'' refers to the diamond
norm~\cite{Kitaev:02a} (see also
\cite{CB-book,KW03}).
The example most studied in this paper is when the target resource
$\beta = \<\cN\>$ with a channel $\cN$.
The initial resource $\alpha$ is transformed, via a protocol involving
local operations, into a channel $\cN^{\prime(n)}$, with diamond-norm error
\[
 \| \cN^{\otimes n}-\cN^{\prime(n)} \|_\diamond =
   \max_{\rho} \left\| \bigl(\id_R \otimes (\cN^{\otimes n}-\cN^{\prime(n)})\bigr)(\Phi_\rho) \right\|_1,
\]
where the maximization is over states $\rho$ on $A^n$ and $\Phi_\rho$
is an arbitrary purification of it.

The clean version of this reducibility,
$\leqclo$, which is important when we
wish to coherently superpose protocols, adds the restriction that any
quantum subsystem discarded during the transformation be in the $\ket{0}$
state up to an error that vanishes in the limit of large $n$.  When
$\alpha \leqlo \beta$ and $\beta \leqlo \alpha$ we have a resource
equivalence, designated $\eqlo_{\text{L}}$, or $\eqlo$, or for
the clean version $\eqclo$. Resource reducibilities and equivalences
will often be referred to as resource relations or RRs.

For example, the coding theorem for entanglement-assisted classical
communication can be stated as \be \<\cN\> + \infty [qq] \ \geqlo\
C_E(\cN)\ [c\ra c].\ee where $C_E(\cN)$ is defined as in \eq{cedef}.

In this language, to simulate (resp. cleanly simulate) a channel $\cN$
is to find standard resources $\alpha$ (made up of qubits, ebits,
cbits and so on) such that $\<\cN\>\leqlo \alpha$ (resp. $\leqclo$).
For example, the simplest form of the classical reverse Shannon
theorem can be stated as $\forall_N \<N\> \leqlo C(N)[c\ra c] + \infty
[cc]$, with $C(N)$ defined in \eq{mutinfo}.

We will also introduce notation for two refinements of the problem.
First, we (still following \cite{DHW05}) define the {\em relative
resource}  $\<\cN\!\!:\!\!\rho\>$ as many uses of a channel $\cN$ whose
asymptotic accuracy is guaranteed or required only when $n$ uses of $\cN$ are
fed an input of the form $\rho^{\ot n}$. This means that the error
is evaluated with respect to $\Phi_\rho^{\ot n}$ rather than the worst
case entangled input state:
\[
  \| \cN^{\otimes n}-\cN^{\prime(n)} \|_{\rho^{\ot n}} =
    \left\| \bigl(\id_R \otimes (\cN^{\otimes n}-\cN^{\prime(n)})\bigr)(\Phi_\rho^{\ot n}) \right\|_1.
\]

Most coding theorems still
apply to relative resources, once we drop the maximization over
input distributions. So for a classical channel $\<N\!:\!p\> \;\geqlo\;
I(X;Y)_p [c\ra c]$ and for a quantum channel $\<\cN:\rho\>
+\infty[qq] \;\geqlo\; I(R;B)_\rho[c\ra c]$ (notation following \eq{cedef}).
%with the last term representing the quantum mutual information
%between the channel output $B$ and the reference system $R$ purifying
%the channel input, when the input is distributed according to density
%matrix $\rho$, as in \eq{cedef}.

Second, we will consider simulating channels with {\em passive feedback}.
The classical version of a passive feedback channel has Alice obtain a copy of Bob's output $Y=N(X)$.  We denote this form of channel by $N_F$ if the original channel is $N$.
For a quantum channel, we cannot give Alice a copy of Bob's output because of the no-cloning theorem~\cite{WZ82}, but instead define a {\em coherent feedback\/} version of the channel as an isometry in which the part of the output that does not go to Bob is retained by Alice, rather than escaping to the environment~\cite{Winter-ident}. We denote this $\cN_F^{A\ra BE}$, where the subscript $F$ indicates that $E$ is retained by Alice. When it is clear from the context, we will henceforth use "feedback" to mean conventional passive feedback for a classical channel and coherent feedback for a quantum channel.\footnote{The term "feedback" has been used in multiple ways.  Bowen\cite{Bowen-feedbacks} compares several kinds of feedback, both quantum and classical. In his terminology, both the classical and coherent feedbacks we consider here are {\em passive\/}, meaning that they do not grant the sender and receiver any additional resource but require them to perform an additional task (e.g. giving the sender a copy of the output) beyond what would have been required in an ordinary execution or simulation of the channel. For this reason passive feedback capacities are never greater than the corresponding plain capacities. {\em Active\/} feedback, by contrast, involves granting the sender and receiver an additional resource (e.g. unlimited quantum back-communication, as in \cite{Bowen-feedback}), to perform the {\em same\/} task as in a plain execution or simulation of the channel. Accordingly, active feedback capacities are never {\em less\/} than the corresponding plain capacities. We do not discuss active feedback further in this paper.}

Coherent feedback is an example of quantum state
redistribution~\cite{HOW05,DY08,YD09} in which the same global pure state $\Psi$ is redistributed among a set of parties. The redistribution corresponding to a feedback channel $\cN_F^{A\ra BE}$ involves Alice, Bob, and a purifying reference system $R$. Alice's share $A$ of the initial state $\Psi^{A:R}$, is split
into two parts, $E$ and $B$, with $E$ remaining with her
party, while $B$ passes to Bob, who initially held
nothing, leading to a final state $\Psi^{E:B:R}$.

Classical and coherent feedback are thus rather different notions, indeed one might say opposite notions, since in coherent feedback Alice gets to keep {\em everything but\/} what Bob receives, and as a result coherent feedback is sometimes a stronger resource than free classical back-communication.  Despite these differences, there are close parallels in how feedback affects the tradeoff between static
resources (rbits, ebits) and dynamic resources (cbits, qubits)
required for channel simulation. In both cases, when the static
resource is restricted, simulating a non-feedback version of the
channel requires less of the dynamic resource than simulating a
feedback version, because the non-feedback simulation can be
economically split into two sequential stages.  For a feedback
simulation, no such splitting is possible.

Other notational conventions we adopt are as follows.  If $\ket\psi$
is a pure state then $\psi := \proj\psi$ and $\psi^X$ refers to the
state of the $X$ subsystem of $\psi$.  For a subsystem $X$, we define
$|X|$ to be the cardinality of $X$ if $X$ is classical or $\dim X$
when $X$ is quantum.  We take $\log$ and $\exp$ to be base 2.  The
fidelity~\cite{Uhlmann76} between $\rho$ and $\sigma$ is
$\|\sqrt{\rho}\sqrt{\sigma}\|_1$ and the trace distance is
$\frac{1}{2}\|\rho-\sigma\|_1$.  For a channel $\cN^{A\ra B}$ we
  observe that $\cN = \tr_E \circ \cN_F$ and we define the {\em
    complementary channel} $\hat\cN^{A\ra E} := \tr_B\circ \cN_F$.
  Since isometric extensions of channels are unique only up to an
  overall isometry on $E$, the same is true for the complementary
  channel~\cite{Holevo07}, and our results will not be affected by
  this ambiguity.

%For classical random variables $X,Y,Z$, we use $\markov{X}{Z}{Y}$ to
%express the fact that $X,Z,Y$ form a Markov chain, i.e. that
%$I(X;Y|Z)=0$.
Additional definitions related to entanglement spread will be
introduced in \secref{spread}.

%Following Ref.~\cite{HW02}, we define
%\ba
%H_{0,\eps}(\rho) & := \log\min_\Pi \rank \Pi\rho\Pi \\
%H_{\infty,\eps}(\rho) &:= -\log \min_\Pi \| \Pi\rho\Pi\|_\infty \\
%\Delta_{\eps}(\rho) &:= \log\min_\Pi \| \Pi\rho\Pi\|_\infty \cdot \rank\Pi\rho\Pi,
%\ea
%where in each minimization, we let $\Pi$ range over all projections that commute with $\rho$ and satisfy $\tr\Pi\rho \geq 1-\eps$.  Observe that $\Delta_\eps(\rho) \geq H_{0,\eps}(\rho) - H_{\infty,\eps}(\rho)$.

\subsection{Overview of results}
In this paper we consider what resources are required to simulate a
quantum channel.  In particular, one might hope to show, by analogy
with the classical reverse Shannon theorem, that $Q_E(\cN)$ qubits of
forward quantum communication, together with a supply of shared
ebits, suffice to efficiently simulate any quantum channel $\cN$ on
any input.  This turns out not to be true in general (see below), but
it is true in some important special cases:
\bit
\item When the input is of tensor power form $\rho^{\ot n}$, for some
  $\rho$.  In this case, we are simulating the relative resource
  $\<\cN:\rho\>$.
\item When the channel $\cN$ has the property that its output entropy
  $H(\cN(\rho))$ is uniquely determined by the state of the
  environment.  Such channels include those with classical inputs or
  outputs.
\eit
However, for general channels on general (i.e. non-tensor-power)
inputs, we show that efficient simulation requires additional
resources beyond ordinary entanglement. Any of the following resources
will suffice:
\bit
\item more general forms of entanglement, such as an
  entanglement-embezzling state~\cite{vDH03}, in place of the
  supply of ordinary ebits, or
\item additional communication from Alice to Bob, or
\item backward classical or quantum  communication, from Bob to Alice.
\eit
The quantum reverse Shannon theorem is thus more fastidious than its
classical counterpart.  While classical shared random bits (rbits)
suffice to make all classical channels equivalent and cross-simulable,
standard ebits cannot do so for quantum channels. The reason is that
quantum channels may require different numbers of ebits to simulate on
different inputs.  Therefore, to maintain coherence of the simulation
across a superposition of inputs, the simulation protocol must avoid
leaking to the environment these differences in numbers of ebits used.
Fortunately, if the input is of tensor power form $\rho^{\ot n}$, the
entanglement ``spread'' required is rather small ($O(\sqrt{n})$), so
it can be obtained at negligible additional cost by having Alice
initially share with Bob a slightly generous number of ebits, then at
the end of the protocol return the unused portion for him to destroy.
On non-tensor-power inputs the spread may be $O(n)$, so other
approaches are needed if one is to avoid bloating the forward
communication cost.  If the
channel itself already leaks complete information about the output
entropy to the environment, there is nothing more for the simulation
to leak, so the problem becomes moot.  Otherwise, there are several
ways of coping with a large entanglement spread without excessive
forward communication, including: 1) using a more powerful
entanglement resource in place of standard ebits, namely a so-called
entanglement-embezzling state \cite{vDH03},
\be \ket{\varphi_N} = \frac{1}{\sqrt{\sum_{j=1}^N \frac{1}{j}}}
                        \sum_{j=1}^N \frac{1}{\sqrt{j}}\ket{j}\ket{j}
\ee
from which (in the limit of large $N$) a variable amount of
entanglement can be siphoned off without leaving evidence of how much
was taken, or 2) using a generous supply of standard ebits but
supplementing the protocol by additional backward classical
communication to coherently ``burn off'' the unused ebits.  We discuss
the role of entanglement spread in the quantum reverse Shannon theorem
in \secref{spread}.  There we will precisely define the resource
$\Emb$, which informally can be thought of as an embezzling state
$\ket{\varphi_N}$ with $N$ allowed to be arbitrarily large.

When simulating quantum feedback channels, we are sometimes able to
establish resource equivalences rather than reducibilities, for
example (as we will see in part (a) of \thmref{qrst})
\be \<\cN_F:\rho\> \eqlo \frac{1}{2}I(R;B)[q\ra q] + \frac{1}{2}I(E;B)[q
q]. \label{eq:feedback-equality}\ee
This both indicates the numbers of qubits and ebits asymptotically
necessary and sufficient to perform the redistribution
$\Psi^{A:R}\ra\Psi^{E:B:R}$ on tensor powers of a source with
density matrix $\rho^A$, and expresses the fact that any combination
of resources asymptotically able to perform the feedback simulation
of $\cN$ on $\rho$ can be converted into the indicated quantities of
qubits and ebits.  These results reflect the fact that the state
redistribution performed by a quantum feedback channel is
asymptotically reversible.   One interesting special case is when
$\cN$ is a noiseless classical channel, in which case
\eq{feedback-equality} reduces to the ``cobit'' resource
equality~\cite{Har03}.  This observation, and our derivation of
\eq{feedback-equality}, are due to \cite{devetak-triangle}.

\emph{Applications:} Our results also have implications for proving rate-distortion
theorems and strong converses for the entanglement-assisted
capacities.  The rate-distortion problem is a variant of the reverse
Shannon theorem which differs in that instead of simulating a specific
channel with high blockwise fidelity the goal is to minimize an
average distortion condition.  This is a less stringent condition than
demanded by the reverse Shannon theorem, so our simulations imply
rate-distortion theorems at the rate one would expect: the least
capacity of any channel satisfying the distortion bound.  This
connection was observed for classical channels in \cite{Winter:RST}
(see also \cite{SS96})
and for quantum channels in \cite{Datta-distort}.
The second application of our result is to derive a strong converse
theorem, meaning that attempting to send classical bits through a
quantum channel at rates above $C_E$ results in an exponentially small
success probability.  We discuss this application further in \secref{converses}.

\emph{Coordination capacity:} Another interpretation of reverse
Shannon theorems is in terms of ``coordination capacities'', defined
as the minimum rate of communication required to achieve certain
correlated probability distributions subject to constraints on some of
the variables~\cite{coordination}.  For example, the classical reverse
Shannon theorem corresponds to the goal of reproducing the
input-output distribution of a channel given one party's knowledge of
the input.  However, the framework of coordination capacity also
encompasses many network-coding generalizations of this task.

\section{Statement of results}
Figure \ref{fig:Personae} shows the parties, and corresponding random
variables or quantum subsystems, involved in the operation of a
discrete memoryless classical channel (top left) and a discrete
memoryless quantum channel (top right).  Dashed arrows indicate
additional data flows characterizing a feedback channel.  The bottom
of the figure gives flow diagrams for simulating such channels using,
respectively, a classical encoder and decoder (bottom left) or a
quantum encoder and decoder (bottom right). Shared random bits (rbits)
and forward classical communication (cbits) are used to simulate the
classical channel; shared entanglement (ebits) and forward quantum
communication (qubits) are used to simulate the quantum channel. As
usual in Shannon theory, the encoder and decoder typically must
operate in parallel on multiple inputs in order to simulate multiple
channel uses with high efficiency and fidelity.

Where it is clear from context we will often use upper case letters
$X$, $B$, etc.  to denote not only a classical random variable (or
quantum subsystem) but also its marginal probability distribution (or
density matrix) at the relevant stage of a protocol, for example
writing $H(B)$ instead of $H(\rho^B)$.  Similarly we write $I(E;B)$
for the quantum mutual information between outputs $E$ and $B$ in the
upper right side of Figure \ref{fig:Personae}.  However, it is not
meaningful to write $I(A;B)$, because subsystems $A$ and $B$ do not
exist at the same time. Thus the conventional classical notation
$I(X;Y)$ for the input-output mutual information may be considered to
refer, in the quantum way of thinking, to the mutual information
between $Y$ and a {\em copy\/} of $X$, which could always have been
made in the classical setting.

\begin{figure}[htbp]
\includegraphics[width=3.5in]{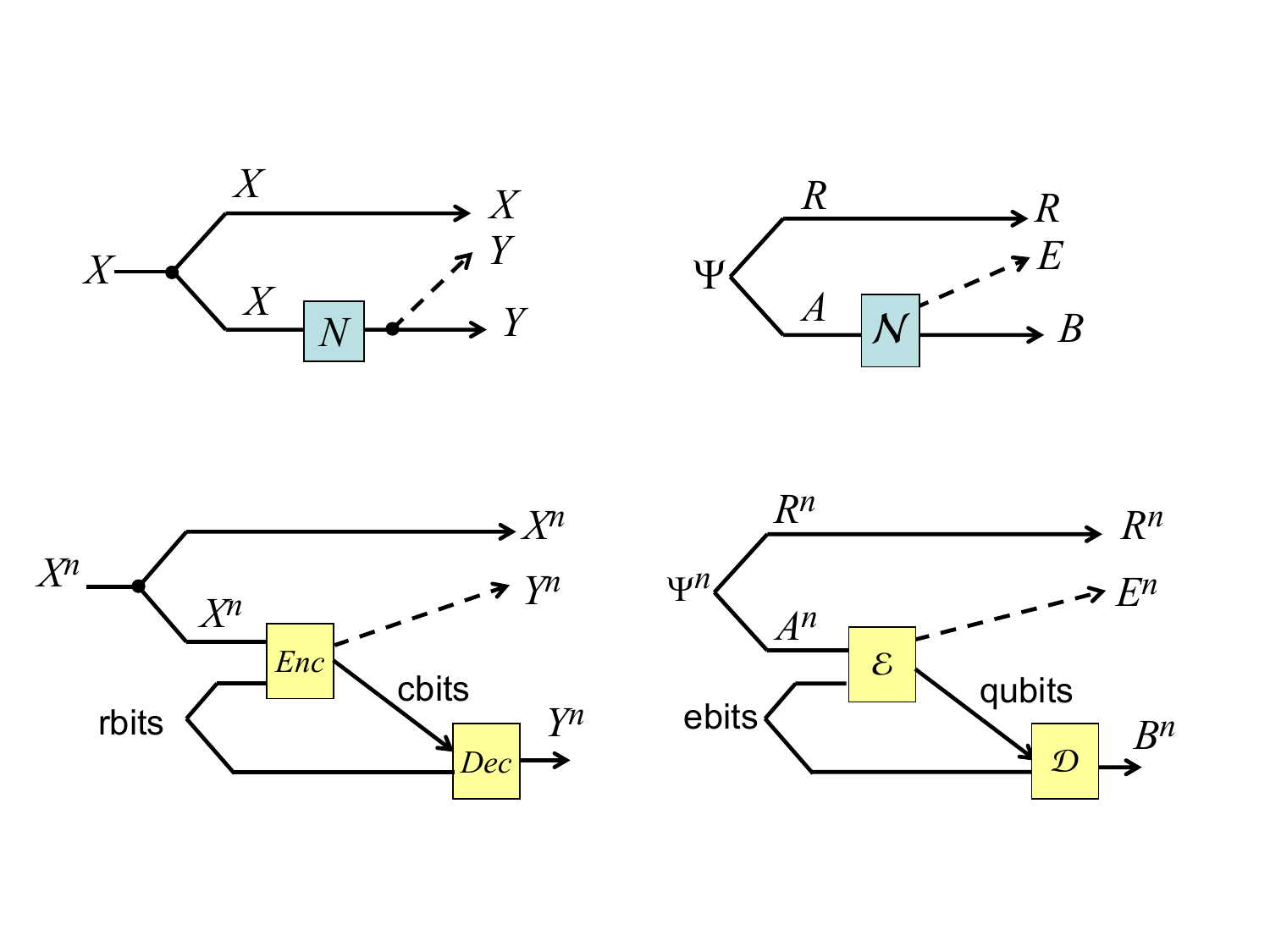}
\caption{Parties and subsystems associated with classical and quantum channels
(top left and right, resp.) and with their simulation using standard resources
(bottom left and right respectively).  The dashed lines represent
  systems that are sent to Alice only in the case of feedback
simulations.}\label{fig:Personae}
\end{figure}

Figure \ref{fig:CRSTQRSTtable} shows some of the known results on
communications resources required to simulate classical and quantum
channels under various conditions.

\begin{figure*}[htbp]
\newcolumntype{M}[1]{>{\centering\arraybackslash}m{#1}}
\renewcommand{\arraystretch}{1.4}
\centering
%\begin{tabular}{|c|c|c|c|c|c|}
\begin{tabular}{|M{10mm}|M{10mm}|M{25mm}|M{35mm}|M{25mm}|M{35mm}|}
\hline
% header row
\multicolumn{2}{|c|}{\large Kind of Channel} &
\multicolumn{2}{|c|}{ \large Classical} &
\multicolumn{2}{|c|}{{\large Quantum}} \\
% below header row
\cline{3-6} \multicolumn{2}{|c|}{\large Kind of Simulation} &
 Classical Feedback &  Non-feedback &  Coherent Feedback &  Non-feedback \\ \hline
% lots of entanglement
\multirow{2}{*}{\parbox[c][][c]{8mm}{\centering Excess shared ebits or rbits}} &
% high-entanglement, tensor power source
\parbox[c][15mm][c]{6ex}{ Tensor-power source} &
\multicolumn{2}{|@{}c@{}|}{
\parbox[c]{40mm}{\centering{$c=I(X;Y)$\\ when $r\geq H(Y|X)$}}} &
\multicolumn{2}{|c|}{\parbox[c]{40mm}{
~~~$q=I(R;B)/2$\\ when $ e\geq I(E;B)/2$
}} \\
\cline{2-6}
% high-entanglement, general source
& \parbox[c][15mm][c]{10mm}{ General source} &
\multicolumn{2}{|c|}{%\parbox[t]{1in}{
$c=C(N) = \max_p I(X;Y)$
} &
\multicolumn{2}{|c|}{\cellcolor{green!25}
\parbox[c]{50mm}{\centering $q=Q_E(\cN)=\max_\rho I(R;B)/2$\\
Ordinary ebits insufficient}}
\\ \hline
% limited entanglement
\multirow{2}{*}{\parbox[c][][c]{10mm}{\centering Limited shared ebits or rbits}} &
% limited entanglement, tensor power source
\parbox[c][15mm][c]{6ex}{ Tensor-power or IID source} &
%\parbox[c]{40mm}{\centering{
$c(r)=\max\{I(X;Y)$, $H(Y)-r\}$ &
$c(r) = \min \{ \max($ $I(X;W),$ $I(XY;W)-r)$ $: W \text{s.t. } \markov{X}{W}{Y}\}$.
&
$q(e)=\max\{\frac{1}{2}I(R;B)$, $H(B)-e\}$ &
$q(e)=\lim_{n\ra\infty}\max$
$\{ \frac{1}{2n}I(R;E_BB^n),$ $\frac{1}{n}H(E_BB^n)-e\} $\\
\cline{2-6}
% limited entanglement, general source
& \parbox[c][15mm][c]{10mm}{ General source} &
$c(r)=\max_X$ $\max\{I(X;Y)$, $H(Y)-r\}$ &
$c(r) = \max_X \min_W$ $\{\max(I(X;W),$ $I(XY;W)-r)$
: \markov{X}{W}{Y}\}. &
\multicolumn{2}{|c|}{\cellcolor{green!25}
\parbox[c]{50mm}{Various tradeoffs possible (see text)\\
Ordinary ebits insufficient}}
\\ \hline
% no entanglement
\multirow{2}{*}{\parbox[c][][c]{10mm}{\centering No shared ebits or rbits}} &
% no entanglement, tensor power source
\parbox[c][15mm][c]{6ex}{ Tensor-power or IID source} &
%\parbox[c]{40mm}{\centering{
$c= H(Y)$ &
$c= \min \{$ $I(XY;W)$ $: W \text{s.t. } \markov{X}{W}{Y}\}$.
&
$q = H(B)$ $=H(\cN(\rho))$ &
$q=\lim_{n\ra\infty}$ $\min$ $\{\frac{1}{n}H(\omega):$
$\exists \omega, \cN_1, \cN_2$  s.t. $\cN_1(\rho^{\ot n})=\omega$
\& $\cN_2(\omega)=\cN(\rho)^{\ot n}\}$\\
\cline{2-6}
% no entanglement, general source
& \parbox[c][15mm][c]{10mm}{ General source} &
$c =$ $\max_X$ $H(Y)$ &
$c = \max_X \min_W \{$$I(XY;W)$ :
\markov{X}{W}{Y}\}. &
$q=$ $\max_\rho H(B)$ $=\max_\rho H(\cN(\rho))$ &
$q=\max_\rho\lim_{n\ra\infty}$ $\min$ $\{\frac{1}{n}H(\omega):$
$\exists \omega, \cN_1, \cN_2$  s.t. $\cN_1(\rho^{\ot n})=\omega$
\& $\cN_2(\omega)=\cN(\rho)^{\ot n}\}$
\\ \hline
\end{tabular}
\caption{Resource costs of simulating classical and quantum channels:
Some known results on the forward communication cost ($c$=cbits or $q$=qubits) for
simulating classical and quantum channels are tabulated as a function of the kind of source
(tensor power or arbitrary), the kind of simulation (feedback or non-feedback),
and the quantity of shared random bits ($r$) or ebits ($e$) available to assist simulation.
For non tensor power quantum sources (green shaded cells), efficient entanglement-assisted
simulation is not possible in general using ordinary ebits, because of the problem of
entanglement spread.  To obtain an efficient simulation in such cases requires
additional communication (wlog backward classical communication), or a stronger form of
entanglement resource than ordinary ebits, such as an entanglement-embezzling state.}
\label{fig:CRSTQRSTtable}
\end{figure*}

\subsection{Classical Reverse Shannon Theorem} Most of these results
are not new; we collect them here for completeness, and give alternate
proofs that will help prepare for the analogous quantum results.  The
high-shared-randomness and feedback cases below (a,b,e) were proved in
\cite{BSST99, BSST01, Winter:RST}. The low- and zero-shared-randomness
cases (c,d,f) were demonstrated by Cuff~\cite{Cuff08} building on
Wyner's classic common randomness formula~\cite{Wyner75}.  The
connection to rate distortion was first developed in the 1996
Steinberg-Verd\'u paper \cite{SS96}, which also proved
a variant of the high-randomness case.

%Case (e) combines all results.

\begin{theorem}[Classical Reverse Shannon Theorem (CRST)]\label{thm:crst}
Let $N$ be a discrete memoryless classical channel with input
$X$ (a random variable) and induced output $Y=N(X)$.
We will use $I(X;Y)$ to indicate the mutual
information between input and output.    Let $N_F$ denote the feedback version of $N$, which
gives Alice a copy of Bob's output $Y=N(X)$.  Trivially $N \leqlo N_F$ and
$\<N:p\> \leqlo \<N_F:p\>$ for all input distributions $p$.
\bit
\item[(a)] {\em Feedback simulation on known sources with sufficient shared
randomness to minimize communication cost:}
\be \quad\<N_F:p\> \leqlo I(X;Y)[c\ra c]+H(Y|X)[cc].\ee
In fact this is tight up to the trivial reduction $[cc]\leqlo[c\ra c]$.
In other words, for $c$ and $r$ nonnegative,
\be \<N_F:p\>\leqlo c[c\ra c]+r[cc]\ee
iff $c\geq I(X;Y)$ and $c+r \geq H(Y)$.
\item[(b)] {\em Feedback simulation on general sources with sufficient
    shared randomness to minimize communication cost:}
\be  \<N_F\> \leqlo C(N)[c\ra c] + (\max_p H(Y) - C(N))[cc].\ee
\item[(c)] {\em Non-feedback simulation on known sources, with limited shared
randomness:}  When shared randomness is present in abundance, feedback simulation requires
no more communication than ordinary non-feedback simulation, but when only limited shared randomness
is available, the communication cost of non-feedback simulation can be less.
\be \<N : X\> \leqlo c[c\ra c] + r[cc] \ee
if and only if there exists a random variable $W$ with \markov{X}{W}{Y}, such that
$c\geq I(X;W)$ and $c+r\geq I(XY;W)$.
\item[(d)] {\em Non-feedback simulation on known sources with no
    shared randomness:} A special case of case (c)
  is the fact that
\be \<N:p\> \leqlo c[c\ra c]  \ee
if and only if there exists $W$ such that \markov{X}{W}{Y} and
$c\geq I(XY;W)$.
\item[(e)] {\em Feedback simulation on arbitrary sources, with arbitrary shared randomness:}
For non-negative $r$ and $c$,
\be \<N_F\> \leqlo c[c\ra c] + r[cc] \label{eq:crst-f-arb-source}\ee
iff $c \geq C(N) = \max_p I(X;Y)$ and $r \geq \max_p H(Y)-\max_p I(X;Y)$.
Because the two maxima may
be achieved for different $p$ the last condition is not simply $r\geq H(Y|X)$.
\item[(f)] {\em Without feedback} we have, for non-negative $r$ and $c$,
\be \<N\> \leqlo c[c\ra c] + r[cc] \label{eq:crst-nf-arb-source} \ee
if and only if for all $X$ there exists $W$ with \markov{X}{W}{Y}, such that
$c\geq I(X;W)$ and $c+r\geq I(XY;W)$.
\eit
\end{theorem}

Parts (b,e,f) of the theorem reflect the fact that the cost of a
channel simulation depends only on the empirical distribution or type
class of the input\footnote{Types are defined and reviewed in
  \secref{types}.}, which can be communicated in at asymptotically
negligible cost ($O(\log n)$ bits), and that an i.i.d.~source $p$ is
very likely to output a type $p'$ with $\| p - p'\|_1 \sim
1/\sqrt{n}$. Also note that in general the resource reducibility
\eq{crst-f-arb-source} is not a resource equivalence because $H(Y)$
and $I(X;Y)$ may achieve their maxima on different $X$.

Part (c), and the low-randomness simulations in general, are based on
the possibility of splitting the simulation into two stages with the
second performed by Bob, and part of the first stage's randomness
being recycled or derandomized\footnote{Here, ``recycled'' means that
  using a sublinear amount of additional randomness, privacy
  amplification can be used to make the shared randomness
  approximately independent of the output of the first stage.  Our
  proof (in \secref{CRST}) will instead use the somewhat simpler
  ``derandomization'' approach in which we argue that some of
  the random bits in the $X$--$W$ stage can be set in a way that works
  for all input strings $x^n$ simultaneously.}.  Since Alice
does not get to see the output of the second stage, this is a
non-feedback simulation. Indeed, part (a) implies that non-trivial
cbit-rbit tradeoffs are only possible for non-feedback simulations.

\fig{FNF-CRST} and \fig{TwoStageCRST} schematically
illustrate the form of the cbit-rbit tradeoffs.  For feedback simulation
on a fixed source, the tradeoff between communication and shared
randomness is trivial: Beginning at the point $c=I(X;Y), r=H(Y|X)$
on the right, $r$ can only be decreased
by the same amount as $c$ is increased, so that $c=H(Y)$ when $r=0$.
By contrast, if the simulation is not required to provide feedback to
the sender, a generally nontrivial tradeoff results, for which the
amount of communication at $r=0$ is given by Wyner's common
information expression $\min \{I(XY;W): \markov{X}{W}{Y}\}$.
This is evident in \fig{CvsRforCBEC} showing the tradeoff for non-feedback
simulation of the classical binary erasure channel for several values of the
erasure probability $t$.  This figure also shows that for some channels
(in particular for erasure channels with $t>0.5$), even the non-feedback
tradeoff begins with a -45 degree straight line section at low $r$ values.
\begin{figure}[htbp]
\includegraphics[width=3.5in]{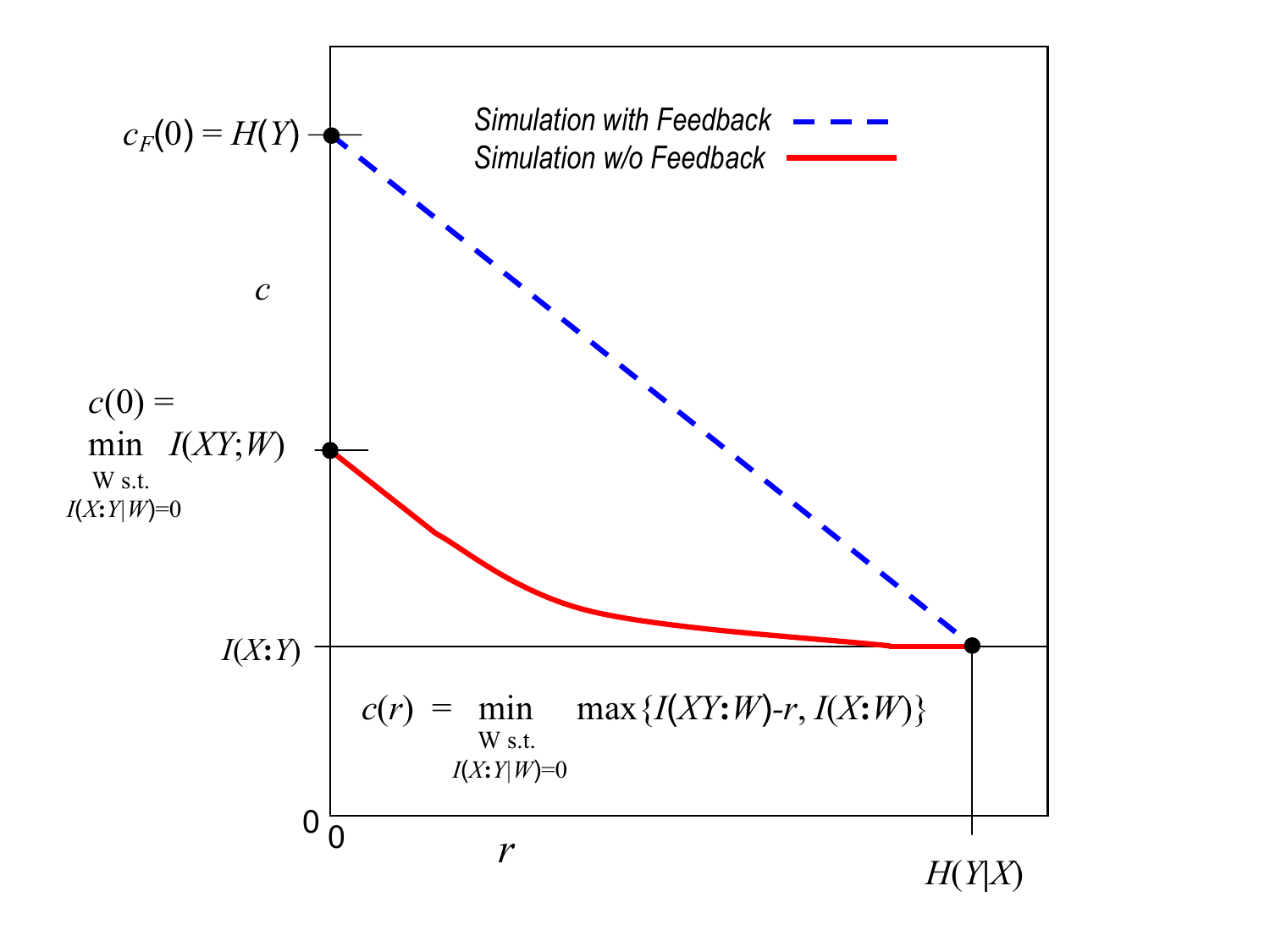}
\caption{Classical communication $c$ versus shared randomness $r$ tradeoff for
feedback and non-feedback simulations of a classical channel on a specified
source $p$ (\thmref{crst}).}
\label{fig:FNF-CRST}
\end{figure}

\begin{figure*}
\includegraphics[width=0.9\textwidth]{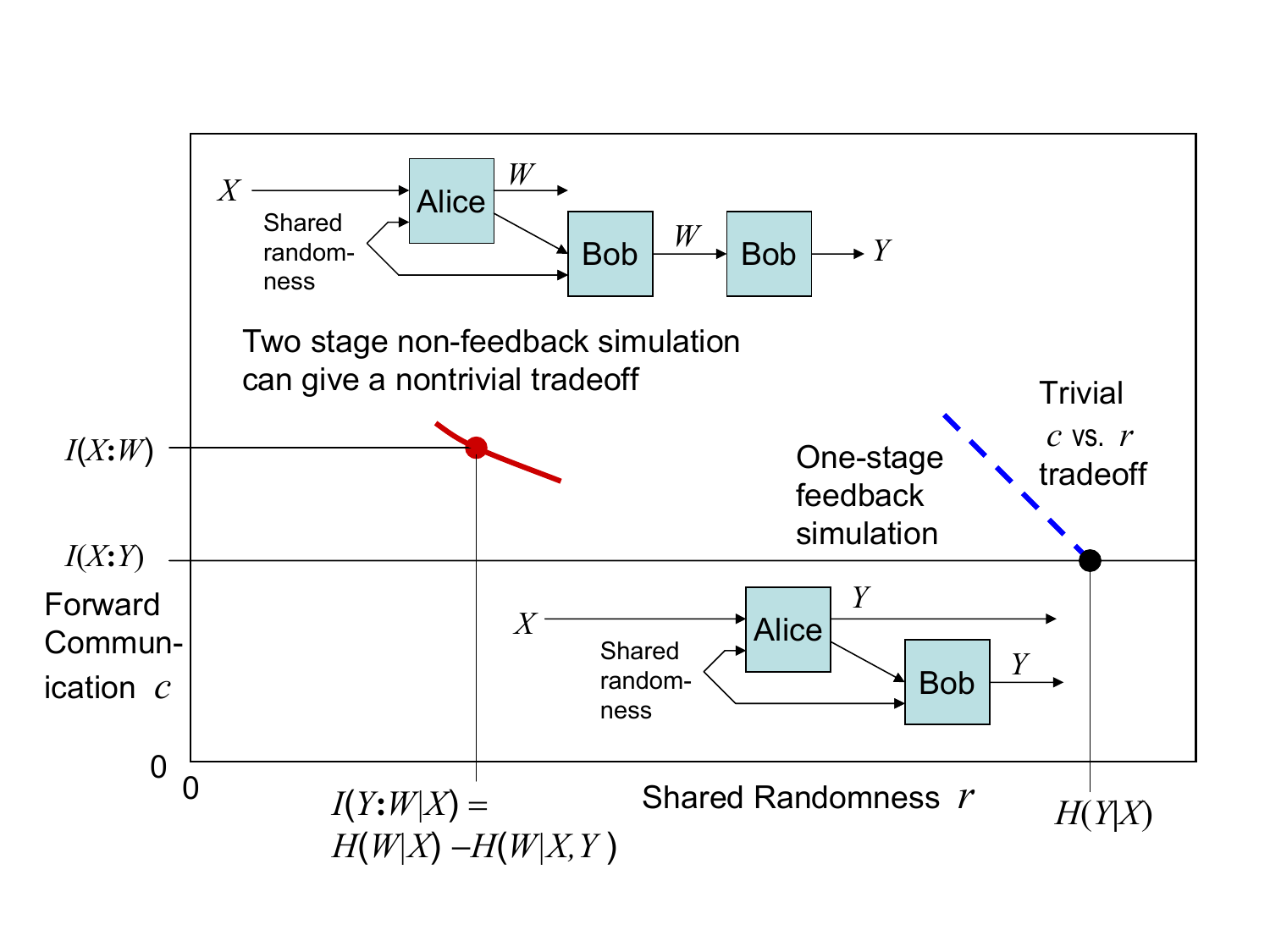}
\caption{Two-stage non-feedback simulation of a classical channel, via a Markov chain
$X \ra W \ra Y$ allows a nontrivial tradeoff between forward communication $c$ and
shared randomness $r$.  A typical point on the optimal tradeoff curve
is shown with $c=I(X:W)$ and $r=I(Y:W|X)$, and with a segment of the
optimal tradeoff curve depicted.
The second term, $H(W|XY)$, in the expression for $r$ represents the portion of the
shared randomness in the first stage simulation of $X\ra W$ that can be recycled
or derandomized.  On the right side is also depicted the ``full
randomness'' solution consisting of a one-stage feedback simulation
that uses communication $I(X:Y)$ and randomness $H(Y|X)$.  Since cbits
can always be traded for rbits, this yields the upper bound depicted
by the 45-degree dashed line coming out of this point.}
\label{fig:TwoStageCRST}
\end{figure*}

\begin{figure}
\includegraphics[width=3.5in]{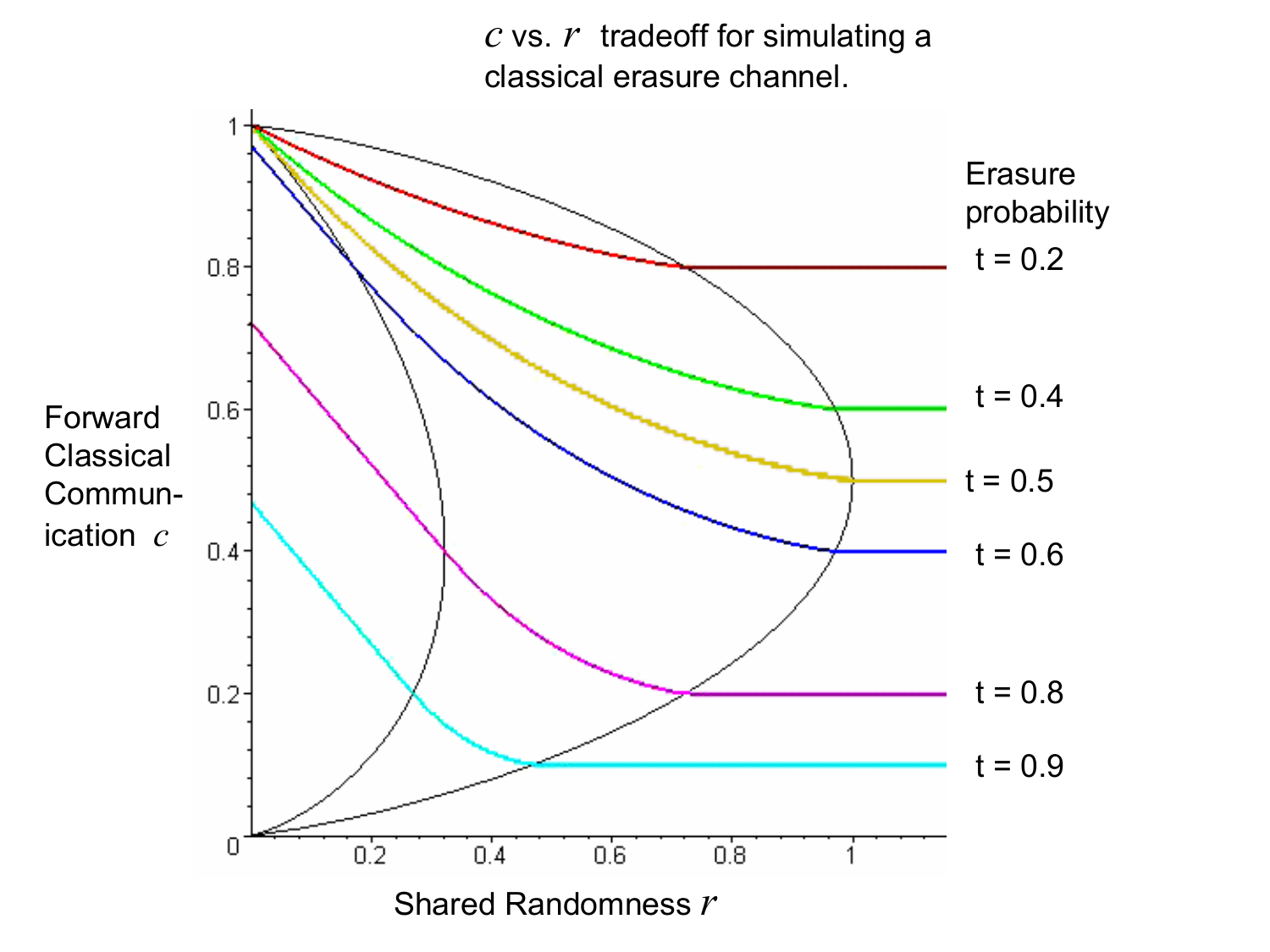}
\caption{Classical communication $c$ vs shared randomness $r$ tradeoff for
non-feedback simulation of binary erasure channels with erasure probabilities
$t = 0.2,\ 0.4,\ 0.5,\ 0.6,\ 0.8$ and $0.9$ (colored graphs). One can show that
in Theorem \ref{thm:crst} part (f) it is enough to consider $W$ such that both
legs $X\rightarrow W$ and $W\rightarrow Y$ are erasure channels. The two black
curves mark the boundaries of the region where the tradeoff has slope $-1$, viz.
$r \leq H_2(c/2)-c$, and where it is horizontal, $r \geq H_2(c)$.  Note that
for $t \leq \frac12$, Wyner's quantity $c(0) = 1$, and that for these
channels the tradeoff graphs have no section of slope $-1$.  These
tradeoff curves were first given in \cite{Cuff08}.}
\label{fig:CvsRforCBEC}
\end{figure}

The converse to (a) follows from Shannon's original noisy channel
coding theorem, which states that $\<N:p\> \geq I(X;Y)_p[c\ra c]$. A
slight refinement~\cite{AhlswedeCsiszar1,AhlswedeCsiszar2} implies
that $\<N_F:p\> \geq
I(X;Y)_p[c\ra c] + H(Y|X)_p[cc]$.

Thus we have the following resource equivalences.
\begin{corollary}
\ba \<N_F:p\> & \eqlo I(X;Y)[c\ra c] + H(Y|X)[cc]
\label{eq:cl-feedback-equiv}\\
\<N_F:p\> + \infty [cc] &\eqlo \<N:p\> + \infty[cc]\non
\\ & \eqlo I(X;Y)[c\ra c] + \infty[cc] \\
\<N_F\> + \infty[cc] &\eqlo \<N\> + \infty[cc] \non
\\ &\eqlo (\max_p I(X;Y))[c\ra c] + \infty[cc]\ea
\end{corollary}

{\em Remark:} The task considered in case (d) above, of simulating a
channel on a known source by forward communication alone without
shared randomness, is a variant of the problem originally
considered by Wyner~\cite{Wyner75}, who sought the minimum rate of a
source allowing two correlated random variables $X$ and $Y$ to be
generated from it by separate decoders.  He called this the common
information between $X$ and $Y$, and showed it was given by
$\min\{I(XY;W):\markov{X}{W}{Y}$\}.

\subsection{Quantum Reverse Shannon Theorem (QRST)}
\label{sec:qrst-statement}

\begin{theorem}[Quantum Reverse Shannon Theorem]\label{thm:qrst}
Let $\cN$ be a quantum channel from $A\ra B$ or equivalently an
isometry from $A\ra BE$ and $\cN_F$ the feedback
channel that results from giving system $E$ to Alice.  If we are given
an input density matrix $\rho^{A}$ then entropic quantities such as
$I(R;B)$ or $I(R;B)_\rho$ refer to the state $\Psi^{RBE}=(I^R \ot \cN^{A\ra
  BE})(\Phi_\rho^{RA})$, where $\Phi_\rho$ is any state satisfying
$\Phi_\rho^{A} = \rho$.
\bit
\item[(a)] {\em Feedback simulation on known tensor power input, with
  sufficient ebits of entanglement to minimize the forward qubit
  communication cost:} \be \forall_\rho\;\<\cN_F:\rho\>\; \eqlo
  \;\half I(R;B)_\rho [q\ra q] + \half I(E;B)_\rho [q q].
\label{eq:feedback-father}
\ee
In view of the trivial tradeoff between ebits and qubits for simulating a
feedback channel, this implies that the qubit communication rate necessary and
sufficient for feedback simulation of a channel on a tensor power source using
ordinary entanglement at the rate $e$ ebits per channel use is
\be q_F(e) = \max \{\half I(R;B),H(B)-e\}.\label{eq:qFvse}\ee
\item[(b)] {\em Known tensor power input, non-feedback simulation, entanglement
possibly insufficient to minimize the forward communication cost:}
\be \<\cN:\rho\> \;\leqlo \;q[q\ra q] + e[qq], \label{eq:low-ent-toff}\ee
if and only if for all $\delta>0$ there exists an $n\!>\!0$ and an
isometry
$V\!:\!{E^n\ra E_AE_B}$ such that
\ba q &\geq \frac{1}{n}\cdot \half I(R^n;B^nE_B)_\Psi-\delta \text{ and}\label{eq:q-low-iid}\\
q+e & \geq \frac{1}{n} H(B^n E_B)_\Psi-\delta \text{ where}\label{eq:qe-low-iid} \\
\ket{\Psi}^{R^nB^nE_AE_B} & := V^{E^n \ra E_A E_B} \cN_F^{\ot n} \ket{\Phi_\rho}^{\ot n}.
\ea
Thus the communication cost for non-feedback simulation on a tensor power source,
as a function of $e$, is given by
\bmu q(e)= \liminf_{n\rightarrow\infty, \exists V:E^n\rightarrow E_A,E_B}\\ \max
\{\half I(R^n; B^nE_B)/n,H(B^nE_B)/n-e\}.\label{eq:qvse}\emu
\item[(c)]{\em Known tensor power input, non-feedback, no entanglement:}
This is obtained from setting $e=0$ in case (b) above.  In this case,
\eq{q-low-iid} is always dominated by \eq{qe-low-iid} and we have that
\be \<\cN:\rho\> \;\leqlo\; q [q\ra q],\ee
iff $q \geq \lim_{n\ra\infty} \frac{1}{n} \min_V H(B^n E_B)$, where
the minimum is over isometries $V\!:\!E^n\ra E_AE_B$.
%; see\fig{TwoStageQRSTregularized}.
The latter is a well-known quantity: it is the regularized
entanglement of purification (EoP)~\cite{purification}
$E_P^\infty(\Psi^{RB}) = \lim_{n\ra\infty} \frac{1}{n}E_P((\Psi^{RB})^{\ot n})$
of the channel's Choi-Jamio\l{}kowski state $\Psi$.
\item[(d)]{\em Arbitrary input, feedback simulation:}
For a communication resource $\alpha$ in the sense of \cite{DHW05}
comprising any combination of
ebits, embezzling states $\Emb$, backward cbits $[c\leftarrow c]$, and/or
forward or backward quantum communication,
\be \alpha \;\geqlo\; \<\cN_F\>\ee
iff there exists a resource $\beta$ such that for all $\rho$,
\be \alpha \;\geqclo\; \<\cN_F:\rho\> + \beta . \ee
Specifically, using embezzling states we have
\be \<\cN_F\> \leqlo Q_E(\cN) [q\ra q] + \Emb
% = C_E(\cN)[c\ra c] + \Emb
\label{eq:emb-qrst}\ee
and when considering back communication
\besp   \<\cN_F\> & \leq
 Q_E(\cN) [q\ra q] + C[c\la c] \\
& \quad + (\max_\rho H(B)_\rho - Q_E(\cN))[qq]
\label{eq:back-qrst}\eesp
iff $C \geq  \max_\rho H(B)_\rho - \min_\rho H(B|R)_\rho -
C_E(\cN)$.
Other examples are discussed in \secref{spread}.
\item[(e)] {\em Arbitrary input, no feedback:}  This case combines
  elements of cases (b) and (d), although we now consider only fully
  coherent input resources.  If $\alpha$ is a combination of ebits,
  embezzling states $\Emb$ and forward and/or backward qubits, then
  $ \alpha \geqlo\<\cN\>$ iff for all $\delta>0$ there exists
  an $n>0$, a resource $\beta_n$ and an isometry $V_n:E^n\ra E_AE_B$
  such that
\be \alpha \geqclo \frac{1}{n}\<V_n\circ \cN_F^{\ot n}\> + \beta_n.
\label{eq:full-non-feedback}\ee
\eit
\end{theorem}

Part (a) of \thmref{qrst} can equivalently be stated as
\be  \<\cN_F:\rho\> \;\eqlo\; I(R;B)_\rho [q \ra qq] + H(B|R)_\rho [qq],
\label{eq:cobit-feedback-equiv}\ee
where $[q\ra qq]$ denotes a co-bit~\cite{Har03,DHW05}, which is
equivalent to $([q\ra q]+[qq])/2$.  The formulation in \eq{cobit-feedback-equiv} is parallel to the classical version in \eq{cl-feedback-equiv} if we replace quantum feedback with classical feedback, co-bits with cbits and ebits with rbits.

A weaker version of (a) was proven in a long unpublished and now
obsolete version of the present paper.  The idea there was to simulate
the channel using a noisy form of teleportation, and then to use
measurement compression~\cite{Winter:POVM}\footnote{More concretely,
  suppose that Alice uses the ``Homer Simpson protocol,'' which means
  applying $\cN$ to her input and then
  teleporting the output to Bob, using a classical message of size
  $2\log d_A$.  Alice's entire part of the protocol can be viewed as
  a measurement that she performs on her input state and on half of a
  maximally entangled state.  The mutual information between her
  classical message and Bob's residual quantum state is given by
  $I(R;B)$.  Therefore \cite{Winter:POVM} can be used to simulate
  $n$ applications of this measurement by a block measurement with
  $\approx\exp(n I(R;B))$ outcomes.  Finally, it is necessary to
observe that the error analysis in \cite{Winter:POVM} shows that the
simulated measurement not only has the correct output statistics, but
essentially has the correct Kraus operators.  Thus the compressed
measurement gives a high-fidelity simulation of the Homer
Simpson protocol, and thus of the original channel.  However, the
measurement compression step relies on knowledge of the input density
matrix $\rho$, and so new ideas are necessary for the non-tensor-power case.
}.  The full statement of (a) has
since been proved by Devetak~\cite{devetak-triangle} using his triangle of
dualities among protocols in the ``family tree'' -- see also~\cite{DHW05};
by Horodecki \emph{et al.}~\cite{HOW05} as the inverse of the
``mother'' protocol, a coherent version of state merging;
and by Abeyesinghe \emph{et al.}~\cite{ADHW06} in the context
of a direct derivation of the ``mother'' protocol.
We will present another proof of (a) in \secref{qrst}, partly in order
to prepare for the proof of the rest of \thmref{qrst}.

To prove (b), we argue that any protocol using only qubits and ebits
for a non-feedback simulation of $\cN^{\ot n}$ is equivalent to one
that performs a feedback simulation of $V^{E^n\ra E_AE_B} \circ
\cN_F^{\ot n}$.  The argument is that the resources used (qubits and
ebits) leak nothing to the environment, so the only non-unitary
elements are those that are deliberately introduced by Alice and Bob.
Thus, we can replace any non-unitary operation by an isometry that
instead sends the system to be discarded to a local ``environment'',
labeled $E_A$ for Alice and $E_B$ for Bob.  By Uhlmann's theorem and
the fact that any two purifications are related by an isometry, it
follows that if our original simulation had fidelity $1-\eps$ with the
action of $\cN^{\ot n}$, then this modified simulation has fidelity
$1-\eps$ with $V^{E^n\ra E_AE_B}\circ \cN_F^{\ot n}$ for some isometry
$V$.  This is an equivalence, since this procedure turns any
simulation of $\cN^{\ot n}$ into a method of simulating $V^{E^n\ra
  E_AE_B}\circ \cN_F^{\ot n}$ for an isometry $V$, and the reverse
direction is achieved simply by discarding the $E_A, E_B$ systems.

Part (c) is simply a special case of (b), and was proven in the case
when $\cN$ is a CQ channel (that is, has classical inputs) by
Hayashi~\cite{Hayashi:EoP}.  It corresponds to the regularized
entanglement of purification~\cite{purification} of $\ket\Psi$.  In
both cases, the additivity problem (i.e. the question of whether
regularization is necessary) is open, although recent evidence
suggests strongly that the entanglement of purification is {\em not}
additive~\cite{EoP-doubt} and thus that it is not a single-letter
formula for the simulation cost.

Proving, and indeed understanding, parts (d) and (e) will require the
concept of entanglement spread, which we will introduce in
\secref{spread}.  At first glance, the statements of the theorem may
appear unsatisfying in that they reduce the question of whether
$\<\cN\>\leq \alpha$ or $\<\cN_F\>\leq \alpha$ to the question of
whether certain other clean resource reductions hold.  However,
according to part (a) of \thmref{qrst}, the corresponding clean
resource reductions involve the standard resources of qubits and
ebits.  As we will explain further in \secref{spread}, this will allow
us to quickly derive statements such as \eq{emb-qrst} and
\eq{back-qrst}.  An alternate proof\footnote{This proof was developed
 in parallel with ours (cf discussion in \cite{CKR09}) and differs
 primarily by describing merging in
  terms of one-shot entropies (compared with our applying merging only
  to ``flat'' spectra) and by reducing to the tensor-power case using
  the post-selection principle of \cite{CKR09} (compared with
  our use of Schur duality to reduce to the flat case).  Note that the
post-selection principle can also be thought of in terms of the Schur
basis as the statement that tensor-power states are ``almost flat'' in
a certain sense (cf. \cite{Hayashi-post}).} of the QRST for general sources using
embezzling states as the entanglement resource,
Eq.~(\ref{eq:emb-qrst}), was given by Berta, Christandl, and
Renner~\cite{BCR11}.

The situation in part (d) when embezzling states are not present
(i.e. general input, unlimited ebits, and some combination of
forward quantum communication and backwards quantum and classical
communication) is somewhat surprising in that the simulation requires an
asymptotically greater rate of communication than the communication
capacity of the channel.  To capture this gap, we introduce the
following definition.
\begin{definition}\label{def:deficit}
The {\em spread deficit} of a channel $\cN$ is defined as
\be \deficit(\cN) := \max_\rho H(B)_\rho - \min_\sigma H(B|R)_\sigma - C_E(\cN).
\label{eq:deficit}\ee
\end{definition}
Thus, we could equivalently say that the resource inequality in
\eq{back-qrst} holds iff $C\geq \deficit(\cN)$.

It is important to note that the maximization of $H(B)$ and the
minimization of $H(B|R)$ on the RHS of \eq{deficit} are taken
separately.  Indeed, $C_E(\cN)$ is simply the maximization of
$H(B)_\rho-H(B|R)_\rho$ over all $\rho$, so \eq{deficit}
expresses how much larger this expression can be by breaking up the
optimization of those two terms.

Fortunately, each term in the RHS of \eq{deficit} is additive, so there is no
need to take the limit over many channel uses.  The additivity of
$H(B)_\rho$ follows immediately from the subadditivity of the von
Neumann entropy, or equivalently the nonnegativity of the quantum
mutual information.  The other two terms have already been proven to
be additive in previous work: \cite{DJKR06} showed that $\min H(B|R) =
-\max H(B|E)$ is additive and \cite{AC97} showed that $C_E$ is additive.
Thus, we again obtain a single-letter formula in the case of unlimited
ebits.

The fact that $\deficit(\cN)$ provides a single-letter
characterization involving convex optimizations makes it possible to
explicitly and efficiently evaluate it.    For the important class of
so-called covariant channels, such
as the depolarizing and erasure channels,
entropic quantities are invariant under unitary rotation of the
inputs.  In this case, $H(B)$, $H(R)-H(E)$ and $I(R;B)$ are all
simultaneously maximized for the maximally-mixed input and
$\deficit(\cN)=0$.  However, channels that lack this symmetry will
generally have nonzero $\deficit$.  As an example, we plot the
entanglement-assisted capacity $C_E(\cN)$ against the ebit-assisted
simulation cost $C_E(\cN)+\deficit(\cN)$ for the amplitude-damping
channel in \fig{ampl}.

For the amplitude damping channel, the spread deficit is comparable
to the other costs of the simulation.  But there exist other
channels for which the spread deficit can dominate the cost of the
channel simulation.  Here is an example: for any $d$, define the
``variable-entropy'' channel $\cM_d$ mapping 2 dimensions to $d+1$ dimensions as follows:
it measures the input, and upon outcome 0, outputs $\proj 0$, and
upon outcome 1, outputs $\frac{1}{d}\sum_{i=1}^d \proj i$.  For any
$d$, $C_E(\cM_d)=1$, but $\deficit(\cM_d)=\log(d+1)-1$, which is
asymptotically larger as $d$ grows\footnote{Proof: $\cM_d$ can be
perfectly simulated with a single bit of forward classical communication,
  which proves that $C_E(\cM_d)\leq 1$, while the obvious protocol for
sending one classical bit through the channel proves that
$C_E(\cM_d)\geq 1$.   To evaluate $\deficit(\cM_d)$, observe that
$H(B)$ achieves its maximal value of $\log(d+1)$ upon input
$\frac{1}{d+1}\proj 0 + \frac{d}{d+1}\proj 1$, while $H(B|R)$ can be
zero if the input $\proj 0$ is given.}.  Thus, when performing a
feedback simulation of $\cM_d$ using cbits and ebits, nearly all of
the communication cost comes from the need to create entanglement
spread (discussed further in \secref{spread}).

What about non-feedback simulations?  In this case, it turns out that
the variable-entropy and amplitude-damping channels can both be
simulated at the communication rate given by $C_E$.  We will discuss
this in more detail in
\secref{spread}, but the intuitive reason for this is that
non-feedback simulations allow us to damage the environment of the
channel and in particular to measure it.  This can result in
collapsing superpositions between different amounts of entanglement,
thus reducing the contribution of entanglement spread.  However, there remain
channels whose simulation cost with ebits is higher than with
embezzling states or back communication even for non-feedback
simulation; we describe an example (the ``Clueless Eve'' channel) in
\secref{clueless}.

\begin{figure}[htbp]
\includegraphics[width=3.5in]{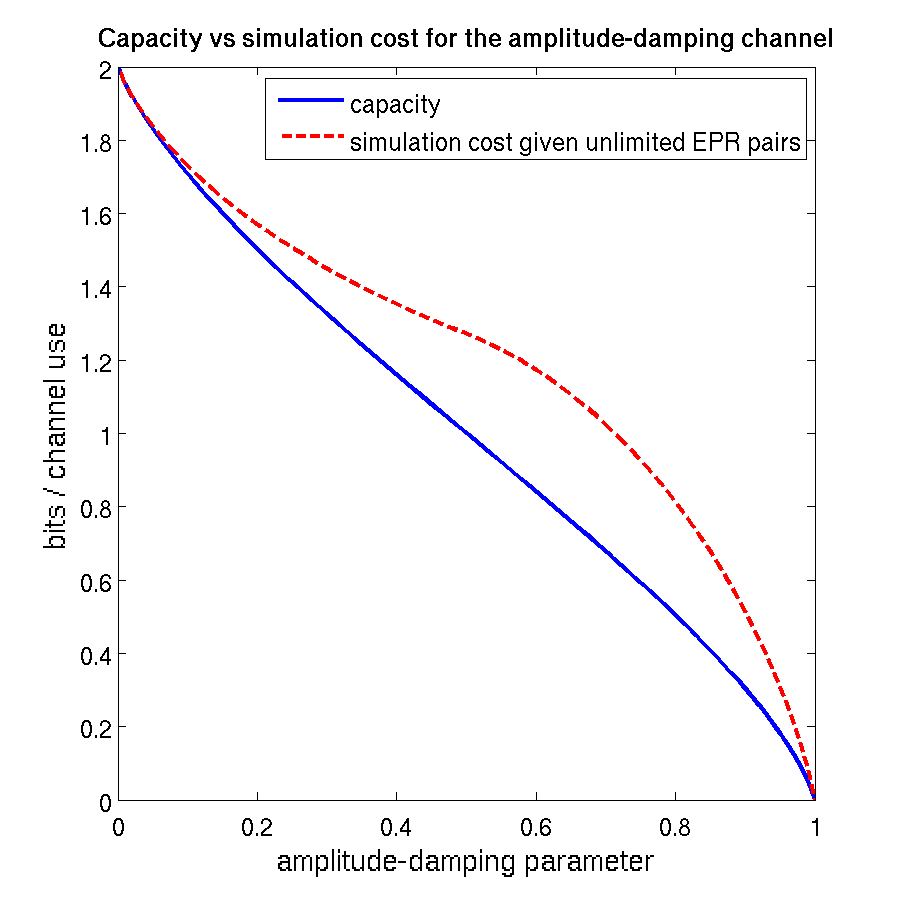}
\caption{The amplitude damping channel with parameter $\gamma$ has
  Kraus operators $\proj 0 + \sqrt{1-\gamma}\proj 1$ and $\sqrt\gamma
  \ket 0 \bra 1$.  The lower, solid, curve is the
  entanglement-assisted classical capacity of the amplitude-damping
  channel, or equivalently the (w.l.o.g. feedback) simulation cost in cbits when back
  communication or embezzling states are given, or when the source is
 a tensor power.  The upper, dashed, curve is the feedback simulation cost in cbits
 (calculated using \eq{deficit}) when instead unlimited ebits
 are given.  The gap between the two curves is the spread deficit
 from \defref{deficit}, and illustrates the extra
 communication cost of producing entanglement spread.}
\label{fig:ampl}
\end{figure}

The proofs of parts (d) and (e) will be given in \secref{qrst}.  To
prove them, we restrict attention to the case when $\alpha$ is a
combination of entanglement-embezzling states and/or ``standard''
resources (qubits, cbits and ebits).  However, for part (e), we need
to further restrict our claim to exclude cbits, for reasons related to
the fact that we do not know the tradeoff curve between quantum and
classical communication when simulating classical channels.

{\em Remark:} Analogously to the low-shared randomness regime in
classical channel simulation (Figure \ref{fig:TwoStageCRST} and cases
(c) and (d) of the CRST), simulating a non-feedback channel permits a
nontrivial tradeoff between ebits and qubits, in contrast to the
trivial tradeoff for feedback simulation.
While the cbit-rbit tradeoff curve for simulating classical channels is
additive and given by a single-letter formula~\cite{Wyner75,Cuff08},
no such formula or additivity result is known for the qubit cost in
the zero- and low-entanglement regime.
\begin{figure*}[htbp]
\includegraphics[width=0.9\textwidth]{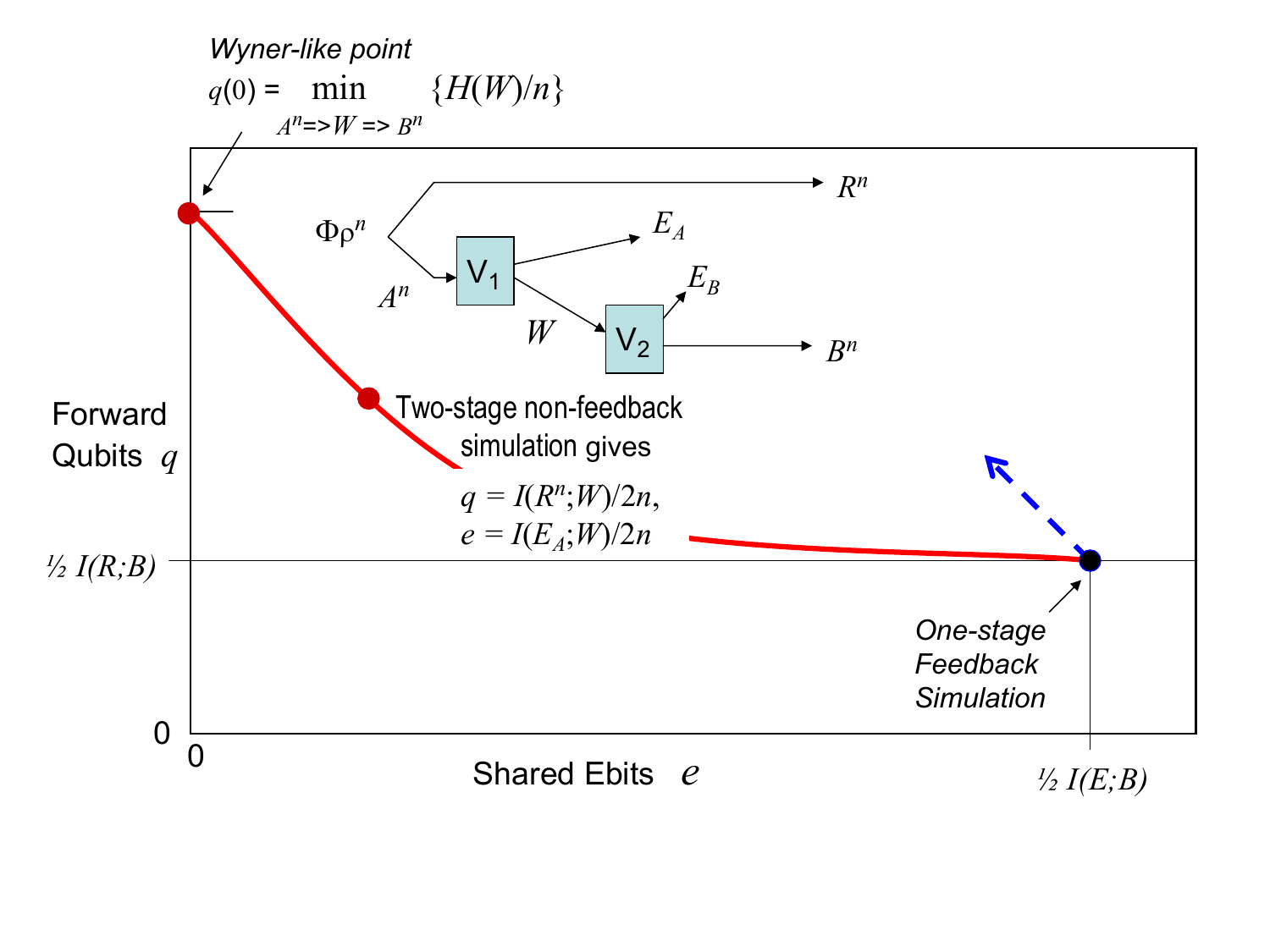}
\caption{Two-stage non-feedback simulation of a quantum channel (solid
  red curve) on a specified input $\rho$, via an intermediate state
  $W$, makes possible a nontrivial tradeoff between forward
  communication $q$ and shared entanglement $e$. By contrast, for a
  feedback simulation (right, dashed blue curve) only a trivial
  tradeoff is possible, where any deficit in ebits below the
  $\frac{1}{2}I(E;B)$ needed for optimal simulation must be compensated
  by an equal increase in the number of qubits used.}
\label{fig:TwoStageQRST}
\end{figure*}

{\em Remark:} Interestingly, quantum communication or entanglement can
sometimes improve simulations of even classical channels.  In
\cite{Winter-triples} an example of a classical channel is given with
$d$-dimensional inputs which requires $\Omega(\log d)$ classical bits
to simulate, but can be simulated quantumly using $O(d^{-1/3})$ qubits
of communication, asymptotically. Curiously, the classical reverse
Shannon theorem (\thmref{crst}) is only a special case of the quantum
reverse Shannon theorem (\thmref{qrst}) when in the unlimited shared
entanglement regime; one of the problems left open by this work is to
understand how entanglement can be more efficient than shared
randomness in creating correlated classical probability distributions.
More generally, which values of $c,q,r,e$ are consistent with the
reducibility $\<\cN\> \leqlo c[c\ra c] + q[q\ra q] + r[cc] + e[qq]$?
We know how to convert this problem to the equivalent relative
resource problem with $\<\cN\>$ replaced with $\<\cN:\rho\>$, but this
in turn we do not have an answer for.

{\em Remark:}
Our results imply unbounded gaps (for growing dimension) between the
costs of simulating channels when (a) no entanglement is given, (b) a linear or unlimited
rate of ebits are given, and (c) stronger forms of entanglement, such as
embezzling states, are given.  An example of a large gap between (a) and (b)
is given by the Werner-Holevo channel~\cite{WH02}, defined on $d$-dimensional inputs
to be $\cN(\rho) = ((\tr \rho)I - \rho^T)/(d-1)$.  This channel has
$C_E(\cN)\lessapprox 1$, but when acting on half of a maximally entangled state
produces a state with entanglement of purification equal to $\log d$~\cite{CW05}.
Thus, the gap between the ebit-assisted simulation cost and the unassisted simulation
cost grows with dimension.  For an asymptotically growing gap between (b) and (c),
we give an example in \secref{clueless}.

\subsection{Entanglement spread}\label{sec:spread}
To understand parts (d) and (e) of \thmref{qrst}, we need to introduce
the idea of entanglement spread.  This concept is further explored in
\cite{Har-spread, HW02}, but we review some of the key ideas here.

If Alice's input is known to be of i.i.d.~form
$\rho^{\ot n}$ then we know that the channel
simulation can be done using $\frac{1}{2}I(R;B)[q\ra
q]+\frac{1}{2}I(B;E)[qq]$.  To see the complications that arise from a
general input, it suffices to consider the case when Alice's input is
of the form $(\rho_1^{\ot n} + \rho_2^{\ot n})/2$.  We omit
explicitly describing the reference system, but assume that Alice's
input is always purified by some reference and that the fidelity of any
simulation is with respect to this purification.

Assume that $\rho_1^{\ot n}$ and $\rho_2^{\ot n}$ are nearly perfectly
distinguishable and that the channel simulation should not break the
coherence between these two states.  Naively, we might imagine that
Alice could first determine whether she holds $\rho_1^{\ot n}$ or
$\rho_2^{\ot n}$ and coherently store this in a register
$i\in\{1,2\}$.  Next she could conditionally perform the protocol for
i.i.d.~inputs that uses $\frac{1}{2}I(R;B)_{\rho_i}[q\ra q] +
\frac{1}{2}I(B;E)_{\rho_i}[q q]$.  To use a variable amount of
communication, it suffices to be given the resource $\max_i
\frac{1}{2}I(A;B)_{\rho_i}[q\ra q]$, and to send $\ket{0}$ states
when we have excess channel uses.  But unwanted entanglement cannot in general
be thrown away so easily.  Suppose that
$I(B;E)_{\rho_1}>I(B;E)_{\rho_2}$, so that simulating the channel on
$\rho_1^{\ot n}$ requires a higher rate of entanglement consumption than
$\rho_2^{\ot n}$.  Then it is not possible to start with
$\half nI(B;E)_{\rho_1}$ (or indeed any number) of ebits and perform
local operations to obtain a superposition of $\half nI(B;E)_{\rho_1}$ ebits and $\half nI(B;E)_{\rho_2}$ pairs.

The general task we need to accomplish is to coherently create a superposition
of different amounts of entanglement.  Often it is convenient to think
about such superpositions as containing a small ``control'' register
that describe how many ebits are in the rest of the state.  For example, consider the state
\be
\ket{\psi} = \sum_{i=1}^m \sqrt{p_i} \ket{i}^A\ket{i}^B
\ket{\Phi}^{\ot n_i} \ket{00}^{\ot N-n_i},
\label{eq:spread-state}\ee
where $0\leq n_i \leq N$ for each $i$.
Crudely speaking\footnote{
  This neglects the entanglement in the $\ket{ii}$ register.  However,
  in typical applications, this will be logarithmic in the total
  amount of entanglement.}, we say that $\max_i n_i - \min_i n_i$ is
the amount of entanglement spread in the state $\ket{\psi}$, where the
$\max$ and $\min$ are taken over values of $i$ for which $p_i$ is
nonnegligible.
%A simple version of this task that is sufficient for our purposes is to map $\ket{\phi}= \sum_{i=1}^m \sqrt{p_i} \ket{i}^A\ket{i}^B$ to $\ket{\psi}$.

A more precise and general way to define entanglement spread for any
bipartite state $\ket{\psi}$ is (following \cite{HW02}) as
$\Delta(\psi^A) = H_0(\psi^A) - H_\infty(\psi^A)$, where $H_0(\rho) =
\log\rank \rho$ and $H_\infty(\rho) = - \log \|\rho\|_\infty$.  (The
quantities $H_0$ and $H_\infty$ are also known as $H_{\max}$ and
$H_{\min}$ respectively.  Alternatively, they can be interpreted as
R\'enyi entropies.)  Ref.~\cite{HW02} also defined an $\eps$-smoothed
version of entanglement spread by
$$\Delta_\eps(\rho) = \min\{\Delta(\sigma) : 0\leq \sigma \leq \rho,
\tr \sigma \geq 1-\eps\}$$
that reflects the communication cost of approximately preparing $\ket{\psi}$.  More precisely, we have
\begin{theorem}[Theorem 8 of \cite{HW02}]\label{thm:spread}
If $\ket{\psi}$ can be created from ebits using $C$ cbits of communication and error $\leq \eps = \delta^8/4$, then
\be C \geq \Delta_\delta(\psi^A)  + 3\log(1-\delta)
\label{eq:spread-bound}\ee
\end{theorem}
The factor of 3 in \eq{spread-bound} is because the definition of $\Delta_\eps$ we have used is actually the alternate version used in Remark 4 of \cite{HW02}.  We can similarly define $H_{0,\eps}(\rho) := \log\min \{\rank \sigma: 0\leq\sigma\leq \rho, \tr\sigma\geq 1-\eps\}$ and
$H_{\infty,\eps}(\rho) := -\log\min \{\|\sigma\|_\infty:
0\leq\sigma\leq \rho, \tr\sigma\geq 1-\eps\}$.  Our definition of
$H_{0,\eps}$ is the same as the one used in \cite{HW02}, but our
definition of $H_{\infty,\eps}$ may be as much as $-\log(1-\eps)$ smaller.  As with the definitions in \cite{HW02}, our quantities trivially satisfy
\be \Delta_\eps(\rho) \geq H_{0,\eps}(\rho) - H_{\infty,\eps}(\rho)
.\label{eq:spread-vs-S}\ee
Our quantities can also be expressed as
\begin{subequations}\label{eq:S-opt}
\ba H_{0,\eps}(\rho) &= \min_M H_0(\sqrt{M}\rho\sqrt{M}) \\
H_{\infty,\eps}(\rho) &= \max_M H_\infty(\sqrt{M}\rho\sqrt{M}) ,
\ea
\end{subequations}
where in each case $M$ must satisfy $0\leq M\leq I$ and $\tr M\rho
\geq 1-\eps$.  In fact, we can WLOG assume that $M$ commutes with
$\rho$ and $\tr M\rho=1-\eps$.  Similarly, in the definitions that
optimized over $\sigma\leq\rho$, we can assume that $\rho$ and
$\sigma$ commute.

One advantage of this version of $\Delta_\eps(\rho)$ is that it has the following natural interpretation as a minimization over nearby {\em normalized} states.
\begin{lemma}\label{lem:spread-normalized}
\begin{multline} \max(0,\Delta_\eps(\rho))
\\= \min\{\Delta_0(\sigma) : \frac{1}{2}\|\rho-\sigma\|_1\leq \eps, 0\leq \sigma, \tr\sigma=1\}
\end{multline}
\end{lemma}
The lemma is proved in the appendix.  It improves upon Lemma 5 of \cite{HW02}, and could thus be used to tighten \thmref{spread}, although we do not carry out that exercise here.

There are a few different ways of producing entanglement spread, which are
summarized in \cite{Har-spread}.
%In [???], a nearly tight converse to the lower bound of \cite{HW02}
%was proven, showing that $\Delta(\psi^A) + O(1)$ cbits and
%$H_0(\psi^A)$ ebits suffice to produce $\ket{\psi}$.
For example, one cbit can be used to coherently eliminate one ebit, or
to do nothing; and since both of these tasks can be run in
superposition, this can also be used to create entanglement spread.
Likewise one qubit can coherently either create or disentangle one
ebit.  To put this on a formal footing, we use the {\em clean resource
reducibility} $\leqclo$ (called $\leqclean$ in \cite{Har-spread}).  A
resource $\beta$ is said to be ``cleanly LO-reducible'' to $\alpha$
iff there is an asymptotically faithful clean transformation from
$\alpha$ to $\beta$ via local operations: that is, for any
$\eps,\delta>0$ and for all sufficiently large $n$, $n(1+\delta)$
copies of $\alpha$ can be transformed by local operations into $n$
copies of $\beta$ with overall diamond-norm error $\leq \eps$, and
moreover, any quantum subsystem discarded during the transformation is
in a standard $\ket{0}$ state, up to an error vanishing in the limit
of large $n$.  In particular, entangled states cannot be discarded.
This restriction on discarding states means that clean protocols can
be safely run in superposition.

Finally, we can define the clean entanglement capacity of a resource
$\alpha$ to be the set $E_{\text{clean}}(\alpha) = \{E
: \alpha\geqclo E [qq]\} \subseteq \bbR$.  Negative values of $E$
correspond to the ability to coherently eliminate entanglement.
By time-sharing, we see that
$E_{\text{clean}}(\alpha)$ is a convex set.  However, it will
typically be bounded both from above and below, reflecting the fact
that coherently undoing entanglement is a nonlocal task.  The clean
entanglement capacities of the basic resources are
\begin{subequations}\ba E_{\text{clean}}([q\ra q]) =
E_{\text{clean}}([q\la q]) & = [-1, 1] \label{eq:q-spread-cap}\\
E_{\text{clean}}([c\ra c]) =
E_{\text{clean}}([c\la c]) & = [-1, 0] \label{eq:c-spread-cap}\\
E_{\text{clean}}([qq]) & = \{1\}
\ea \label{eq:spread-caps}\end{subequations}
To understand \eq{q-spread-cap}, observe that
transmitting one qubit can map $\ket{\Phi_2}^{AB}$ to or from
$\ket{00}^{AB}$, in each case without changing anything in the
environment.  The reasoning behind \eq{c-spread-cap} is less obvious
since sending any classical message leaks the message to the
environment by definition.  However, the protocol can be made clean by
always sending a uniformly random bit through the channel.  If Alice
generates this bit locally (or more simply sends a $\ket +$ state) and
Bob discards it, then this does not
change their amount of shared entanglement.  Alternatively, if Alice
sends her half of an ebit through the classical channel and Bob
performs a CNOT with the transmitted bit as control and his half of
the ebit as target, then this will eliminate the entanglement
while presenting the same information to the environment.

These resources can be combined in various ways to create entanglement
spread.  For example, to create a superposition of 5 and 9 ebits,
we might start with 8 ebits, use two cbits to create a superposition
of 6 and 8 ebits and then use one qubit to create the desired
superposition of 5 and 9 ebits.  Implicit in these sort of protocols
is a pair of control registers $A_C, B_C$ that specify how many ebits Alice
and Bob would like to end up with. In this example, they would like to map
the state $(c_0 \ket{00} + c_1\ket{11})^{A_CB_C} \ot \ket{\Phi_2}^{\ot 8}$
(for arbitrary coefficients $c_0,c_1$) to
\be c_0 \ket{00}^{A_CB_C} \ket{\Phi_2}^{\ot 5}\ket{00}^{\ot 4} +
c_1\ket{11}^{A_CB_C}
\ket{\Phi_2}^{\ot 9}.
\label{eq:spread-example}\ee
To achieve this transformation, Alice and Bob perform the following
sequence of actions conditioned on their control qubits:
\bit
\item If Alice's control qubit is zero, she sends two of her entangled
  qubits through the classical channel.  If it is one, she sends two
  $\ket +$ states through the channel.  Either way, the environment
  observes two random bits sent through the channel.
\item If Bob's control qubit is zero, he uses the two bits he has
  received as controls for CNOTs that are applied to his halves of the
  ebits that Alice sent, thus mapping them to $\ket{00}$.  Then he
  discards the bits he received.  If Bob's control bit is one, he
  simply discards the bit he received.  Either way the environment
  sees another copy of the same random bit being discarded by Bob.
  Moreover, this bit is now independent of the residual quantum state
  held by Alice and Bob.  Alice and Bob now share
$$c_0 \ket{00}\ket{\Phi_2}^{\ot 6}\ket{00}^{\ot 2} + c_1\ket{11}\ot
\ket{\Phi_2}^{\ot 8}.$$
\item If Alice's control qubit is zero, she sends half of one of her
  ebits through the qubit channel and locally creates a $\ket 0$
  state.  If her control qubit is one, she locally creates a
  $\ket{\Phi_2}$ and sends half through the channel.
\item If Bob's control qubit is zero, he now holds both halves of one
  of the $\ket{\Phi_2}$ states.  He rotates this to a $\ket{00}$
  state and discards one of the $\ket 0$ qubits.  If his control qubit
  is one, he keeps the transmitted qubit, but also creates and
  discards a $\ket 0$ qubit.  Alice and Bob are now left with the
  state in \eq{spread-example}.
\eit
Observe that in this example the classical and quantum communication
could have been sent in either direction.  Thus, while some parts of
the simulation protocol can only use forward communication, the spread
requirements can be met with communication in either direction.

While this framework gives us a fairly clear understanding of the
communication resources required to create entanglement spread, it
also shows how unlimited ebits are not a good model of unlimited
entanglement.  Instead of maximally entangled states, we will use the so-called
entanglement-embezzling~\cite{vDH03} states
$\ket{\varphi_N}^{AB}$, which are parameterized by their Schmidt rank
$N$, and can be used catalytically to produce or destroy any Schmidt
rank $k$ state up to an error of $\frac{\log k}{\log N}$ in the trace norm.  See
\cite{vDH03} for a definition of $\ket{\varphi_N}$ and a proof of
their entanglement-embezzling abilities.  We let the resource $\Emb$
denote access to an embezzling state of arbitrary size: formally,
$\Emb = \bigcup_{N\geq 1} \<\varphi_N\>$ and so we have
$$E_{\text{clean}}(\Emb) = (-\infty,\infty).$$
  By the above
discussion, this is strictly stronger than the resource $\infty[qq]$.

We remark that these sorts of entanglement transformations were
studied by Nielsen~\cite{Nielsen99a} who gave conditions for when an
entangled state could be prepared using unlimited classical
communication.  In this context, the term ``maximally entangled''
makes sense for ebits, since together with unlimited classical
communication they can be used to prepare any other state with the
same or smaller Schmidt rank.  The low-communication case was also
considered by Daftuar and Hayden~\cite{DH-schubert}.

We now return to parts (d) and (e) of \thmref{qrst}.  In (d), we need
to run the simulation protocol for $\<\cN_F:\rho\>$ for all possible
$\rho$ in superposition.\footnote{For technical reasons, our coding
  theorem will adopt a slightly different approach.  But for the
  converse and for the present discussion, we can consider general
  inputs to be mixtures of tensor power states.}  We can discard a
resource $\beta$  at the end of the protocol, but $\beta$ must be
either independent of $\rho$ for a feedback simulation or can depend
only on $\hat\cN(\rho)$ for a non-feedback simulation.
By the equality in \eq{feedback-father}, this reduces to producing
coherent superpositions of varying amounts of qubits and ebits.

The simplest case is when $\alpha = Q(\cN)[q\ra q] + \Emb$.  In this
case, $\alpha\geqclo Q(\cN)[q\ra q] + E [qq] + \Emb$ for any $E$.  Thus we can take $\beta = \Emb$ and so we  have $\alpha \geqclo \<\cN:\rho\> + \beta$ for all $\rho$.  This establishes \eq{emb-qrst}.

The most general case without embezzling states is when
\be \alpha = Q_1 [q\ra q] + Q_2[q\la q] + C_2 [c\la c] + E [qq].
\label{eq:qqce}\ee
 In this case, we always have the constraint
\be Q_1 \geq Q(\cN) = \max_\rho \half I(R;B)_\rho,\ee
since $Q_1[q\ra q]$ is the only source of forward communication.
Suppose that $\beta = (E-e) [qq]$, for some $0 \leq e \leq E$, i.e. we
will use all of the communication, but discard $E-e$ ebits of
entanglement.
Now, for each  $\rho$, being able to simulate the channel on
input $\rho^{\ot n}$ requires creating at least $I(R;B)_\rho$ mutual
information and $I(B;E)+\beta$ entanglement which is only possible if
$$\alpha \geqclo \half I(R;B)_\rho [q\ra q]
+ (\half I(E;B)_\rho + E-e )[qq].$$
Equivalently
\bmu (Q_1 - \half I(R;B)_\rho)[q\ra q] + Q_2 [q\la q] + C_2 [c\la c]
\\ \geqclo (\half I(E;B)_\rho- e) [qq].
\label{eq:qqce-needed}\emu
We can calculate when \eq{qqce-needed} holds by using the spread
capacity expressions in \eq{spread-caps}.  First, if
$\half I(E;B)_\rho- e \geq 0$ then the $C_2 [c\la c]$ is not helpful and we simply have $Q_1 - \half I(R;B)_\rho + Q_2 \geq \half I(E;B)_\rho- e$, or equivalently
$$Q_1 + Q_2  \geq H(B)_\rho - e .$$
Alternatively, if $\half I(E;B)_\rho- e \leq 0$ then we have the inequality
$Q_1 - \half I(R;B)_\rho + Q_2 + C_2 \geq e - \half I(E;B)_\rho$,
which is equivalent to
$$Q_1 + Q_2 + C_2  \geq e - H(B|R)_\rho .$$
We will consider the case when $E$ is sufficiently large so that it
does not impose any constraints on the other parameters.  This results
in the bound
\be 2(Q_1+ Q_2) + C_2 \geq \max_\rho H(B)_\rho -
\min_\rho H(B|R)_\rho - C_E(\cN).\label{eq:qqc-needed},\ee
whose RHS is precisely $\deficit(\cN)$ from \defref{deficit}.

The role of communication can be thought of as both creating
mutual information between $R$ and $B$ and in creating entanglement
spread.  Both are necessary for channel simulation, but only forward
communication can create mutual information, while backwards or
forward communication (or even other resources, such as embezzling
states) can be used to create spread.

 % = e-H(B|R)_\rho. $$
% R + B - E - (E + B - R) = R - E = BE - E = B|E.  E - R = B|R
%\max\L\{\max_\rho H(B)_\rho - e, \min_\rho H(B|R) + e - C_2 \R\}.

% q >= a - e.  q+c >= b + e
%q >= min_{e <= E} max( a - e, b + e - c)

The non-feedback case (e) of \thmref{qrst} adds one additional
subtlety: since the simulation gives part of the input to Eve, it does
not have to preserve superpositions between as many different input
density matrices.  In particular, if the input density matrix is
$\rho^{\ot n}$, then Eve learns $\hat{\cN}(\rho)$. Thus, we need to run
our protocol in an incoherent superposition over different values of
$\hat{\cN}(\rho)$ and then in a coherent superposition within each
$\hat{\cN}^{-1}(\omega)$.  Intuitively we can think of the input as a
superposition over purifications of different tensor powers $\rho^{\ot
  n}$.  This picture can be made rigorous by the post-selection
principle~\cite{CKR09} and gentle tomography~\cite{Hayashi:02c,BHL06}, but we will not
explore this approach in detail.
In this picture Eve learns $\omega = \hat{\cN}(\rho)$ up to
accuracy $O(1/\sqrt{n})$, and this collapses the superposition to
inputs $\rho^{\ot n}$ with $\rho\in \hat{\cN}^{-1}(\omega)$.  Thus we
need only consider entanglement spread over the sets $\hat{\cN}^{-1}(\omega)$.

Unfortunately, even in the case of a fixed input
$\rho$, the additivity question is open.  Until it is resolved, we
cannot avoid regularized formulas.  However, conceptually part (e)
adds to part (d) only the issues of regularization and optimization
over ways of splitting $E^n$ into parts for Alice and Bob.

At this point it is natural to ask whether spread is only helpful for {\em feedback\/} simulations.  The amplitude damping channel
of Fig.~\ref{fig:ampl} and the variable-entropy channel both have
efficient non-feedback simulations on general inputs, using $C_E$ bits of forward
communication, even when entanglement is supplied as ordinary
ebits.   One way to see why is to observe that in each
case $H(\cN(\rho))$ is uniquely determined by $\hat{\cN}(\rho)$,
so that measuring the average density matrix of Eve will leave
no room for spread.  An optimal simulation can gently measure
the average density matrix of Eve, transmit this information
classically to Bob, and then use the appropriate number of ebits.

A second way to see that spread is not needed to simulate these two
channels is to give explicit choices of the isometry in
\eq{full-non-feedback}.   This is easier to do for the
variable-entropy channel, for which
$$\cM_{d,F} = \ket{00}^{BE}\bra{0}^A + \frac{1}{\sqrt d}
\sum_{i=1}^d \ket{ii}^{BE} \bra{1}^A.$$
Define an isometry $V_1:E\ra E_AE_B$ by
$$V_1 =
\sum_{i=1}^d
\L(
\frac{1}{\sqrt d} \ket{0i}^{E_A}\ket{i}^{E_B}\bra 0^E
+ \ket{1i}^{E_A}\ket{i}^{E_B}\bra 1^E\R).$$
Then $V_1 \circ \cM_{d,F}$ is equivalent to transmitting a classical
bit from Alice to Bob and creating a $d$-dimensional maximally
entangled state between Bob and Eve.  This can be simulated using one
cbit by having Bob locally create a
$d$-dimensional maximally mixed state.

However, the above reasoning does not extend to more complicated situations.  In \secref{clueless} we exhibit a channel whose efficient simulation requires spread-generating resources such as embezzling states or back communication even in the non-feedback setting.
%For the amplitude-damping channel, we can write
%$$V_\cN = \ket{0}^E(\ket{0}^B\bra 0^A + \sqrt{1-\gamma}
%\ket 1^B\bra 1^A)  + \sqrt\gamma \ket 1^E \ket 0^B \bra 1^A.$$

%{\em Remarks:}
%\bit\item The theorem needs to be modified to incorporate
%regularization.
%\item Perhaps (d) and (e) should be a theorem on their own.
%\eit

\subsection{Relation to other communication protocols}
Special cases of \thmref{qrst} include remote state
preparation~\cite{BHLSW03} (and the qubit-using variant, super-dense
coding of quantum states~\cite{HHL03}) for CQ-channels $\cN(\rho) =
\sum_j \bra{j} \rho \ket{j} \sigma_j$; the co-bit equality $[q\ra qq]
= ([q\ra q] + [qq])/2$~\cite{Har03}; measurement
compression~\cite{Winter:POVM} (building on~\cite{MP00,MW01}) for
qc-channels $\cN(\rho) = \sum_j \tr(\rho M_j) \proj{j}$ where $(M_j)$
is a POVM; entanglement dilution~\cite{BBPS96} for a constant channel
$\cN(\rho) = \sigma_0$; and entanglement of purification
(EoP)~\cite{purification} -- it was shown by Hayashi~\cite{Hayashi:EoP} that
optimal visible compression of mixed state sources is given by the
regularized EoP.

The Wyner protocol for producing a classical correlated
distribution~\cite{Wyner75} is a static analogue of the cbit-rbit
tradeoff.  Similarly, the entanglement of purification is a static
version of the qubits-but-no-ebits version of the QRST.

For feedback channels, \cite{devetak-triangle} showed that the QRST
can be combined with the so-called ``feedback father'' to obtain the
resource equivalence \eq{feedback-father}.  On the other hand,
\cite{devetak-triangle} also showed that running the QRST backwards on
a fixed i.i.d. source yields state merging~\cite{HOW05},
a.k.a. fully-quantum Slepian-Wolf.  This implies that merging can be used
to provide an alternate construction of the QRST on a known
i.i.d. source~\cite{ADHW06}.
More recently, \cite{YD09} has introduced {\em state redistribution} which
simultaneously generalizes state merging and splitting, by determining
the optimal rate at which a system can be sent from one party to
another when both parties hold ancilla systems that are in some way
entangled with the system being sent.

As remarked earlier, the version of the classical reverse Shannon theorem proved here,
\thmref{crst}, differs from the version originally proved in \cite{BSST01} (which
also first conjectured \thmref{qrst}).  In the earlier version, the simulation was
exactly faithful even for finite block size, and asymptotically efficient in the amount of
communication used, but exponentially inefficient in the amount of shared randomness.
The version proved here is only asymptotically faithful, but importantly stronger in being
asymptotically efficient in its use both of classical communication and shared randomness.
None of our simulations, nor other results in this area (apart from \cite{CLMW11}), achieve
the zero-error performance of \cite{BSST01}.  We believe that zero-error simulation of
classical channels using optimal rates of communication requires exponential amounts of shared randomness, and that for quantum
channels, zero-error simulations do not exist in general.  Apart from some easy
special cases (e.g. quantum feedback channels), we do not know how to
prove these conjectures.

\section{Simulation of classical channels}
\label{sec:CRST}

\subsection{Overview}
This section is devoted to the proof of \thmref{crst} (the classical
reverse Shannon theorem).  Previously the high-randomness cases of
\thmref{crst} were proved in \cite{BSST01,Winter:RST} and its converse
was proved in \cite{Winter:RST}.  Here we will review their proof and
show how it can be extended to cover the low-randomness case (parts
(c,d,e) of \thmref{crst}).  Similar results have been obtained
independently in \cite{Cuff08}.

The intuition behind the reverse Shannon theorem can be seen by
considering a toy version of the problem in which all probabilities
are uniform.   Consider a regular bipartite graph with vertices
divided into $(X,Y)$ and with edges $E\subset X\times Y$.  Since the
graph is regular, every vertex in $X$ has degree $|E|/|X|$ and every
vertex in $Y$ has degree $|E|/|Y|$.  For $x\in X$, let
$\Gamma(x)\subset Y$ be its set of neighbors.  We can use this to
define a channel from $X$ to $Y$: define $N(y|x)$ to be $1/|\Gamma(x)| = |X|/|E|$ if
$y\in \Gamma(x)$ and $0$ if not.  In other words, $N$ maps $x$ to a
random one of its neighbors.  We call these channels ``unweighted''
since their transition probabilities correspond to an unweighted graph.

In this case, it is possible to simulate the channel $N$ using a
message of size $\approx \log(|X|\cdot |Y| / |E|)$ and using $\approx
\log(|E|/|X|)$ bits of shared randomness.  This can be thought of as a
special case of part (a) of \thmref{crst} in which $N$ is an unweighted
channel and we are only simulating a single use of
$N$.  This is achieved by approximately decomposing $N$ into a probabilistic mixture of channels and using the shared randomness to select which one to use.  We will choose these channels such that their ranges are disjoint subsets of  $Y$, and in fact, will construct them by starting with a partition of $Y$ and working backwards.
The resulting protocol is analyzed in the following lemma.

\begin{lemma}\label{lem:flat-crst}
Consider a channel $N:X\ra Y$ with $N(y|x) = 1_{\Gamma(x)}(y) /
|\Gamma(x)|$, where $1_S$ denotes the indicator function for a set
$S$.  Choose positive integers $r,m$ such
that $rm=|Y|$ and let $\gamma = m|E|/|X|\,|Y|$.  Choose a random
partition of $Y$ into subsets $Y_1,\ldots,Y_r$, each of size $m$, and
for $y\in Y$ define $i(y)$ to be the index of the block containing
$y$.    Define
$$\tilde{N}(y|x) = \frac{1_{\Gamma(x)}(y)}{r\cdot |\Gamma(x) \cap Y_{i(y)}|}$$
to be the channel that results from the following protocol:
\benum\item Let $i\in[r]$ be a uniformly chosen random number shared by
Alice and Bob.
\item Given input $x$, Alice chooses a random element of
  $\Gamma(x)\cap Y_i$ (assuming that one exists) and transmits its
  index $j\in[m]$ to Bob.
\item Bob outputs the $j^{\text{th}}$ element of $Y_i$.
\eenum
Then it holds with probability $\geq 1-2r e^{-\gamma\eps^2}$
that $\|N(\cdot | x) -  \tilde N(\cdot | x)\|_1\leq \eps$ for all $x$.
\end{lemma}

If we choose $\gamma = 2(\ln 4|E|)/\eps^2$ then there is a nonzero
probability of a good partition existing. In this case we can derandomize
the construction and simply say that a partition of $Y$ exists such that
the above protocol achieves low error on all inputs.

The idea behind \lemref{flat-crst} is that for each $x$ and $i$, the random
variable $|\Gamma(x) \cap Y_i|$ has expectation close to
$$|\Gamma(x)|\cdot |Y_i| / |Y| = \frac{|E|}{|X|}\cdot \frac{m}{|Y|} =
\gamma,$$ with typical fluctuations on the order of $\sqrt{\gamma}$.
If $\gamma$ is large then these fluctuations are relatively small, and
the channel simulation is faithful.  Similar ``covering lemmas'' appeared
in Refs.~\cite{Winter:RST,Cuff08}, and were anticipated by Ref.~\cite{HV93}
and Thm 6.3 of \cite{Wyner75}.
The details of the proof are
described in \secref{flat-crst}.

The difference between \lemref{flat-crst} and the classical reverse
Shannon theorem (i.e. part (a) of \thmref{crst}) is that in the latter we are
interested in an asymptotically growing number of channel uses $n$ and
in simulating general channels $N$, instead of unweighted channels.
It turns out that when $n$ is large, $N^n$ looks mostly like an
unweighted channel, in a sense that we will make precise in
\secref{types}.  We will see that Alice need communicate only
$O(\log(n))$ bits to reduce the problem of simulating $N^n$ to the
problem of simulating an unweighted channel.  This will complete the
proof of the direct part of part (a) of \thmref{crst}.

One feature of the protocol in \lemref{flat-crst} is that Bob uses
only shared randomness ($i$) and the message from Alice ($j$) in order
to produce his output $y$.  As a result, the protocol effectively
simulates the feedback channel $N_F$ in which Alice also gets a copy
of $y$.  Conversely, in order to simulate a feedback channel, Bob
cannot use local randomness in any significant way.

On the other hand, if Alice does not need to learn $y$, then we can
consider protocols in which some of the random bits used are shared
and some are local to Alice or Bob.  This will allow us to reduce the
use of shared randomness at the cost of some extra communication.  The
resulting trade-off between the resources is given in part (c) of
\thmref{crst}.  In order to prove it, we will again first consider the
unweighted case.

The idea will be to decompose the channel $N(y|x)$ as the composition
of channels $N_1(w|x)$ and $N_2(y|w)$; i.e. $N(y|x) = (N_2\circ
N_1)(x) = \sum_{w\in W} N_1(w|x) N_2(y|w)$.  In this case Alice can
simulate the channel $N$ on input $x$, by simulating $N_1$ to produce
intermediate output $w$ on which Bob locally applies $N_2$ to produce
$y$.  Since $w$ is generally more correlated with $x$ than $y$, this
will require more communication than simply simulating $N$ directly as
in \lemref{flat-crst}.  However, since Bob simulates $N_2$ using local
randomness, the protocol may require less shared
randomness, and more importantly, the total amount of communication
plus shared randomness may be lower.

We will assume that the channels $N, N_1$ and $N_2$ are all
unweighted channels.  Let the corresponding bipartite graphs for $N,
N_1,N_2$ have edges
$E_{XY} \subset X\times Y$, $E_{XW} \subset X\times W$ and $E_{YW}
\subset W\times Y$, respectively.  We use $\Gamma_{XY}(x)$ to denote
the neighbors of $x$ in $Y$; that is, $\Gamma_{XY}(x) = \{y : (x,y)
\in E_{XY}\}$.  Similarly, we can define $\Gamma_{YX}(y)$ to be the
neighbors of $y$ in $X$, $\Gamma_{XW}(x)$ to be the neighbors of $x$
in $W$ and so on.  We assume that the graphs are regular, so that
$|\Gamma_{XW}(x)|=|E_{XW}|/|X|$ for all $x$,
$|\Gamma_{WY}(w)|=|E_{WY}|/|W|$ for all $w$, and so on.  Combined with
the fact that $N=N_2\circ N_1$, we find that
\bmu \frac{1_{\Gamma_{XY}(x)}(y)}{|E_{XY}|/|X|} = N(y|x) =
\sum_{w\in W}N_2(y|w) N_1(w|x) \\=
\sum_{w\in W} \frac{1_{\Gamma_{XW}(x)}(w)}{|E_{XW}|/|X|} \cdot
\frac{1_{\Gamma_{WY}(w)}(y)}{|E_{WY}|/|W|}
\\ =
\frac{|\Gamma_{XW}(x)\cap \Gamma_{YW}(y)|}
{|E_{XW}|\cdot|E_{WY}| / |X|\, |W|},
\label{eq:Nyx-markov}\emu
Rearranging terms yields the identity
 \be |\Gamma_{XW}(x)\cap \Gamma_{YW}(y)| = 1_{\Gamma_{XY}(x)}(y)
\frac{|E_{XW}|\, |E_{WY}| }{ |E_{XY}|\, |W|}.
\label{eq:intersection}\ee

The protocol is now defined in a way similar to the one in
\lemref{flat-crst}.
\begin{lemma}\label{lem:flat-wyner}
Choose positive integers $r,m$ such that $m = \gamma
|X|\,|W|/|E_{XW}|$ and $r = |E_{XY}|\,|W| / |E_{WY}|\, |X|$.
Choose disjoint sets $W_1,\ldots,W_r\subset W$ at random, each of size
$m$.  Let $W_i = \{w_{i,1}, \ldots, w_{i,m}\}$.  Let $\tilde N(y|x)$ be
the channel resulting from the following protocol:
\benum\item Let $i\in[r]$ be a uniformly chosen random number shared by
Alice and Bob.
\item Given input $x$, Alice chooses a random $w_{i,j} \in
  \Gamma(x)\cap W_i$ (assuming that one exists) and transmits its
  index $j\in[m]$ to Bob.
\item Bob outputs $y$ with probability $N_2(y|w_{i,j})$.
\eenum
Then it holds with probability $\geq 1-2|E_{XY}|e^{-\gamma\eps^2/32}$
that $\|N(\cdot | x) -  \tilde N(\cdot | x)\|_1\leq \eps$ for all $x$.
\end{lemma}
We can take $\gamma = 32(\ln 2|E_{XY}|)/\eps^2$ and derandomize the statement of the Lemma.

Note that in general we will have $rm < |W|$, so that this protocol
does not use all of $W$.  This should not be surprising, since
faithfully simulating the channel $N_1$ should in general be more
expensive than simulating $N$.  The trick is to modify the simulation
of $N_1$ that would be implied by \lemref{flat-crst} to use less
randomness, since we can rely on Bob's application of $N_2$ to add in
randomness at the next stage.

\subsection{Proof of unweighted classical reverse Shannon theorem}
\label{sec:flat-crst}

In this section we prove \lemref{flat-crst} and \lemref{flat-wyner}.
The main tool in both proofs is the Hoeffding bound for the hypergeometric
distribution~\cite{Hoeffding63}.  The version we will need is
\begin{lemma}[Hoeffding~\cite{Hoeffding63}]\label{lem:hyper}
For integers $0 < a \leq b < n$, choose $A$ and $B$ to be random subsets of
$[n]$ satisfying $|A|=a$ and $|B|=b$.  Then $\mu:=\bbE[|A\cap B|]=ab/n$ and
\ba
\Pr\L[|A \cap B|\geq(1+\eps)\mu\R] &\leq e^{-\frac{\mu\eps^2}{2}} \\
\Pr\L[|A \cap B|\leq(1-\eps)\mu\R] &\leq e^{-\frac{\mu\eps^2}{2}} \\
\Pr\L[\L|\, |A \cap B| - \mu \R| \geq \eps\mu\R] &
\leq 2 e^{-\frac{\mu\eps^2}{2}}
\ea
\end{lemma}

Now we turn to \lemref{flat-crst}.  We can calculate
\bmu \| N(\cdot | x) -  \tilde N(\cdot | x)\|_1
 = \sum_{y\in \Gamma(x)} | N(y|x) - \tilde N(y|x)|
\\ =
\sum_{y\in \Gamma(x)} \L|\frac{1}{|\Gamma(x)|} - \frac{1}{r\cdot
|\Gamma(x) \cap Y_{i(y)}|}\R|
\\  =
\sum_{i=1}^r \sum_{y\in \Gamma(x)\cap Y_i} \L|\frac{1}{|\Gamma(x)|} -
\frac{1}{r\cdot |\Gamma(x) \cap Y_{i}|}\R|
 \\  =
\sum_{i=1}^r \L|\frac{|\Gamma(x) \cap Y_{i}|}{|\Gamma(x)|} -
\frac{1}{r}\R|\label{eq:flat-proof-final}\emu
To apply \lemref{hyper}, take $A=\Gamma(x)$ and $B=Y_i$, so that
$a=|E|/|X|=r\gamma$, $b=m=|Y|/r$, $n=|Y|$ and $\mu=\gamma$.  Then
each term in the sum in \eq{flat-proof-final} is $\leq \eps/r$ with
probability $\geq 1-2e^{-\gamma\eps^2/2}$.  Taking the union bound
over all $a$ and $i$ completes the proof of \lemref{flat-crst}.

\oldcomment{I am sure that there's some way to avoid the union bound,
and thereby avoid having to put a $\log(|E|)$ factor into $\gamma$.
But this is the sort of thing to deal with later.  or never.}

The proof of \lemref{flat-wyner} is similar.  This time
\bas \tilde N(y|x) &= \frac{1}{r}\sum_{i=1}^r
\sum_{w\in W_i} \Pr\L[\text{Alice sends }w|x,i\R] N_2(y|w)
\\ &
=\frac{1}{r}\sum_{i=1}^r \sum_{w\in W_i}
\frac{1_{\Gamma_{XW}(x)}(w)}{|\Gamma_{XW}(x) \cap W_i|}
\cdot \frac{1_{\Gamma_{WY}(w)}(y)}{|\Gamma_{WY}(w)|}
\\ & = \frac{|W|}{r|E_{WY}|}\sum_{i=1}^r
\frac{|\Gamma_{XW}(x) \cap \Gamma_{YW}(y) \cap W_i|}
{|\Gamma_{XW}(x) \cap W_i|}
.\eas
We will use \lemref{hyper} twice.  First, consider $|\Gamma_{XW}(x) \cap W_i|$.  This has expectation equal to $\gamma$ and therefore
$$\Pr\L[|\Gamma_{XW}(x) \cap W_i| \geq (1+\eps/4)\gamma\R] \leq e^{-\gamma\eps^2/32}.$$
(We will see that the one-sided bound  simplifies some of the later calculations.)
Taking the union bound over all $|X|r %= |E_{XY}|\,|W|/|E_{WY}|
\leq |E_{XY}|$ values of $x,i$, we find that
$|\Gamma_{XW}(x) \cap W_i| \leq (1+\eps/4)\gamma$ for all $x,i$ with probability $\geq 1-|E_{XY}| e^{-\gamma\eps^2/32}$.
Assuming that this is true, we obtain
\ba \tilde N(y|x) & \geq
\frac{|W|}{r|E_{WY}|}\sum_{i=1}^r
\frac{|\Gamma_{XW}(x) \cap \Gamma_{YW}(y) \cap W_i|}
{(1+\eps/4)\gamma}\non\\
&= \frac{|W|}{r|E_{WY}|}
\frac{|\Gamma_{XW}(x) \cap \Gamma_{YW}(y) \cap \tilde W|}
{(1+\eps/4)\gamma},
\label{eq:wyner-good1}\ea
where we define $\tilde W = W_1\cup \ldots \cup W_r$.  Note that $\tilde W$ is a random subset of $W$ of size $rm$.  Using \eq{intersection} we find that
$\bbE[|\Gamma_{XW}(x) \cap \Gamma_{YW}(y) \cap \tilde W|]$ is equal to $\gamma$ when $(x,y)\in E_{XY}$ and 0 otherwise.  Again we use \lemref{hyper} to bound
$$\Pr\L[|\Gamma_{XW}(x) \cap \Gamma_{YW}(y) \cap \tilde W| \leq (1-\eps/4)\gamma]\R] \leq e^{-\gamma\eps^2/32}$$
for all $(x,y)\in E_{XY}$.
Now we take the union bound over all pairs $(x,y)\in E_{XY}$ to find that
\be |\Gamma_{XW}(x) \cap \Gamma_{YW}(y) \cap \tilde W| \geq (1-\eps/4)\gamma
\label{eq:wyner-good2}\ee
 with probability $\geq 1-|E_{XY}|e^{-\gamma\eps^2/32}$.
 When both \eq{wyner-good1} and \eq{wyner-good2} hold and $(x,y)\in E_{XY}$ it follows that
 \ba \tilde N(y|x) & \geq \frac{|W|}{r|E_{WY}|} \frac{1-\eps/4}{1+\eps/4}
 >(1-\eps/2)\frac{W}{r|E_{WY}|} \non\\ &
= (1-\eps/2)\frac{|X|}{|E_{XY}|}
 = (1-\eps/2)N(y|x) .\ea
Finally we compare with \eq{Nyx-markov} to obtain
\ba \| N(y|x) - \tilde N(y|x)\|_1 &=
2 \sum_{y\in Y}\max(0,  N(y|x) - \tilde N(y|x))
\non \\ & < \eps \sum_{y\in Y} N(y|x) = \eps.\ea
This concludes the proof of \lemref{flat-wyner}.

\subsection{Classical types}\label{sec:types}

In this section we show how the classical method of types can be used to
extend Lemmas \ref{lem:flat-crst} and \ref{lem:flat-wyner} to prove the
coding parts of \thmref{crst}.  We begin with a summary of the arguments
aimed at readers already familiar with the method of types (a more
pedagogical presentation is in \cite{CT91}).  The idea is
to for Alice to draw a joint type according to the appropriate distribution
and to send this to Bob.  This requires $O(\log(n))$ bits of communication and
conditioned on this joint type they are left with an unweighted channel and
can apply \lemref{flat-crst}.  It is then a counting exercise to show that
the communication and randomness costs are as claimed.  For the low-randomness
case, the protocol is based on a decomposition $N^{X\ra Y}$ into
$N_2^{W\ra Y}\circ N_1^{X\ra W}$.  Alice draws an appropriate joint
type for all three variables $(X,W,Y)$ and transmits this to Bob.
Again this involves $O(\log(n))$ bits of communication and leaves
them with an unweighted channel, this time of the form that can be
simulated with \lemref{flat-wyner}.

To prove these claims, we begin by reviewing the method of types, following
\cite{CT91}.  We will use $\cX, \cY, \cW$ to denote single-letter
alphabets, while reserving $X,Y,W$ for block variables.  Consider a string $x^n=(x_1,\ldots,x_n)\in \cX^n$.
Define the type of $x^n$ to be the $|\cX|$-tuple of integers $t(x^n)
:= \sum_{j=1}^n e_{x_j}$, where $e_j\in\bbZ^{|\cX|}$ is the unit
vector with a one in the i$^{\text{th}}$ position.  Thus $t(x^n)$ counts
the frequency of each symbol $x\in\cX$ in $x^n$.  Let
$\cT_\cX^n$
denote the set of all possible types of strings in $\cX^n$ .  Since an
element of $\cT_\cX^n$ can be written as $|\cX|$ numbers ranging from
$0,\ldots,n$ we obtain the simple bound
$|\cT_\cX^n|=\binom{n+|\cX|-1}{|\cX|-1}\leq (n+1)^{|\cX|}$.  For a
type $t$, let the normalized probability distribution $\bar{t}:=t/n$
denote its empirical distribution.

For a particular type $t\in\cT_\cX^n$, denote the set of all strings in
$\cX^n$ with type $t$ by $T_t=\{x^n\in \cX^n : t(x^n)=t\}$.  From
\cite{CT91}, we have
\be (n+1)^{-|\cX|} \exp(nH(\bar{t})) \leq |T_t| = \binom{n}{t}
\leq
\exp(nH(\bar{t})), \label{eq:binom-bounds}\ee
where $\binom{n}{t}$ is defined to be  $\frac{n!}{\prod_{x\in \cX}t_x!}$.
Next, let $p$ be a
probability distribution on $\cX$ and $p^{\ot n}$ the probability
distribution on $\cX^n$ given by $n$ i.i.d.~copies of $p$,
i.e. $p^{\ot n}(x^n) := p(x_1)\cdots p(x_n)$.  Then for any $x^n\in T_t$ we
have $p^{\ot n}(x^n) = \prod_{x\in \cX} p(x)^{t_x} =
\exp(-n(H(\bar{t}) + D(\bar{t}\|p)))$.  Combining this with
\eq{binom-bounds}, we find that
\be \frac{\exp\L(-n D(\bar{t}\|p)\R)}{(n+1)^{|\cX|}} \leq
 p^{\ot n}(T_t) \leq \exp\L(-n D(\bar{t}\|p)\R),
 \label{eq:PT-bounds}\ee

Thus, as $n$ grows large, we are likely to observe an empirical
distribution $\bar{t}$ that is close to the actual distribution $p$.
To formalize this, define the set of {\em typical sequences}
$T_{p,\delta}^n$ by
\be T_{p,\delta}^n :=
\bigcup_{\substack{t\in\cT_\cX^n\\\|\bar{t}-p\|_1\leq\delta}}
 T_t.\ee
To bound $p^{\ot n}(T_{p,\delta}^n)$, we apply Pinsker's
inequality~\cite{Pinsker64}:
\be D(q\|p) \geq \frac{1}{2\ln 2}\|p-q\|_1^2
\label{eq:pinsker}\ee
to show that
\be p^{\ot n}(T_{q,\delta}^n)
\geq 1- (n+1)^{|\cX|} \exp\L(-\frac{n\delta^2}{2\ln 2}\R).
\label{eq:p-typical}\ee
We will also need the Fannes-Audenaert
inequality~\cite{Fannes73,Aud-ineq} which establishes the continuity
of the entropy function.
Let $\eta(x) = -x\log x-(1-x)\log(1-x)$.
Then if $p,q$ are probability distribution on $d$ letters,
 \be  |H(p)-H(q)| \leq
\frac{1}{2} \|p-q\|_1\log(d-1) + \eta\L(\frac{1}{2}\|p-q\|_1\R)
 \label{eq:fannes}\ee

If we have a pair of strings $x^n\in \cX^n,y^n\in\cY^n$, then we can
define their joint type $t(x^ny^n)$ simply to be the type of the
string $(x_1y_1,\ldots,x_ny_n)\in(\cX \times \cY)^n$.  Naturally the
bounds in \eq{binom-bounds} and \eq{PT-bounds} apply equally well to joint types, with $\cX$ replaced by $\cX\times\cY$.  If $t$ is a joint type then we can define its marginals
$t^\cX \in \bbZ^{|\cX|}$ and  $t^\cY \in \bbZ^{|\cY|}$ by
$t^\cX_x = \sum_{y\in\cY} t_{x,y}$ and
$t^\cY_y = \sum_{x\in\cX} t_{x,y}$.
Let $N(y|x)$
denote a noisy channel from $\cX\ra \cY$ with $N^n(y^n|x^n) :=
N(y_1|x_1)\cdots N(y_n|x_n)$.  Then $N^n(y^n|x^n)$ depends only on the
type $t=t(x^ny^n)$ according to $N^n(y^n|x^n) = \prod_{x,y} N(y|x)^{t_{x,y}}$.

We now have all the tools we need to reduce \thmref{crst} to Lemmas
\ref{lem:flat-crst} and \ref{lem:flat-wyner}.   First, consider
parts (a,b) of \thmref{crst}, where we have an ample supply of
shared randomness.  In either case, the protocol is as follows:

\benum
\item Alice receives input $x^n$.  This may be expressed in
a type-index representation as $(t_A, p_A)$.  Here $t_A=t(x^n)$ is
the input type, and $p_A \in [|T_{t_A}|]$ is defined by assigning the
integers $1,\ldots,|T_{t_A}|$ arbitrarily to the elements of $T_{t_A}$.
\item Alice simulates the channel $N^n$ locally to generate a provisional
output $\tilde{y}^n$.  Let $t_B=t(\tilde{y}^n)$ be the output type and
$t_{AB}\in\bbZ^{|\cX\times\cY|}$ the joint type between $x^n$ and
$\tilde{y}^n$. Having determined these types, Alice discards
$\tilde{y}^n$, as it is no longer needed.
\item Alice sends $t_B$ to Bob using $|\cX\times\cY|\log(n+1)$ bits.
\item Alice and Bob use $n(H(Y)-C)+o(n)$ bits of shared randomness to
pick a subset $S_i$ from a preagreed partitioning of outputs of type
$t_B$ into approximately equal disjoint subsets, each of cardinality
approximately $2^{nC}$, where $C$ is the Shannon capacity of channel
$N$.
\item Alice finds a string $y^n\in S_i$ having the same joint type with $x^n$
as $\tilde{y}^n$ had. But because $y^n$ lies in the chosen subset
$S_i$, which Bob already knows, Alice can transmit $y^n$ to Bob more
efficiently, by a message of size only $nC+o(n)$ bits, using the
method of \lemref{flat-crst} (Let $X=T_{t_A}$, $Y=T_{t_B}$ and
$E=T_{t_{AB}}\subset X\times Y$ define a regular bipartite graph. To
simulate the action of $N^n$ on $x^n\in X$, conditioned on
$(x^n,y^n)\in E$, we need only to choose a random neighbor $y^n$ of
$x^n$ in this graph.) \eenum

This protocol is depicted in \fig{effCRST}.
\begin{figure*}[htbp]
\includegraphics[width=0.9\textwidth]{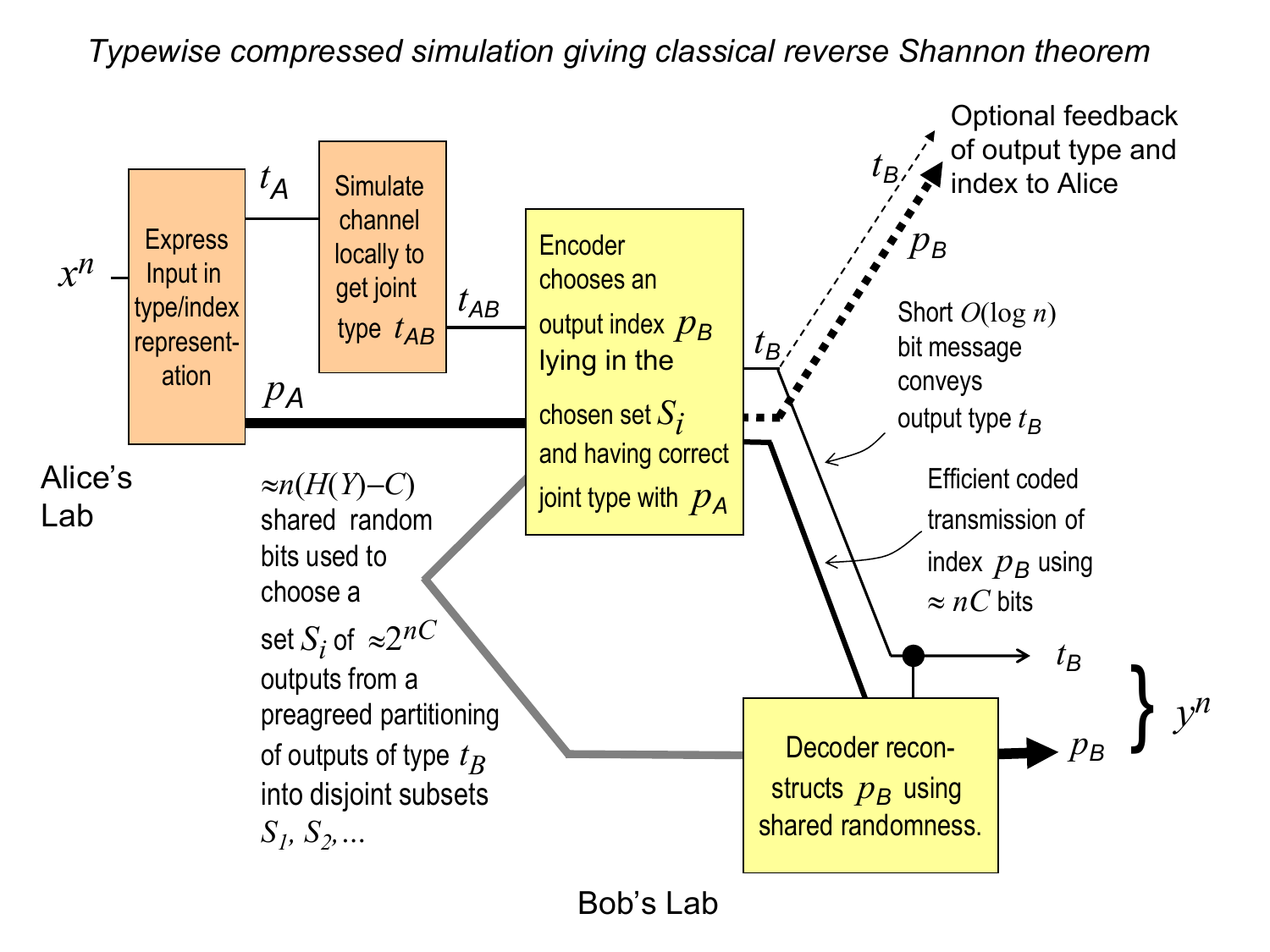}
\caption{The protocol for the classical reverse Shannon theorem (\thmref{crst}).}
\label{fig:effCRST}
\end{figure*}

%\benum\item Suppose Alice's input is $x^n$.
%\item Alice simulates $N^n(x^n)$ locally to obtain $\tilde y^n$.
%\item Let $t=t(x^n\tilde y^n)$.  Alice sends $t$ to Bob using $|\cX\times\cY|\log(n+1)$ bits.
%\item Let $X=T_{t^\cX}$, $Y=T_{t^\cY}$ and $E=T_t \subset X\times Y$ define a regular bipartite graph.  To simulate the action of $N^n$ on $x^n\in X$, conditioned on $(x^n,y^n)\in T_t=E$, we need only to choose a random neighbor of $x^n$ in this graph.    This is achieved with \lemref{flat-crst}.
%\eenum
%Note that Alice uses $\tilde y^n$ only to obtain the correct distribution for the joint type and, other than belonging to the same conditionally typical set, it has no correlation with $y^n$.

It remains only to analyze the cost of this last step.  The communication cost is (taking notation from the statement of \lemref{flat-crst})
\bas \log(m) &= \log\L(\frac{|X|\,|Y|\gamma}{|E|}\R) \\ &=
\log\L(\frac{|T_{t_A}|\,|T_{t_B}|\,(2\ln(2|T_{t_{AB}}|)/\eps^2)}{|T_t|}\R)
\\ & \leq n(H(\bar t_A)+H(\bar t_B)-H(\bar t_{AB})) + |\cX\times\cY|\log(n+1)
\\ & \qquad + \log(2\ln(2)nH(\bar t_{AB})/\eps^2)
\\ & = n I(\cX;\cY)_{\bar t_{AB}}  + O(\log(n)) + \log(1/\eps^2).
\eas
Since $C(N) \geq I(\cX;\cY)_{\bar t}$ for all $t$, this establishes part (b) of \thmref{crst}.
Continuing, we estimate the randomness cost to be
\bas \log(r) &= \log\L(\frac{|E|}{|X|\gamma}\R) \leq
\log\L(\frac{|E|}{|X|}\R)
\\ & \leq n(H(\bar t_{AB}) - H(\bar t_A)) + O(\log(n))
\\ & = nH(\cY|\cX)_{\bar t_{AB}} + O(\log(n)) .\eas

To prove part (a), we need to relate entropic quantities defined for
$\bar t$ to the corresponding quantities for $p$.   This will be done
with typical sets (\eq{p-typical}) and the Fannes-Audenaert inequality
(\eq{fannes}).   If $p$ is a distribution on $X$ then let $q=N_F(p)$
be the joint distribution on $X$ and $Y$ that results from sending $X$
through $N$ and obtaining output $Y$.  Then \eq{p-typical} implies
that following the above protocol results in values of $\bar t$ that
are very likely to be close to $q$.  In particular, $q^{\ot
  n}(T_{q,\delta}^n) \geq 1-(n+1)^de^{-n\delta^2/2}$, where
$d=|\cX\times\cY|$.  Next, the Fannes-Audenaert inequality says that if $\bar t\in
T_{q,\delta}^n$ then $|H(\bar t) - H(q)| \leq
\delta \log(d/\delta)$.   Applying this to each term in $I(\cX;\cY) = H(\cX)
+ H(\cY) - H(\cX\cY)$, we obtain that
$|I(\cX;\cY)_{\bar t} - I(\cX;\cY)_q| \leq 3\delta \log(d/\delta)$ and
$|H(\cY|\cX)_{\bar t} - I(\cY|\cX)_q| \leq 2\delta \log(d/\delta)$.
Taking $\delta$ to be $n^{-1/4}$, we obtain a sequence of protocols
where both error and inefficiency simultaneously vanish as $n\ra
\infty$.

Similarly, for part (c), we need to consider the joint distribution $q$ of $\cX\cW\cY$ that results from drawing  $\cX$ according to $p$, sending it through $N_1$ to obtain $\cW$ and then sending $\cW$ through $N_2$ to obtain $Y$.    The protocol is as follows:
\benum
\item Suppose Alice's input is $x^n$.
\item Alice simulates $N_1^n(x^n)$ to obtain $\tilde w^n$ and then simulates $N_2^n(\tilde w^n)$ to obtain $\tilde y^n$.
\item Alice sets $t_{AWB}=t(x^n\tilde w^n\tilde y^n)$.  She will not make any further use of $\tilde w^n$ or $\tilde y^n$.
\item Alice sends $t$ to Bob using $|\cX\times\cW\times \cY|\log(n+1)$ bits.
\item Define $X=T_{t_A}$, $W=T_{t_\cW}$, $Y=T_{t_B}$, $E_{XY}=T_{t_{AB}}$, $E_{XW}=T_{t_{AW}}$ and
$E_{WY}=T_{t_{WB}}$.  To simulate the action of $N^n$ on $x^n\in X$, conditioned on $(x^n,w^n,y^n)\in T_t$, we need only to choose a random element of $\Gamma_{XY}(x^n)$ in this graph.    This is achieved with \lemref{flat-wyner}.
\eenum
The analysis of this last step is similar to that of the previous protocol.  The communication cost is
$\log(m) = \log\L(|X|\,|W|\gamma/|E_{XW}|\R) = nI(\cX;\cW)_{\bar t} +O(\log n)$ and the randomness cost is
\bmu \log(r) =\log\L(\frac{|E_{XY}|\cdot |W|} {|E_{WY}|\,|X|}\R)
\\
= n(H(\cX\cY)_{\bar t} + H(\cW)_{\bar t} - H(\cW\cY)_{\bar t} - H(\cX)_{\bar t}) + O(\log n)
 \\ \text{using \eq{binom-bounds}}
\\  =
n(H(\cX\cY)_{\bar t} + H(\cX\cW)_{\bar t} - H(\cX\cW\cY)_{\bar t} - H(\cW)_{\bar t}) + O(\log n)
\\ \text{using the Markov condition $I(\cX;\cY|\cW)_{\bar t}=0$}
\\  \shoveright{= nI(\cY;\cW|\cX)_{\bar t} + O(\log n)}
\\  = n(I(\cX\cY;\cW)_{\bar t} - I(\cX;\cW)_{\bar \tau}) + O(\log n)
\emu
This concludes the proofs of the existence of channel simulations claimed in \thmref{crst}.

\subsection{Converses}
In this section we discuss why the communication rates for the above protocols cannot be improved.  The lower bound for simulating feedback channels was proven in \cite{Winter:RST} and for non-feedback channels in \cite{Cuff08}.    We will not repeat the proofs here, but only sketch the intuition behind them.

First, the communication cost must always be at least $C(N)$, or
$I(X;Y)_p$ if the input is restricted to be from the distribution $p$.
Otherwise we could combine the simulation with Shannon's [forward]
noisy channel coding theorem  to turn a small number of noiseless
channel uses into a larger number of uses.  This is impossible even
when shared randomness is allowed.

Next, if $N_F$ (i.e., the channel including noiseless feedback) is to be simulated,
then Bob's output (with entropy $H(Y)$) must be entirely determined by the $C$ bits of classical communication sent and the $R$ bits of shared randomness used.  Therefore we must have $C+R\geq H(Y)$.

The situation is more delicate when the simulation does not need to
provide feedback to Alice.  Suppose we have a protocol that uses $C$
cbits and $R$ rbits.  Then let $W=(W_1,W_2)$ comprise both the message
sent ($W_1$) and the shared random string ($W_2$).  We immediately
obtain $I(XY;W)\leq H(W)\leq C+R$.  Additionally, the shared
randomness $W_2$ and the message $X$ are independent even given the
message $W_1$; in other words $I(X;W_2|W_1)=0$.  Thus $I(X;W) =
I(X;W_1) \leq H(W_1)\leq C$.  Finally we observe that $X-W-Y$
satisfies the Markov chain condition since Bob produces $Y$ only by
observing $W$.  This argument is discussed in more detail in
\cite{Cuff08}, where it is also proven that it suffices to consider
single-letter optimizations.

These converses are also meaningful, and essentially unchanged, when
we consider negative $R$, corresponding to protocols that output
shared randomness.

We observe that some of these converses are obtained from coding
theorems and others are obtained from more traditional entropic
bounds.   In the cases where the converses are obtained from coding
theorems then we in fact generally obtain strong converses, meaning
that fidelity decreases exponentially when we try to use less
communication or randomness than necessary.  This is discussed in
\cite{Winter:RST} and we will discuss a quantum analogue of this point
in \secref{converses}.

\section{Simulation of quantum channels on arbitrary inputs}
\label{sec:qrst}
This section is devoted to proving parts (d) and (e) of \thmref{qrst}.

\subsection{The case of flat spectra}\label{sec:flat-qrst}
By analogy with \secref{flat-crst}, we will first state an
unweighted or ``flat'' version of the quantum reverse Shannon theorem.  We
will then use a quantum version of type theory (based on Schur-Weyl
duality) to extend this to prove the QRST for general inputs.

\begin{definition}\label{def:flat}
An isometry $V^{A\ra BE}$ is called {\em flat} if, when applied to
half of a maximally entangled state $\ket{\Phi}^{RA}$, it produces a
state $\ket{\psi}^{RBE}$ with $\psi^R$, $\psi^B$ and $\psi^E$ each
maximally mixed.
\end{definition}
We note two features of the definition.  First, the requirement that
$\psi^A$ be maximally mixed is satisfied automatically, but we include
it to emphasize that each marginal of $\psi$ should be maximally
mixed.  Second, the definition of a flat isometry does not depend on
the choice of maximally entangled input $\ket\Phi$.

An important special case of flat channels occurs when $A,B,E$ are irreps of some group $G$ and $V$ is a $G$-invariant map.  We will return to this point in \secref{schur-channel}.

\begin{lemma}[\cite{HOW05,ADHW06}]\label{lem:flat-qrst}
Let $A, B, E$ have dimensions $D_A, D_B, D_E$
respectively. Consider furthermore quantum systems $K_A,K_B, M$ with
dimensions $D_K, D_K, D_M$, respectively, such that $D_B=D_KD_M$ and
$D_M \geq \frac{256}{\delta\eps^4} \sqrt{D_A D_B/D_E}$.
If $V^{A\ra BE}$ is a flat isometry, it can be simulated up to error
$\eps$ with respect to the maximally mixed input state,
by consuming $\log(D_K)$ ebits and sending $\log(D_M)$ qubits
from Alice to Bob.
More precisely, there exist isometries $S_H^{K_BM\ra B}, S_V^{AK_A\ra
  ME}$ such that
\begin{multline} V^{A\ra BE}\ket{\Phi_{D_A}}^{RA} \approx_\eps \\
S_H^{K_BM\ra B}  S_V^{AK_A\ra ME} \ket{\Phi_{D_K}}^{K_AK_B}\ket{\Phi_{D_R}}^{RA}.
\label{eq:flat-splitting}\end{multline}
Furthermore, $S_H$ can be taken to be a Haar random unitary of
dimension $D_B$ with $S_V$ chosen deterministically based on $S_H$
and $V$.  \eq{flat-splitting} holds with probability $\geq 1 - \delta$ over the random choice of $S_H$.
\end{lemma}

The protocol in \eq{flat-splitting} is depicted in \fig{splitting}.
\begin{figure}[ht]
\centerline{\Qcircuit @=1em {
\lstick{K_B\hspace{1.5em}} & \ar @{-} [1,-1]  & \qw  & \qw & \multigate{2}{S_H}\qw \\
\lstick{\ket{\Phi_{D^K}}} &&&&\pureghost{S_H} & \rstick{B} \qw \\
\lstick{K_A\hspace{1.5em}} & \ar @{-} [-1,-1] & \multigate{1}{S_V}\qw  & \qw_<{M}
& \ghost{S_H}\\
\lstick{A\hspace{2em}} &  \ar @{-} [1,-1]& \ghost{S_V} \qw &\qw & \qw& \rstick{E} \qw \\
\lstick{\ket{\psi}} & \\
\lstick{R\hspace{2em}} & \ar @{-} [-1,-1] & \qw & \qw & \qw &\rstick{R}\qw \\
}}
\caption{The simulation of flat channels described in
  \lemref{flat-qrst}.  The entangled state $\ket{\Phi_{D_K}}$ is
  consumed in order to simulate the action of $V^{A\ra BE}$ on the $A$
  part of $\ket{\psi}^{RA}$.  $S_H$ is chosen from the Haar measure on
  $\cU_{D_B}$, while $S_V$ is chosen to depend on $S_H$ and $V$, as
  described by \cite{ADHW06}. While any $\ket\psi^{RA}$ can be input
  into the channel, the fidelity guarantee of \eq{flat-splitting} only
  holds when $\psi^A$ is maximally mixed.}
\label{fig:splitting}
\end{figure}

In an earlier unpublished version of this work, we proved a version of
\lemref{flat-qrst} using the measurement compression theorem of
\cite{Winter:POVM}. This version used classical instead of quantum
communication (with correspondingly different rates), but by making
the communication coherent in the sense of \cite{Har03,DHW05} it is
possible to recover \lemref{flat-qrst}.

However, a conceptually simpler proof of \lemref{flat-qrst} was later
given by \cite{devetak-triangle,HOW05,ADHW06}.  This proof is based on
reversing ``state merging,'' which we can think of as the task of Bob
sending a subsystem $B$ to Alice in a way that preserves its correlation
with a subsystem $E$ which Alice already has, as well as with a purifying reference
system $R$.  In other words, merging is a state redistribution of the form
\be \Psi^{R:E:B} \ra \Psi^{R:EB}.\ee
The simplest proof of state merging is given in
\cite{ADHW06}, where it is shown that if Bob splits $B$ randomly into
systems $M$ and $K_B$ of the appropriate sizes (i.e. by applying
$S_H^\dag$), and sends $M$ to Alice,
then Alice will be able to locally transform $E,M$ into two subsystems
$A$ and $K_A$ such that $A$ is completely entangled with the reference
system $R$ (and thus can be locally transformed by Alice into $E,B$, the desired
goal of the merging.).  On the other hand $K_A$ is nearly completely entangled with
the remaining $K_B$ system that Bob kept, so that it represents a byproduct
of entanglement between Alice and Bob that has been generated by the protocol. When
executed in reverse, the merging becomes splitting, and the $K_AK_B$ entanglement
becomes a resource that is consumed, along with the quantum transmission of system
$M$ from Alice to Bob, in order to implement the state-splitting redistribution
\be \Psi^{R:A} \ra \Psi^{R:EB}\ra \Psi^{R:E:B}.\ee

\subsection{Tensor power inputs}
We next need to reduce the general channel simulation problem to the
problem of simulating flat channels.  To get the idea of how this
works, consider first the problem of simulating $\cN^{\ot n}$ on a
tensor power input $\rho^{\otimes n}$.  While several solutions to
this problem have been previously described in
\cite{HOW05,devetak-triangle,ADHW06} and this section is not strictly
necessary for the proof of \thmref{qrst}, we will present a protocol
for tensor power inputs here in a way that will help us understand the
general case.

Let $\ket{\sigma}^{ABE} = (I \ot \cN)\ket{\Phi_\rho}^{AA'}$ and
$\ket{\psi}=\ket{\sigma}^{\ot n}$.  Unfortunately, none of $\psi^A$,
$\psi^B$ nor $\psi^E$ are in general maximally mixed.  Even
restricting to typical subspaces still leaves these states with
eigenvalues that vary over a range of $2^{\pm O(\sqrt{n})}$.

On the other hand, these eigenvalues have a large amount of
degeneracy.  Let $\{\ket{a_1},\ldots, \ket{a_{d_A}}\}$ be the
eigenbasis of $\rho=\sigma^A$.  Then the eigenvectors of $\psi^A$ can be taken
to be of the form $\ket{a_{i_1}} \ot \cdots \ot \ket{a_{i_n}}$, for
$i=(i_1,\ldots,i_n) \in [d_A]^n$.  Moreover the corresponding
eigenvalue is determined entirely by the {\em type} $t_A$ of $i$, just
as in the classical case.  There are
$\binom{n+d_A-1}{n}$ such types.  For fixed $d_A$, this number is
polynomial in $n$, and thus the ``which type'' information can be
transmitted using $O(\log n)$ qubits.  Conditioned on this
information, we are left with a flat spectrum over a space whose
dimension depends on the type.

The same decomposition into types can be performed for the $B$ and $E$
systems, and for constant $d_B$ and $d_E$ we will still have at most
$\poly(n)$ types $t_B$ and $t_E$.  Furthermore, we can decompose the
action of $U_\cN^{\ot n}$ into a map from $t_A$ to a superposition of
$t_B$ and $t_E$ followed by a flat map within the type classes, which
we call $V_{t_At_Bt_E}$.  Thus, letting $\cong$ denote a global change
of basis, we have
\be U_\cN^{\ot n} \cong \sum_{t_A,t_B,t_E} \ket{t_B,t_E}\bra{t_A} \ot
V_{t_A,t_B,t_E}.\label{eq:q-type-decomp}\ee

The only remaining question is to determine the communication rate.
Here we can use the classical theory of types from \secref{types} to
argue that almost all of the weight of $\psi$ is concentrated in
strings with $\bar t_A, \bar t_B, \bar t_E$ close to the spectra of
$\sigma^A$, $\sigma^B$ and $\sigma^E$ respectively.  If ``close'' is
defined to be distance $\delta$, then ignoring the atypical types
incurs error at most $\exp(-n\delta')$ and we are left with
subspaces of dimensions $D_A = \exp(n(H(A)_\sigma \pm \delta''))$, $D_B
= \exp(n(H(B)_\sigma \pm \delta''))$ and $D_E = \exp(n(H(E)_\sigma \pm
\delta''))$, where $\delta',\delta''$ are constants depending on
$\delta$.\footnote{These claims are based on standard methods of
  information theory.  By ``distance'' $\delta$ we refer to the trace
  distance $\|\sigma^X - \bar t^X\|_1$.  The error bound on ignoring
  atypical types is obtained from \eq{p-typical} and the bound on
  entropy is from the Fannes-Audenaert inequality (\eq{fannes}).
}
depending on $\delta$), .  Applying \lemref{flat-qrst} we obtain the claimed
communication rates of $\frac{H(A)+H(B)-H(E)}{2}=\frac{1}{2}I(A;B)$
qubits and $\frac{H(B)+H(E)-H(A)}{2}=\frac{1}{2}I(B;E)$ ebits per use
of $\cN$.

Two subtleties arise from combining communication protocols involving
different input and output types.  The first problem is that we have
to be careful about who knows what when: unlike in the classical channel
simulation protocol, Alice would like to communicate to Bob only $t_B$
and not $t_A$ or $t_E$.  Indeed, she would like to forget $t_A$ and
retain only knowledge of $t_E$ for herself.  This is addressed by
using the fact that Bob's decoding unitary $S_H$ in \lemref{flat-qrst}  can be chosen to depend
only on $t_B$, since we can choose a single $S_H$ for each $t_B$, and
w.h.p. \eq{flat-splitting} holds simultaneously for all $t_A, t_E$.
Denote the resulting decoding map $S_{H,t_B}$ and call Alice's
encoding $S_{V_{t_A,t_B,t_E}}$.  Then from
\eqs{flat-splitting}{q-type-decomp}, we have that $U_\cN^{\ot n}$ can
be approximately simulated by starting with an appropriate entangled
state (more on this below) and applying
\begin{multline}
% U_\cN^{\ot n} \approx
\L(\sum_{t_B} \proj{t_B} \ot S_{H,t_B}\R)\times \\
\L(\sum_{t_A,t_B,t_E}\ket{t_B,t_E}\bra{t_A} \ot S_{V_{t_A,t_B,t_E}}\R),
\end{multline}
where the first line is applied by Bob and the second line is applied
by Alice.

The second problem is that when we
apply \lemref{flat-qrst} to the $V_{t_A,t_B,t_E}$, the dimensions
$D_A, D_B, D_E$ (and thus $D_K,D_M$ as well) vary by as much as $\exp(\pm n\delta)$, and yet our
protocol needs to act on a single entangled state and send a single message.
 For the $M$ register we can address this by
simply taking $D_M$ to equal $\max \lceil \frac{256}{\eps^4} |T_{t_A}| |T_{t_B}| /
|T_{t_E}| \rceil$, where the maximum is taken over all typical triples
of $t_A, t_B, t_E$.  Thus, $D_M$ is independent of any of the
registers communicated during the protocol.

However, since $D_RD_M$ must equal $|T_{t_B}|$, we cannot avoid having
$D_R$ vary with $t_B$.  (There is a minor technical point related to
$D_B$ needing to be an integer, but this can be ignored at the
cost of an exponentially small error.)  As a result, we need to run
the protocol of \lemref{flat-qrst} in superposition using different
numbers of ebits in different branches of the superposition.  This
cannot be accomplished simply by discarding the unnecessary ebits
in the branches of the superposition that need less entanglement;
instead we need to use one of the techniques from \secref{spread}.
Fortunately, since the number of ebits varies by only $O(n\delta)$
across different values of $t_B$, we only need to generate
$O(n\delta)$ bits of entanglement spread.  This can be done with
$O(n\delta)$ extra qubits of communication, leading to an
asymptotically vanishing increase in the communication rate.  And
since the amount of entanglement generated depends on only $t_B$, this
does not require leaking any information to Bob that he will not
already have.  Alice first creates her half of the entanglement at the
same time as she is transformating $\ket{t_A}$ into $\ket{t_B,t_E}$.
Then she sends her half of the entanglement to Bob (after it has been
mixed with the input by $S_{V_{t_A,t_B,t_E}}$) along with her copy of
$\ket{t_B}$.  This ensures that Alice keeps no record of the amount of
entanglement she has created, while Bob is able to perform his part of
the entanglement-generation protocol.

Earlier versions of the quantum reverse Shannon theorem did not need to
mention this sublinear amount of entanglement spread because the extra
sublinear communication cost could be handled automatically by the protocols
used.  However, when we consider non-tensor power inputs in
\secref{gen-input} we will need to make a more explicit accounting of
the costs of entanglement spread.  Thus, the reason our ``warm-up'' is
more complicated than the previous proofs of the i.i.d.-source QRST is
that it already contains much of the complexity of the full proof.

%Indeed, we could relax the assumptions of flat spectra in
%\lemref{flat-qrst} to require only that $\rank \psi^A \leq D_A$,
%$\rank \psi^B\leq D_B$ and $\|\psi^E\| \leq 1/D_E$.

%Here the eigenvectors of $\psi^B$ and $\psi^E$ are built out of tensor
%products of the eigenbases of $\sigma^B$ and $\sigma^E$ respectively.
%Denote these eigenbases $\{\ket{b_1},\ldots,\ket{b_{d_B}}\}$ and
%$\{\ket{e_1},\ldots,\ket{e_{d_E}}\}$.

\subsection{A quantum theory of types}
\label{sec:q-types}

There is one further difficulty which arises when considering
non-tensor power inputs.  This problem can already been seen in the
case when the input to the channel is of the form
$\frac{1}{2}(\rho_1^{\ot n} + \rho_2^{\ot n})$.  If $\rho_1$ and
$\rho_2$ do not commute, then we cannot run the protocol of the
previous section without first estimating the eigenbases of $\rho_1$
and $\rho_2$.  Moreover, we need to perform this estimation in a
non-destructive way and then be able to uncompute our estimates of the
eigenbasis, as well as any intermediate calculations used.  Such
techniques have been used to perform quantum data compression of
$\rho^{\ot n}$ when $\rho$ is unknown~\cite{JP03, BHL06}.  However,
even for that much simpler problem they require delicate analysis.  We
believe that it is possible to prove the quantum reverse Shannon
theorem by carefully using state estimation in this manner, but
instead will present
a somewhat simpler proof that makes use of representation theory.

The area of representation theory we will use is known as Schur
duality (or Schur-Weyl duality).  It has also been used for data
compression of unknown tensor power
states~\cite{Hayashi:02b,Hayashi:02c,Hayashi:02e} and entanglement
concentration from tensor powers of unknown pure entangled
states~\cite{Hayashi:01b, Hayashi:02a}.  Some reviews of the role of
Schur duality in quantum information can be found in Chapters 5 and 6
of \cite{Har05} and Chapters 1 and 2 of \cite{matthias}.  A detailed
explanation of the mathematics behind Schur duality can also be found in
\cite{GW98}.  Our treatment will follow \cite{Har05}.  In
\secref{schur-basics}, we will explain how Schur duality can serve as
a quantum analogue of the classical method of types that we described
in \secref{types}.  Then in \secref{schur-channel} we will show this
can be applied to channels, allowing us to decompose $\cN^{\ot n}$
into a superposition of  flat channels. Finally, in \secref{schur-typical} we will use this
to describe quantum analogues of conditional types.  We will
use this to show that the atypical flat sub-channels involve only an
exponentially small amount of amplitude.

In \secref{gen-input}, we will use these tools to prove \thmref{qrst}.

\subsubsection{Schur duality and quantum states}
\label{sec:schur-basics}
This section will review the basics of Schur duality and will explain how
it can serve as a quantum analogue of the classical method of types.
Let $\cS_n$ denote the permutation group on $n$ objects and let
$\cU_d$ denote the $d$-dimensional unitary group.  Both groups have a
natural action on $(\bbC^d)^{\ot n}$.  For $u\in\cU_d$ define $\bQ(u)
= u^{\ot n}$ and for $s\in\cS_n$ define $\bP(s)$ to permute the $n$
systems according to $s$: namely,
\be \bP(s) = \sum_{i_1,\ldots,i_n \in [d]}
\ket{i_1,\ldots,i_n}\bra{i_{s(1)},\ldots,i_{s(n)}}.
\label{eq:perm-rep}\ee
 This convention
is chosen so that $\cP(s)$ is a representation\footnote{The product of
  permutations $s_1,s_2$ is defined by $(s_1 \cdot s_2)(i) =
  s_1(s_2(i))$.  Our definition in \eq{perm-rep} is chosen so that
  $\bP(s_1 \cdot s_2) = \bP(s_1)\bP(s_2)$.}
  These two
representations commute, and can be simultaneously decomposed into
irreducible representations (a.k.a. irreps).  We can also think of
$\bQ(u)\bP(s)$ as a reducible representation of $\cU_d \times \cS_n$.

Define $\cI_{d,n}$ to be the set of partitions of $n$ into $d$ parts:
that is
$\cI_{d,n} = \{\lambda=(\lambda_1,
\lambda_2,\dots,\lambda_d)\in\bbZ^d  :
\lambda_1 \geq \lambda _2 \geq
\cdots \geq \lambda_d \geq 0 \text{ and } \sum_{i=1}^d \lambda_i =
n\}.$
Note that $|\cI_{d,n}| \leq |\cT_{d,n}| \leq (n+1)^d =
\poly(n)$.  It turns out that $\cI_{d,n}$ labels the irreps of both
$\cU_d$ and $\cS_n$ that appear in the decompositions of $\bQ$ and
$\bP$.  Define these representation spaces to be $\cQ_\lambda$ and
$\cP_\lambda$ and define the corresponding representation matrices to
be $\bq_\lambda(u)$ and $\bp_\lambda(s)$.  Sometimes we write
$\cQ_\lambda^d$ or $\bq_\lambda^d$ to emphasize the $d$-dependence; no
such label is needed for $\cP_\lambda$ since $\lambda$ already
determines $n$.

Schur duality states that $(\bbC^d)^{\ot n}$ decomposes under the
simultaneous actions of $\bQ$ and $\bP$ as
\be (\bbC^d)^{\ot n} \cong \bigoplus_{\lambda\in\cI_{d,n}}
\cQ_\lambda^d \ot \cP_\lambda \label{eq:sch-space}\ee
This means that we can decompose $(\bbC^d)^{\ot n}$ into three
registers: an irrep label $\lambda$ which determines the actions of
$\cU_d$ and $\cS_n$, a $\cU_d$-irrep $\cQ_\lambda$ and an
$\cS_n$-irrep $\cP_\lambda$.  Since the dimension of $\cQ_\lambda$ and
$\cP_\lambda$ depends on $\lambda$, the registers are not in a strict
tensor product.  However, by padding the $\cQ_\lambda$ and
$\cP_\lambda$ registers we can treat the $\lambda$, $\cQ_\lambda$ and
$\cP_\lambda$ registers as being in a tensor product.

The isomorphism in \eq{sch-space} implies the existence of a unitary
transform $\Usch$ that maps $(\bbC^d)^{\ot n}$ to
$\bigoplus_{\lambda\in\cI_{d,n}}
\cQ_\lambda^d \ot \cP_\lambda$ in a way that commutes with the action
of $\cU_d$ and $\cS_n$.  Specifically we have that for any $u\in
\cU_d$ and any $s\in \cS_n$,
\be \Usch \bQ(U)\bP(s) \Usch^\dag =
\sum_{\lambda\in\cI_{d,n}}
\proj{\lambda} \ot \bq_\lambda^d(U) \ot \bp_\lambda(s).
\label{eq:sch-action}\ee

While we have described Schur duality in terms of the representation
theory of $\cS_n$ and the Lie group $\cU_d$, there exists a similar
relation between $\cS_n$ and the general linear group $GL_d$.  Indeed,
$\bq_\lambda(U)$ is a polynomial function of the entries of $U$ (of
degree $\sum_i \lambda_i$), and so can be extended to non-unitary and
even non-invertible arguments.  After doing so, one can show an analogue
of \eq{sch-action} for tensor power states (taking $U=\rho$ and $s=\id$)
\be \Usch \rho^{\ot n} \Usch^\dag
= \sum_{\lambda\in\cI_{d,n}}
\proj{\lambda} \ot \bq_\lambda^d(\rho) \ot I_{\cP_\lambda}.
\label{eq:rho-decomp}\ee

So far we have not had to describe in detail the structure of the
irreps of $\cU_d$ and $\cS_n$.  In fact, we will mostly not need to do
this in order to develop quantum analogues of the classical results
from \secref{types}.  Here, the correct analogue of a classical type
is in fact $\lambda$ together with $\cQ_\lambda$.  Classically, we
might imagine dividing a type $(t_1,\ldots,t_d)$ into a sorted list
$t_1^\downarrow \geq \cdots \geq t^\downarrow_d$ (analogous to
$\lambda$) and the $\cS_d$ permutation that maps $t^\downarrow$ into
$t$ (analogous to the $\cQ_\lambda$ register).  Quantumly, we will see
that for states of the form $\rho^{\ot n}$, the $\lambda$ register
carries information about the eigenvalues of $\rho$ and the
$\cQ_\lambda$ register is determined by the eigenbasis of $\rho$.

The main thing we will need to know about $\cQ_\lambda$ and
$\cP_\lambda$ is their dimension.  Roughly speaking, if $d$ is
constant then $|\cI_{d,n}|\leq \binom{n+d-1}{n}\leq \poly(n)$,
$\dim\cQ_\lambda \leq
\poly(n)$ and $\dim\cP_\lambda \approx\exp(n H(\bar\lambda))$.  For
completeness, we also state exact formulas for the dimensions of
$\cQ_\lambda$ and $\cP_\lambda$, although we will not need to use
them.  For $\lambda\in\cI_{d,n}$, define $\tilde{\lambda} := \lambda +
(d-1,d-2,\ldots,1,0)$.
Then the dimensions of $\cQ_\lambda^d$ and $\cP_\lambda$ are given
by~\cite{GW98}
\bea \dim \cQ_\lambda^d &=&\frac{\prod_{1\leq i<j\leq d}
(\tilde{\lambda}_i - \tilde{\lambda}_j)}{\prod_{m=1}^d m!}
\label{eq:cQ-dim}\\
\dim \cP_\lambda&=& \frac{n!}{\tilde{\lambda}_1! \tilde{\lambda}_2!
\cdots \tilde{\lambda}_d!}
\prod_{1\leq i<j\leq d}
(\tilde{\lambda}_i - \tilde{\lambda}_j)
\label{eq:cP-dim}\eea
It is straightforward to bound these by~\cite{Hayashi:02e,CM04}
\ba  & \dim \cQ_\lambda^d \leq (n+d)^{d(d-1)/2}
\label{eq:cQ-bound}\\
\binom{n}{\lambda} (n+d)^{-d(d-1)/2}  \leq &
\dim \cP_\lambda \leq  \binom{n}{\lambda}.
\label{eq:cP-bound1}\ea
Applying \eq{binom-bounds} to \eq{cP-bound1} yields the more useful
\bea \frac{\exp\L(nH(\bl)\R)}{(n+d)^{d(d+1)/2}} \leq \dim\cP_\lambda
\leq \exp\L(nH(\bl)\R).\label{eq:cP-bound2}\eea

To relate this to quantum states, let $\Pi_\lambda$ denote the
projector onto $\cQ_\lambda^d\ot \cP_\lambda \subset (\bbC^d)^{\ot
n}$.  Explicitly $\Pi_\lambda$ is given by
\be \Pi_\lambda = \Usch^\dag\L(\proj{\lambda}\ot I_{\cQ_\lambda^d} \ot
I_{\cP_\lambda}\R) \Usch.
\label{eq:Pi-lambda-def}\ee
From the bounds on $\dim\cQ_\lambda^d$ and
$\dim\cP_\lambda$ in \eqs{cQ-bound}{cP-bound2}, we obtain
\be \frac{\exp\L(nH(\bl)\R)}{(n+d)^{d(d+1)/2}} \leq \tr\Pi_\lambda
\leq \exp\L(nH(\bl)\R)(n+d)^{d(d-1)/2}
\label{eq:Pi-lambda-dim}\ee
As in the classical case, i.i.d.~states have a sharply peaked
distribution of $\lambda$ values.  Let $r=(r_1,\ldots,r_d)$ be the
eigenvalues of a state $\rho$, arranged such that $r_1\geq r_2\geq
\ldots$.  For $\mu\in\bbZ^d$, define $r^\mu = r_1^{\mu_1} \cdots
r_d^{\mu_d}$.  As explained in Section 6.2 of \cite{Har05}, one can
bound $\tr\Pi_\lambda \rho^{\ot n} = \tr \bq_\lambda^d(\rho)\cdot
\dim\cP_\lambda$ by
 \bmu \exp\L(-nD(\bl\|r)\R)(n+d)^{-d(d+1)/2} \\
\leq \tr\Pi_\lambda \rho^{\ot n} \\ \leq
\exp\L(-nD(\bl\|r)\R)(n+d)^{d(d-1)/2}
\label{eq:schur-proj}\emu
Similarly, we have
$\Pi_\lambda\rho^{\ot n} = \rho^{\ot n}\Pi_\lambda
= \Pi_\lambda \rho^{\ot n}\Pi_\lambda$ and
\be  \Pi_\lambda \rho^{\ot n}\Pi_\lambda \leq r^\lambda \Pi_\lambda
= \exp[-n(H(\bl) + D(\bl\|r))] \Pi_\lambda.\ee
 For some values of $\mu$, $r^\mu$ can be much smaller, so we cannot
express any useful lower bound on the eigenvalues of $\Pi_\lambda \rho^{\ot
n}\Pi_\lambda$, like we can with classical types. Of course, tracing
out $\cQ_\lambda^d$ gives us a maximally mixed state in $\cP_\lambda$,
and this is the quantum analogue of the fact that $p^{\ot n}(\cdot | t)$
is uniformly distributed over $T_t$.

We can also define the typical projector
\be \Pi_{r,\delta}^n =
\sum_{\lambda : \|\bl-r\|_1\leq \delta} \Pi_\lambda
%\\ =\Usch^\dag \L[\sum_{\lambda : \|\bl-r\|_1\leq \delta}
%\proj{\lambda} \ot I_{\cQ_\lambda^d} \ot I_{\cP_\lambda}\R] \Usch.
\ee
%where $\cB_{\delta}(r) := \{\bl : \|\bl-r\|_1\leq \delta\}$.
Using Pinsker's inequality, we find that
\be \tr \Pi_{r,\delta}^n \rho^{\ot n} \geq
1-\exp\L(-\frac{n\delta^2}{2}\R)(n+d)^{d(d+1)/2},
\label{eq:typ-proj}\ee
similar to the classical case.
The typical subspace is defined to be the support of the typical
projector.  Its dimension can be bounded (using
\eqs{fannes}{schur-proj}) by
\bmu \tr \Pi_{r,\delta}^n
\leq |\cI_{d,n}| \max_{\lambda: \|\bl-r\|_1\leq \delta}\tr\Pi_\lambda
\\ \leq (n+d)^{d(d+1)/2}
\exp(nH(r) + \eta(\delta) + n\delta\log d).\emu

\subsubsection{Decomposition of memoryless quantum channels}
\label{sec:schur-channel}

The point of introducing the Schur formalism is to decompose $\cN^{\ot
  n}$ (or more accurately, its isometric extension $U_\cN^{\ot n}$)
into a superposition of flat sub-channels.  This is accomplished by
splitting $A^n, B^n$ and $E^n$ each into $\lambda, \cQ_\lambda$ and
$\cP_\lambda$ subsystems labelled $\lambda_A, \cQ_{\lambda_A},
\cP_{\lambda_A}, \lambda_B$, etc.  Then the map from
$\cP_{\lambda_A} \ra \cP_{\lambda_B} \ot \cP_{\lambda_E}$ commutes
with the action of $\cS_n$ and as a result has the desired property of
being flat.

To prove this more rigorously, a general isometry from $A^n\ra B^nE^n$
can be written as a sum of terms of the form $\ket{\lambda_B,
  \lambda_E}\bra{\lambda_A} \ot \ket{q_B,q_E}\bra{q_A} \ot
P_{\lambda_A,\lambda_B,\lambda_E},$ where $\ket{q_A},\ket{q_B},
\ket{q_E}$ are basis states for the respective $\cQ_\lambda$ registers
and $P_{\lambda_A,\lambda_B,\lambda_E}$ is a map from $\cP_{\lambda_A}
\ra \cP_{\lambda_B} \ot \cP_{\lambda_E}$.

Since $U_\cN^{\ot n}$ commutes with the action of $\cS_n$, it follows
that each $P_{\lambda_A,\lambda_B,\lambda_E}$ must also commute with
the action of $\cS_n$.  Specifically, for any
$\lambda_A,\lambda_B,\lambda_E\in\cI_{d,n}$ (with
$d=\max(d_A,d_B,d_E)$) and
any $s\in\cS_n$, we have
\begin{multline*}
(\bp_{\lambda_B}(s) \ot \bp_{\lambda_E}(s))\,
P_{\lambda_A,\lambda_B,\lambda_E} \, \bp_{\lambda_A}(s)=
P_{\lambda_A,\lambda_B,\lambda_E}\\
\bp_{\lambda_A}(s)^\dag P_{\lambda_A,\lambda_B,\lambda_E}^\dag
P_{\lambda_A,\lambda_B,\lambda_E}  \bp_{\lambda_A}(s)=
P_{\lambda_A,\lambda_B,\lambda_E}^\dag
P_{\lambda_A,\lambda_B,\lambda_E}
\end{multline*}
By
Schur's Lemma
$P_{\lambda_A,\lambda_B,\lambda_E}^\dag
P_{\lambda_A,\lambda_B,\lambda_E}$ is proportional to the identity on
$\cP_{\lambda_A}$.  Therefore $P_{\lambda_A,\lambda_B,\lambda_E}$ is
proportional to an isometry.  Furthermore,
$P_{\lambda_A,\lambda_B,\lambda_E}$ maps the maximally mixed state
on $\cP_{\lambda_A}$ to a state proportional to
$P_{\lambda_A,\lambda_B,\lambda_E}
P_{\lambda_A,\lambda_B,\lambda_E}^\dag$.  This state commutes with
$\bp_{\lambda_B}(s) \ot \bp_{\lambda_E}(s)$ for all $s\in\cS_n$, and
so, if we again use Schur's Lemma, we find that the reduced states on
$\cP_{\lambda_B}$  and $\cP_{\lambda_E}$  are both maximally mixed.
Therefore $P_{\lambda_A,\lambda_B,\lambda_E}$ is proportional to a
flat isometry.

This is an example of a broader phenomenon.   For vector spaces $V_1,V_2$,
define $\Hom(V_1,V_2)$ to be
the space of linear maps from $V_1$ to $V_2$. Note that $\Hom(V_1,V_2) \cong
V_1^* \ot V_2$, and if $(\br_1, V_1), (\br_2,V_2)$ are representations
of a group $G$, then there is a representation $\br$ of $G$ on
$\Hom(V_1,V_2)$ given by $\br(g)T = \br_2(g) T \br_1(g^{-1})$.  For a
representation $(\br, V)$ the $G$-invariant subspace $V^G$ is defined
by
$$V^G := \{\ket\psi \in V : \br(g)\ket\psi = \ket\psi \; \forall g\in
G\}.$$
The space $\Hom(V_1,V_2)^G$ is precisely the set of linear operators
from $V_1$ to $V_2$ that commute with the action of $G$.  Using this
notation, Schur's Lemma is equivalent to the statement that if
$V_1,V_2$ are irreducible then $\Hom(V_1,V_2)^G$ is equal to $\{0\}$
if $V_1,V_2$ are inequivalent and is one-dimensional if $V_1,V_2$ are equivalent.

Using this language, we can observe that
$P_{\lambda_A,\lambda_B,\lambda_E}$ belongs to $\Hom(\cP_{\lambda_A},
\cP_{\lambda_B} \ot \cP_{\lambda_E})^{\cS_n}$, i.e. the set of
maps from $\cP_{\lambda_A}$ to $\cP_{\lambda_B} \ot \cP_{\lambda_E}$
that commute with $\cS_n$.  By the arguments in the paragraph before last, any
isometry in $\Hom(\cP_{\lambda_A},
\cP_{\lambda_B} \ot \cP_{\lambda_E})^{\cS_n}$ must also be a flat isometry.
There is a natural isomorphism from
$(\cP_{\lambda_A}^* \ot \cP_{\lambda_B} \ot \cP_{\lambda_E})^{\cS_n}$
into $\Hom(\cP_{\lambda_A},\cP_{\lambda_B} \ot
\cP_{\lambda_E})^{\cS_n}$.  We denote this isomorphism by $S$ (making
the $\lambda_A,\lambda_B,\lambda_E$-dependence implicit) and
normalize $S$ so that if $\ket\mu \in (\cP_{\lambda_A}^* \ot
\cP_{\lambda_B} \ot \cP_{\lambda_E})^{\cS_n}$ is a unit vector then
$S \ket\mu$ is a (flat) isometry.   Below we will offer an operational
interpretation of the $\ket\mu$ register.

To deal with the large numbers of registers, we now introduce some more
concise notation.
\begin{definition}\label{def:concise}
  Let $P_{\lambda_A}^{\lambda_B, \lambda_E} $ be an orthonormal basis
  for $(\cP_{\lambda_A}^*\ot \cP_{\lambda_B} \ot
  \cP_{\lambda_E})^{\cS_n}$.  We also let $T_A$ denote the set of
  pairs $(\lambda_A, q_A)$, where $\ket{q_A}$ runs over some fixed
  orthonormal basis of $\cQ_{\lambda_A}$, and similarly we define
  $T_B$ and $T_E$.
\end{definition}
Now we can represent $U_\cN^{\ot n}$ as
\ba U_\cN^{\ot n} &= \sum_{\substack{ \tau_A\in T_A \\
     \tau_B\in T_B \\ \tau_E\in T_E}}\sum_{\mu\in P_{\lambda_A}^{\lambda_B,
    \lambda_E}}
%P[\lambda_A;\lambda_B, \lambda_E]}
 \!\!\!\! [V_\cN^n]_{\tau_B,  \tau_E, \mu}^{\tau_A}
\ket{ \tau_B,\tau_E}\bra{\tau_A} \ot S\ket\mu
\label{eq:Vn-decomp}
\ea
This is depicted as a quantum circuit in \fig{normal-form}.

(We will not need to know anything more about the
representation-theoretic structure of  $P_{\lambda_A,\lambda_B,
  \lambda_E}$, but the interested reader can find a more detailed
description of this decomposition of $U_\cN^{\ot n}$ in Section 6.4 of
\cite{Har05}, where $S\ket\mu$ is related to the Clebsch-Gordan
transform over $\cS_n$.)

\begin{figure}[ht]
\centerline{\Qcircuit @=1em {
\lstick{\ket{ \tau_A}} & \multigate{2}{[V_\cN^n]} & \qw & \control\qwx[2]\qw
 & \rstick{\ket{ \tau_B}} \qw\\
 & \pureghost{[V_\cN^n]} & \qw & \control\qw &  \rstick{\ket{ \tau_E}} \qw\\
 & \pureghost{[V_\cN^n]} & \qw_<{\ket{\mu}} & \multigate{1}{S} &
 \rstick{\ket{p_B}} \qw\\
\lstick{\ket{p_A}} & \qw & \qw  & \ghost{S}
& \rstick{\ket{p_E}} \qw\\
}}
\caption{The quantum channel $U_\cN^{\ot n}$ is decomposed in the
Schur basis as in \eq{Vn-decomp}. Alice inputs an $n$ qudit state of
the form $\ket{ \tau_A}\ket{p_A}$ and the channel outputs
superpositions of $\ket{ \tau_B}\ket{p_B}$ for Bob and
$\ket{ \tau_E}\ket{p_E}$ for Eve.  The intermediate state
$\ket{\mu}$ belongs to
$(\cP_{\lambda_A}^* \ot \cP_{\lambda_B}\ot\cP_{\lambda_E})^{\cS_n}$.
The figure suppresses the implicit $U_\cN,n$-dependence of $S$, and
expresses the $\lambda_A, \lambda_B, \lambda_E$-dependence of $S$ by
the
control wires from the $\ket{\tau_B}$ and $\ket{\tau_E}$ registers.}
\label{fig:normal-form}
\end{figure}

We now have a situation largely parallel to the classical theory of
joint types with $\tau_A, \tau_B,\tau_E$ representing the quantum
analogues of types for systems $A$, $B$ and $E$.  Since $\tau_B, \tau_E,
\mu$ together describe the joint type of systems $BE$, we can think of
$\mu$ as representing the purely joint part of the type that is not
contained in either of the marginal types.  Further justifying the
analogy with classical types is the fact that all but $\poly(n)$
dimensions are described by the flat
isometries $P_{\lambda_A,\lambda_B,\lambda_E}$.  Next we need to
describe an analogue of jointly typical projectors, so that we can restrict our
attention to triples of $(\lambda_A, \lambda_B, \lambda_E)$ that
contribute non-negligible amounts of amplitude to $U_\cN^{\ot n}$.  In
the next section, we will argue that $[V_\cN^n]^{\tau_A}_{\tau_B, \tau_E,\mu}$ is
exponentially small unless $(\bl_A,\bl_B,\bl_E)$ correspond to the
possible spectra of marginals of some state $\psi^{RBE}$ that is
obtained by applying $U_\cN$ to a pure state on $RA$.

\subsubsection{Jointly typical projectors in the Schur basis}
\label{sec:schur-typical}

 In order for \eq{Vn-decomp} to be useful, we need to
control the possible triples $(\tau_A,\tau_B,\tau_E)$ that can have
non-negligible weight in the sum.  In fact, it will suffice to bound
which triples $(\lambda_A, \lambda_B, \lambda_E)$ appear, since these
determine the dimensions of the $\cP_{\lambda}$ registers and in turn
determine the dominant part of the communication cost.  For large
values of $n$, almost all of the weight will be contained in a small
set of {\em typical} triples of $(\lambda_A,\lambda_B,\lambda_E)$.
These triples are the quantum analogue of joint types from classical
information theory.

Let $\rho^A$ be an arbitrary channel input, and $\ket{\psi}^{ABE}=(I^A
\ot U_\cN^{A'\ra BE})\ket{\Phi_\rho}^{AA'}$ the purified channel
output.  Now define $R(\cN)$ to be set of $\psi^{ABE}$ that can be
generated in this manner.  Further define $\cT_\cN^*$ to be
$\{(r_A,r_B,r_E) : \exists \psi^{ABE}\in R(\cN) \text{~s.t.~}\!
r_A\!=\spec(\psi^A), r_B\!=\spec(\psi^B), r_E\!=\spec(\psi^E)\}$.
This set is simply the set of triples of spectra that can arise from
one use of the channel.  We will argue that it corresponds as well to
the set of $(\bl_A,\bl_B,\bl_E)$ onto which a channel's input and
output can be projected with little disturbance.
Let $T_{\cN,\delta}^n$ denote the set
\bmu \{ (\lambda_A, \lambda_B, \lambda_E) :
\exists (r_A, r_B, r_E) \in \cT_\cN^*, \\
\|\bl_A - r_A\|_1 + \|\bl_B - r_B\|_1 + \|\bl_E - r_E\|_1 \leq
\frac{\delta}{\log(d)}\}\emu
One difficulty in defining joint types is that applying the projector
$\Pi_{\lambda_A}$ to the input may not commute with applying
$\Pi_{\lambda_B}\ot \Pi_{\lambda_E}$ to the output.  Nevertheless,
the following lemma (first proven in Section 6.4.3 of \cite{Har05})
establishes a version of joint typicality that we can use.
\begin{lemma}[\cite{Har05}]\label{lem:joint-schur}
Let $d = \max(d_A, d_B, d_E)$.
For any state $\ket{\varphi}^{RA}$ with $\ket{\Psi} = (I \ot
U_\cN)^{\ot n}\ket\varphi^{\ot n}$,
\bmu \L\|\, \ket\Psi -\!\!\!
\sum_{(\lambda_A,\lambda_B,\lambda_E)\in T_{\cN,\delta}^n}
I\ot ((\Pi_{\lambda_B} \ot \Pi_{\lambda_E})U_\cN^{\ot n}
\Pi_{\lambda_A})\ket\varphi^{\ot n}\R\|
\\\leq n^{O(d^2)}\exp\L(-n\frac{\delta^2}{8\log^2(d)}\R).
\emu
\end{lemma}
For completeness, we include a proof in the appendix.

\subsection{Reduction to the flat spectrum case}
\label{sec:gen-input}

In this section we prove the coding theorem for the QRST.  The outline
of the proof is as follows:
\benum
\item We show that general inputs can be replaced by $\cS_n$-invariant
  inputs by using a sublinear amount of shared randomness (which can
  be obtained from any of the other resources used in the protocol).
\item We show that $\cS_n$-covariant channels (such as $\cN^{\ot n}$)
  decompose into a superposition of flat sub-channels.  This is based
  on \secref{schur-channel}.  The simulation of these flat
  sub-channels on maximally mixed inputs is described in \secref{flat-qrst}.
\item We show that atypical sub-channels can be ignored with
  negligible error (using \secref{schur-typical}).
\item We paste together simulations of different flat channels using
  entanglement spread (introduced in \secref{spread}).
\eenum

We now explain these components in more detail.  First, we show how it
is possible to assume without loss of generality that our inputs are
$\cS_n$-symmetric.  If we did not mind using a large amount of shared randomness, then using $\log(n!)$ rbits would allow Alice to apply a random permutation $\pi\in\cS_n$ to her inputs, and then for Alice to apply $\pi^{-1}$ to the Eve output and for Bob to apply $\pi^{-1}$ to his output.  In some scenarios, these shared rbits might be a free resource (e.g. when entanglement is unlimited), and their cost could be further reduced by observing that they are incoherently decoupled from the protocol (using the terminology of \cite{DHW05}), and thus can be safely reused.

However, in fact, it is possible for Alice and Bob to safely sample $\pi$ from a much smaller distribution.  The idea is that the protocol has $\eps$ error on an $\cS_n$-invariant input, which means that if the input is randomly permuted, then the average error will be $\eps$.  On the other hand, the diamond-norm error is never greater than 2.  Standard concentration-of-measure arguments can then be used to show that $O(\log (n/\eps))$ rbits suffice to reduce the error to $O(\eps)$.  This is detailed in \lemref{worst-to-avg-case}.

For the rest of this section, we simply assume that Alice is given
half of an $\cS_n$-invariant input $\ket{\varphi}^{R^nA^n}$.  Based on
\secref{schur-channel}, we can decompose the action of $U_\cN^{\ot n}$
into a map from $\tau_A$ to $\tau_B, \tau_E, \mu$ followed by a map from
$p_A,\mu$ to $p_B, p_E$.  The $\tau_B$ register has only $\poly(n)$
dimension, and can be transmitted uncompressed to Bob using $O(\log
n)$ qubits.  On the other hand, the map $P_\mu$ is flat, and therefore
can be compressed using \lemref{flat-qrst}.

To understand the costs of compressing $P_\mu$, we need to estimate
the dimensions of the $\cP_{\lambda_A},\cP_{\lambda_B},
\cP_{\lambda_E}$ registers.  In \secref{schur-basics}, we showed that
$\dim\cP_\lambda \approx \exp(nH(\bl))$ up to $\poly(n)$ factors.  So
the cost of simulating a flat map from $\cP_{\lambda_A}$ to
$\cP_{\lambda_B} \ot \cP_{\lambda_E}$ is $\frac{1}{2}n[H(\bl_A) + H(\bl_B)
- H(\bl_E)] + O(\log n)$ qubits and $\frac{1}{2}n[H(\bl_B) + H(\bl_E)
- H(\bl_A)]
+ O(\log n)$ ebits.

Next, we can relate these costs to entropic quantities.   Using
\lemref{joint-schur} from \secref{schur-typical}, it follows that we
need only consider the triples $(\bl_A, \bl_B, \bl_E)$ within
distance $\delta/\log(d)$ of a spectral triple $(r_A, r_B, r_E)$
corresponding to a possible channel output.  Therefore, the problem
of simulating $\cN_F^{\ot n}$ can be reduced to producing a
superposition of \bmu \L(\frac{1}{2}nI(R;B)_\rho + O(n\delta  + \log
n)\R)[q\ra q]
\\ +
\L(\frac{1}{2}nI(E;B)_\rho + O(n\delta+\log n)\R) [qq]
\label{eq:single-rho-cost}\emu
for all possible single-letter $\rho$ (i.e. $\rho$ that are inputs to
a single channel use).  If we take $\delta\ra
0$ as $n\ra \infty$ then this corresponds to an asymptotic rate of
\be \frac{1}{2}I(R;B)_\rho[q\ra q]  +
\frac{1}{2}I(E;B)_\rho [qq]
\label{eq:single-rho-asym}\ee
per channel use.
The resulting protocol is depicted in \fig{IIDQRST}.

\begin{figure*}[htbp]
\includegraphics[width=0.9\textwidth]{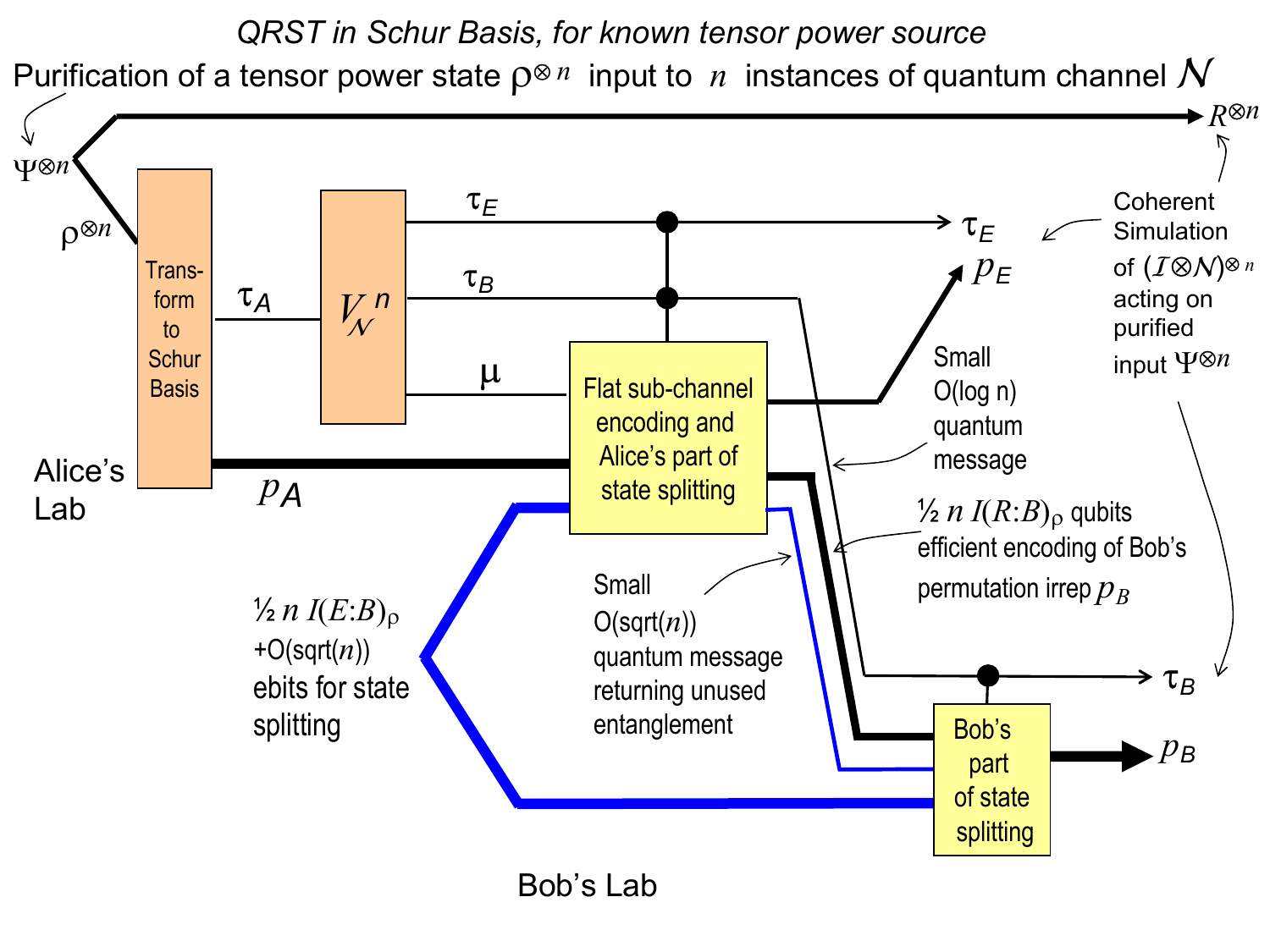}
\caption{Quantum protocol for quantum reverse Shannon theorem on a known tensor power source. Alice transforms the tensor
power input into the Schur representation, comprising a small $\tau$ register containing the quantum type
and a large $p$ register containing the permutation irrep. These registers, together with a slight
($O(\sqrt n)$) excess of halves of ebits shared with Bob, are coherently transformed into
about $\half n I(R;B)$ qubits worth of flat sub-channel codes representing Bob's $p$ register,
which Bob decodes with the help of the other halves of the shared ebits and the small $\tau_B$
register sent from Alice.  Alice also returns the ($O(\sqrt n)$) unused halves of ebits,
allowing them to be coherently destroyed. The remaining registers $\tau_E$ and $p_E$,
representing Eve's share of the output, remain with Alice, as required for a quantum feedback
simulation of the channel $\cN^{\ot n}$.  By discarding them into the environment, one obtains a (not necessarily efficient)
non-feedback simulation.}
\label{fig:IIDQRST}
\end{figure*}

Finally, if our input is not a known tensor power source, then
producing \eq{single-rho-cost} in superposition may require
entanglement spread.  Suppose that $\alpha \geq \beta$ and
$$\beta\; \geqclo\;  \frac{1}{2}I(R;B)_\rho[q \ra q] + \frac{1}{2}I(E;B)_\rho
[qq]$$ for all $\rho$.   Then we can prepare $\beta$ from $\alpha$
and then use $\beta$ to produce the resources needed to simulate
$\<\cN_F:\rho\>$ in superposition across all $\rho$ (or equivalently
across all $\tau_A$ in the input). This can be done using extra
forward communication (in which case the protocol still
qualitatively resembles \fig{IIDQRST}, but the $O(\sqrt{n})$ message
with extra entanglement becomes $\Omega(n)$ qubits), using an
embezzling state (as depicted in \fig{EmbezQRST}) or using backward
communication (as depicted in \fig{BackCommQRST}).  The protocol with
backwards communication appears to require a temporary shuttling of the small
$\tau_B$ register from Alice to Bob and back before finally sending
it to Bob; otherwise backward communication is used to coherently
reduce entanglement the same way that forward communication is.

\begin{figure*}[htbp]
\includegraphics[width=0.9\textwidth]{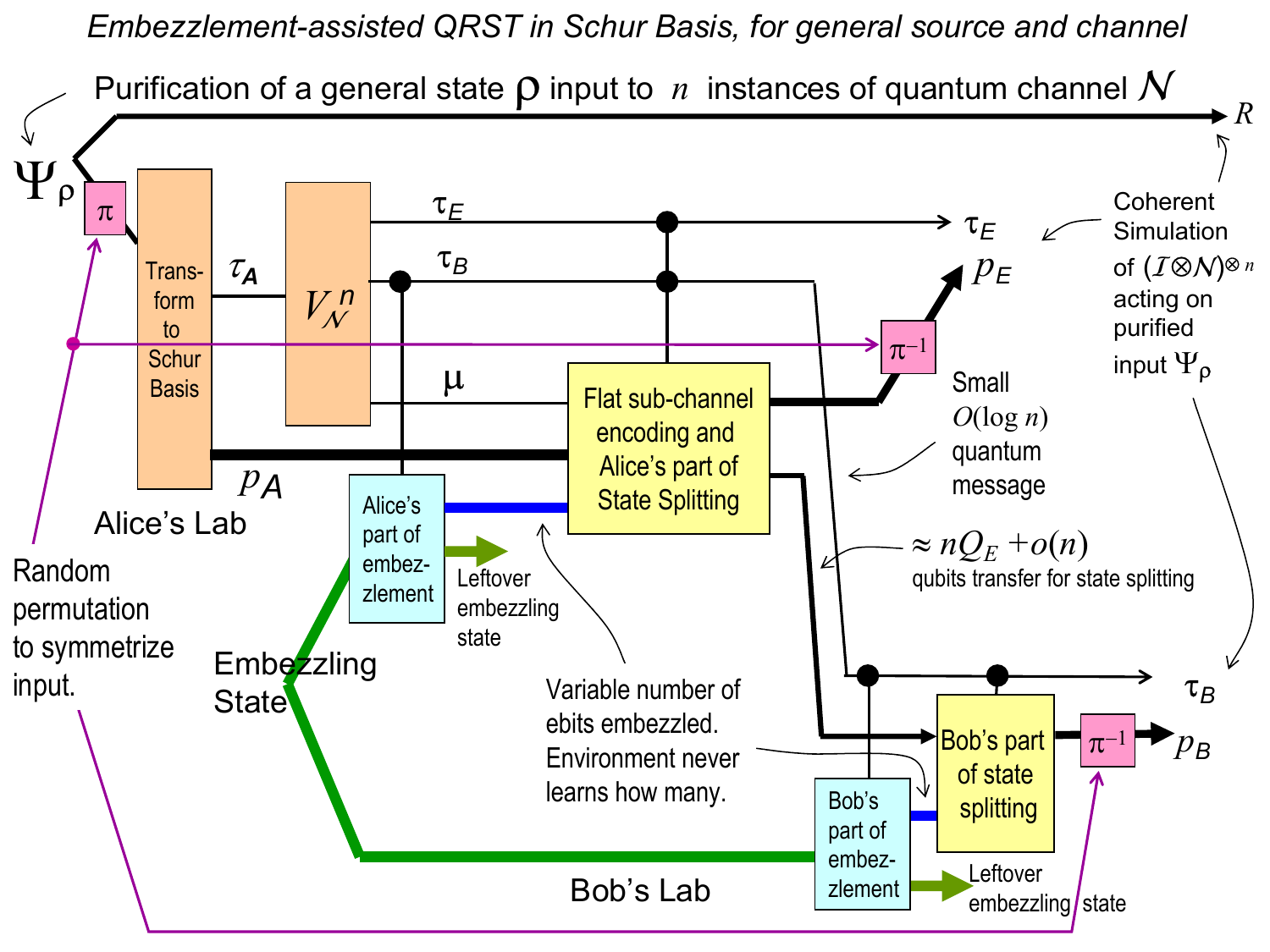}
\caption{QRST on a general input using an entanglement-embezzling
state (green). Alice first applies a randomizing permutation $\pi$
to the inputs to $n$ instances of her quantum channel, using
information shared with Bob (magenta), thereby rendering the overall
input approximately permutation-symmetric. She then uses the
$\tau_B$ register to embezzle the correct superposition of (possibly
very) different amounts of entanglement needed by her sub-channel
encoder, leaving a negligibly degraded embezzling state behind. At
the receiving end (lower right) Bob performs his half of the
embezzlement, coherently decodes the sub-channel codes, and undoes
the random permutation. The shared randomness needed for the initial
randomizing permutation can also be obtained from the embezzling
state, and in any case can be made sublinear in $n$, as shown in
Lemma \ref{lem:worst-to-avg-case}. } \label{fig:EmbezQRST}
\end{figure*}

\begin{figure*}[htbp]
\includegraphics[width=0.9\textwidth]{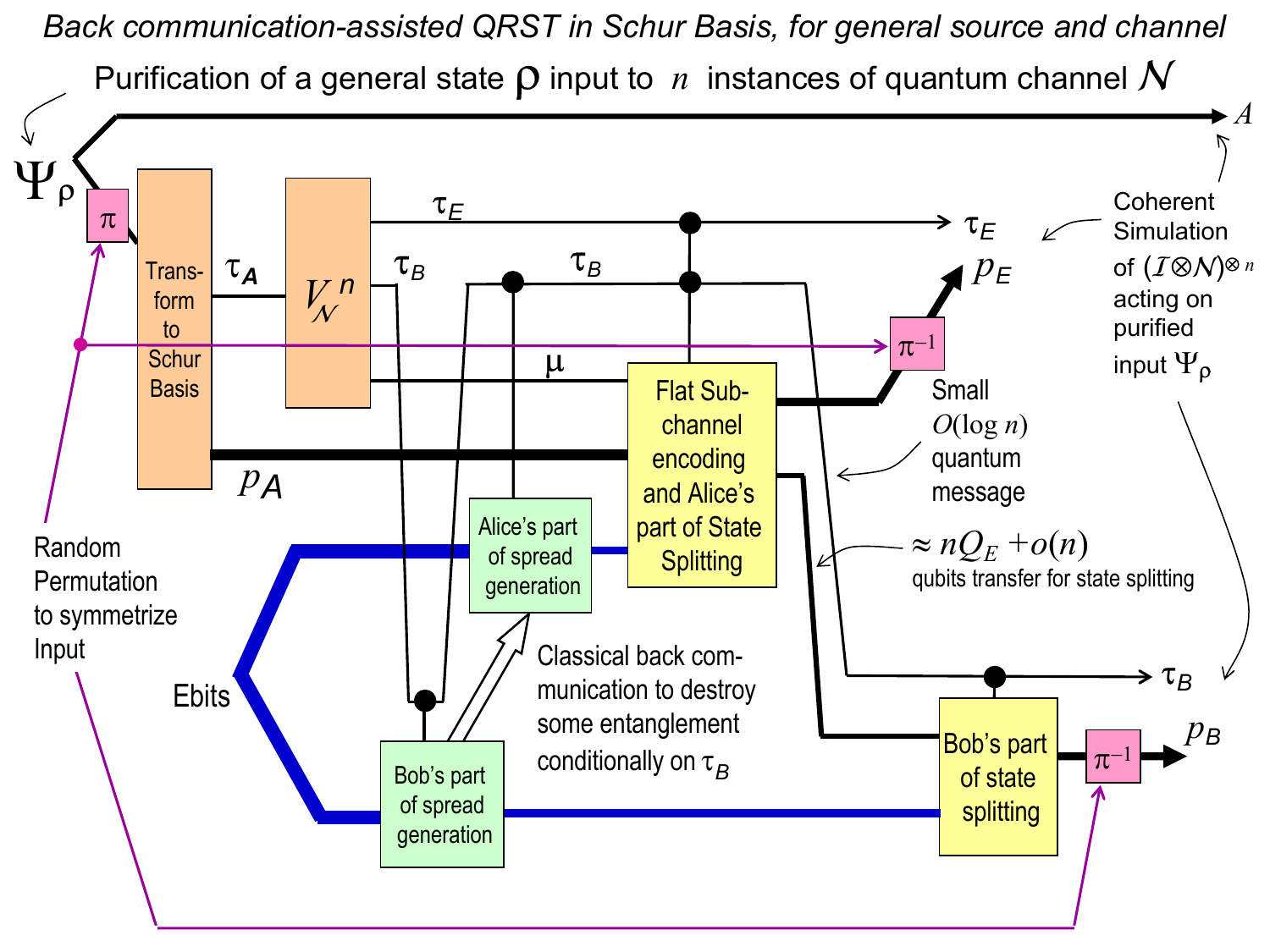}
\caption{QRST using classical back communication. Here the requisite spread is generated by
starting with a large amount of ordinary (i.e. maximal) entanglement, then using back communication and
the $\tau_B$ register to coherently burn off some of it. This requires the $\tau_B$
register to make a round trip from Alice to Bob then back again, before finally returning
to Bob, who needs to be holding it at the end.  Other aspects of the protocol are as in the
embezzlement-assisted implementation of the preceding figure.}
\label{fig:BackCommQRST}
\end{figure*}

When we do not need to simulate feedback, the main difference is that
we can split the $E$ register into a part for Alice ($E_A$) and a part
for Bob ($E_B$).  Additionally, this splitting is not restricted to be
i.i.d., although the corresponding ``additivity'' question here
remains open.
 That is, for any $n\geq 1$ and any $V:E^n\ra E_A E_B$,
simulating the action of $\cN^{\ot n}$ can be achieved by simulating
$V\circ \cN_F^{\ot n}$.  Here Alice gets the output $E_A$ and Bob gets
the output $B^n E_B$.  Moreover, we are in some cases able to break the
superpositions between different $\tau_A$.  If feedback is not required,
then we can assume without loss of generality that Alice has measured $\tau_E$, estimated
$\hat{\cN}(\rho)$ to within $O(n^{-1/2})$ accuracy~\cite{keyl06} and
communicated the resulting estimate to Bob using $o(n)$
communication.

However, in some cases (including an example we will describe in the
next section), $\hat{\cN}(\rho)$ does not uniquely determine
$\cN(\rho)$, and thereby determine the rate of entanglement needed.
In this case, it will suffice to prepare a superposition of entanglement
corresponding to any source in $(\hat{\cN}^{\ot n})^{-1}(\omega)$ for
each $\omega\in\text{range}(\hat{\cN}^{\ot n})$.  This yields the
communication cost claimed in \thmref{qrst}.

We conclude with a rigorous proof that low average-case error can be
turned into low worst-case error, allowing the permutation $\pi$ in
Figs. \ref{fig:EmbezQRST} and \ref{fig:BackCommQRST} to be largely
derandomized, reducing its shared randomness cost to sublinear in
$n$.
\begin{lemma}\label{lem:worst-to-avg-case}
Let $V^{A\ra BE}$ be an isometry that represents an ideal protocol and $\tilde V^{A \ra BE}$ its approximate realization.  Suppose that we have an average-case fidelity guarantee of the form
\be \bra{\Phi_{D_A}}^{RA}(I \ot \tilde V^\dag V)\ket{\Phi_{D_A}}^{RA} \geq 1-\eps
\label{eq:avg-case}.\ee
Let $\mu$ be a distribution over $\cU_{D_A}$ such that $\bbE_{U\sim \mu} U \rho U^\dag = I/D_A$ for any density matrix $\rho$.  If $U_1,\ldots,U_m$ are drawn i.i.d. from $\mu$, then with probability $\geq 1-D_A(4/e)^{-m\eps/2}$, for any $\ket{\psi}$,
\be \frac{1}{m}\sum_{i=1}^m |\bra{\psi}(I \ot U_i^\dag \tilde V^\dag V U_i)\ket{\psi}|^2 \geq 1-6\eps.
\label{eq:worst-case}\ee
\end{lemma}

\begin{proof}
Let $\Delta := I- \tilde V^\dag V$.  Observe that $\|\Delta\|_\infty\leq 2$ and that \eq{avg-case} implies that $\|\Delta\|_1 = \tr\Delta \leq \eps D_A$.  Define
\ba \Delta_0 & := \bbE_{U\sim \mu}[U\Delta U^\dag] = \frac{\tr\Delta}{D_A}I  \leq \eps I\\
\bar\Delta &:= \frac{1}{m}\sum_{i=1}^m U_i \Delta U_i^\dag ,\ea
where $U_1,\ldots,U_m$ are drawn i.i.d.~from $\mu$.
By applying the operator Chernoff bound~\cite{AW02,Tropp-LD}, we find that with probability $\geq 1 - D_A (4/e)^{-m\eps/2}$, we have $\|\bar\Delta-\Delta_0\|_\infty \leq 2\eps$.  In this case, $\|\bar\Delta\|_\infty \leq 3\eps$.   \eq{worst-case} follows by Cauchy-Schwarz.
\end{proof}

\lemref{worst-to-avg-case} can be applied separately to each Schur subspace, with $D_A := \dim\cP_{\lambda_a}$.  Thus, a union bound multiplies the probability of a bad choice of permutations by $n^{O(d_A)}$ and we always have $D_A \leq O(n\log d_A)$.

\subsubsection*{Inefficiencies and errors}
Here we briefly tabulate the various sources of inefficiency and error
in our simulation protocols for quantum channels.  We will consider
allowing an inefficiency of $O(n\delta)$ in each step of the protocol
and will analyze the resulting errors.

\benum
\item In \lemref{flat-qrst}, an extra communication rate of $O(n\delta)$
  means that the splitting step incurs an error of $\exp(-n\delta)$.
\item When restricting to typical triples, our definition of
  $T_{\cN,\delta}^n$ was chosen so that entropic quantities such as
  $H(R^n)+H(B^n)-H(E^n)$ would change by $\leq n\delta + o(n)$.  According
  to \lemref{joint-schur}, this results in error
  $\exp(-n\delta^2/8\log^2(d))$.  This will turn out to be the
  dominant error term; accordingly, we define
  $\delta'=\delta^2/8\log^2(d)$.
\item Suppose Alice and Bob use $\log(m)$ rbits to permute the channel inputs, and unpermute the channel outputs.  Then \lemref{worst-to-avg-case} implies that the error multiplies by only a constant factor if we take $m=O(n\delta + \log(n))$.
\item To achieve an error of $\exp(-n\delta)$ using an embezzling
  state, we need to take it to have $n\exp(n\delta)$ qubits.  This is
  exorbitant, but in some secenarios, such as our proof of the strong
  converse in the next section, the size of the entangled ancilla that
  we use is irrelevant.
\eenum

To summarize, our error scales as $\exp(-n\delta^2/8\log^2(d) + o(n))$.

\subsection{Converses and strong converses}
\label{sec:converses}
In information theory, a ``converse'' is the statement of asymptotic
resource optimality of a coding theorem (which is often called
``direct part''). A ``strong converse'' is a statement of the form
that with too little resources the error parameter in any protocol
approaches $1$ asymptotically.
In channel simulations, the first and foremost resource is forward
communication, but other resources of interest are the amount of
entanglement and specifically the amount of entanglement spread.

Here, we first show that as with the classical reverse Shannon theorem, the
existence of a coding theorem (this time for entanglement-assisted
capacity~\cite{BSST01}) means that no better simulation is possible.
%,at least not for IID inputs.
Indeed, such matching coding theorems
generally give us {\em strong} converses, implying that attempting to
simulate a channel at rate $C_E-\delta$ or lower results
in an error $\geq 1-\exp(-n\delta')$ for some
$\delta'>0$.  At the same time, they give us strong converses for
coding theorems, proving that attempting to code at a rate
$C_E+\delta$ results in the probability of successful decoding being
$\leq \exp(-n\delta')$, again for some $\delta'$ depending on $\delta$.
Second, we use arguments from the theory of entanglement spread
(cf. \secref{spread}) to show that our simulations for non-IID inputs
do require either embezzling states, or -- if only maximally
entangled states are available -- the use of extra communication
(which may be forward or backward directed), to create entanglement
spread.
%cannot be improved, even though in some cases they require
%communication rates that are higher than those of any known channel
%capacity.

\subsubsection{Strong converse for forward communication}
The general principle behind these strong converses is based on the
fact that $m$ forward cbits, assisted by arbitrary
back communication and entanglement,
can transmit $m+k$ bits only with success probability $\leq 2^{-k}$.
The proof is folklore.\footnote{Here is a sketch of the proof.  Suppose a protocol exists that achieves
  success probability $q$ on a randomly chosen $m+k$-bit input.
  Modify this protocol so that the bits transmitted from Alice to Bob
  are replaced by random bits that Bob generates locally.  We can
  think of this as Bob guessing Alice's input.  This modified protocol
  can be simulated locally by Bob and corresponds to him drawing from
  a fixed distribution independent of Alice's
  input.   On a random $m+k$-bit input, this must have success
  probability $2^{-m-k}$.  Our bound on the original protocol also
  means that this has success probability $\geq q
  2^{-m}$, since Bob has probability $2^{-m}$ of correctly guessing Alice's $m$
  transmitted bits.  Thus we obtain $q\leq 2^{-k}$.}
We call this principle {\em the guessing bound}; it is also sometimes
referred to as ``causality.''
To apply the guessing bound, suppose we
have a coding theorem that allows us to use $\cN^{\ot n}$ (perhaps
also with auxiliary resources, such as shared entanglement) to send
$n(C-\delta)$ bits with success probability $1-\eps_{n,\delta}$.
[Typically, $\eps_{n,\delta}$ will be of the form $\exp(-O(n\delta^2))$.]
Now assume that there exists a
simulation of $\cN^{\ot n}$ (using any auxiliary resources that are
not capable of forward communication) that uses $n(C-\delta')$ bits of
communication and achieves error $\eps_{n,\delta'}'$.
If $\delta'>\delta$, then the guessing bound implies that
\be
  \eps_{n,\delta} + \eps_{n,\delta'}' \geq 1-2^{-n(\delta'-\delta)}.
  \label{eq:mutual-converse}
\ee
Thus coding theorems constrain possible simulation theorems. Vice
versa, by the same logic, suppose we had a simulation of $\cN^{\ot n}$
using $n(C+\delta)$ cbits of forward communication plus additional
resources and error $\eps_{n,\delta}$, and consider a coding
of $n(C+\delta')$ cbits into $n$ uses of $\cN$ and auxiliary
resources, achieving error probability $\eps_{n,\delta'}'$.
Then as before, for $\delta'>\delta$,
\be
  \eps_{n,\delta} + \eps_{n,\delta'}' \geq 1-2^{-n(\delta'-\delta)}.
  \label{eq:mutual-converse-2}
\ee
For the purposes of this argument, any auxiliary resources are
permitted as long as they are consistent with the guessing bound.
In particular, embezzling states of unlimited size are allowed,
and so is backwards quantum communication, and in this way we
can also establish whatever type of entangled state we need.

In this case, the
arguments of the last section established that an inefficiency of
$\delta$ in our channel simulation (i.e. spending $n(C_E+\delta)$
bits) allows errors to be bounded by
$\leq \exp(-n\delta^2/8\log^2(d) + o(n))$.  Similarly, it is known
that $n(C_E-\delta)$  bits can be sent through $\cN^{\ot n}$ with
error $\leq \exp(-n\delta^2/8\log^2(d))$ if we are allowed a
sufficient rate of ebits. This establishes that $C_E(\cN)$
is the optimal cbit rate for simulation, and of communication, in
the strong converse sense, even if arbitrary entangled states and
back communication are for free.
Previously this was known to hold only when considering product-state
inputs~\cite{ON99,Winter99} or restricted classes of
channels~\cite{converse09,WWY13}.  Recently an alternate proof of the
entanglement-assisted strong converse has also been given based on a
more direct argument involving completely bounded norms~\cite{GW13}.
The fact that our strong converse also applies
in the setting where free back communication is allowed from Bob
to Alice is perhaps surprising given that back communication
is known to increase the classical capacity in the unassisted
case~\cite{SS09} (although not in the assisted
case~\cite{Bowen-feedback}).  One limitation of our strong converses is
that they only apply
when the entanglement-assisted capacity is exceeded, whereas
\cite{ON99,Winter99,converse09,WWY13} addressed the Holevo capacity or the
ordinary classical capacity.

While the above argument applies to arbitrary use of the channel
to communicate (and allows arbitrary input states in the simulation),
we can also establish such a strong converse in
the case of a known IID~input. Here it is not only known~\cite{DHW05}
that $\<\cN_F:\rho\> \geq
\frac{1}{2}I(R;B)[q\ra q] + \frac{1}{2}I(B;E)[qq]$, but the
corresponding protocol can be shown to have error bounded by
$2^{-n\delta'}$.  Thus, suppose a simulation existed for
$\<\cN:\rho\>$  that used
$\frac{1}{2}(I(R;B)-\delta)[q\ra q]$ and an unlimited amount of
entanglement and back communication to achieve fidelity $f$.  Then
combining this simulation
with teleportation and our coding protocol would give a method for
using cbits at rate $I(R;B)-\delta$ together with entanglement to
simulate cbits at rate $I(R;B)-\delta/2$ with fidelity $\geq f-
2^{-n\delta'}$ for some $\delta'>0$.  By causality, any such
simulation must have fidelity $\leq 2^{-n\delta/2}$, and thus we must
have $f \leq 2^{-n\delta/2} + 2^{-n\delta'}$.

\subsubsection{Converses for the use of entanglement and back communication, based on spread}
Here we have to distinguish between the channel simulation with
and without coherent feedback.

The case \emph{with} coherent feedback
is easier to handle as it places more stringent constraints on the
protocol, and so the bounds are easier to prove.
Thus we begin with this case, which corresponds to part (d) of
\thmref{qrst}.

We shall argue that entanglement spread is necessary. In fact,
we will show a larger communication cost (forward plus backward)
if the only entangled resource consists of  ebits (i.e. maximally
entangled states).
Recall that the simulation theorem for feedback channels
uses communication at rate
\be \max_{\rho_1}H(B)_{\rho_1} + \max_{\rho_2}
(H(R)-H(E))_{\rho_2} = C_E(\cN) + \deficit(\cN).
\label{eq:spread-converse-states}\ee  In what follows, we
will omit $\cN$ from our notation.  We will show that
this rate is optimal by constructing an input on which $U_\cN^{\ot n}$ will
create $\approx n(C_E + \deficit)$ spread.

For $i=1,2$, let $\rho_i$ be the states from
\eq{spread-converse-states} and let  $\ket{\psi_i}$ be a purification of $\rho_i^{\ot n}$.  Let
$$\ket\Psi = \frac{\ket{1}^A \ket{1}^B\ket{\psi_1}^{A'A}\ket{\psi_2}^{AB} +
\ket{2}^A \ket{2}^B \ket{\psi_2}^{A'B}}{\sqrt{2}}.$$
Here we use $A$ repeatedly to indicate registers under Alice's
control, $B$ to indicate registers owned by Bob, and $A'$ for a
register controlled by Alice that will be input to $U_\cN^{\ot
  n}$.  We omit describing the $\ket 0$ registers that should pad
Alice and Bob's registers so that each branch of the superposition has
the same number of qubits.
 Let $\ket{\varphi_i}^{RBE} = (U_\cN^{A'\ra BE} \ot I^R)^{\ot
  n}\ket{\psi_i}^{A'R}$.
Then,
\begin{multline}\ket{\Theta} := U_\cN^{\ot n}\ket{\Psi}
\\=
\frac{\ket{1}^A \ket{1}^B\ket{\varphi_1}^{ABE}\ket{\psi_2}^{AB} +
\ket{2}^A \ket{2}^B \ket{\varphi_2}^{BBE}}{\sqrt{2}}.
\end{multline}
Here, $E$ is again a register controlled by Alice and we observe that
$\ket{\varphi_2}$ has two out of its three registers controlled by Bob.

We argue that $\approx n(C_E + \deficit)$ spread has been created by
applying $U_\cN^{\ot n}$.  First, observe that $\ket\Psi$ is locally
equivalent to $\frac{\ket{1,1}+\ket{2,2}}{\sqrt{2}} \ot
\ket{\psi_2}^{AB}$, which has $O(\sqrt{n})$ spread.  More precisely,
$\Delta_\eps(\psi_2^A) \leq O(\sqrt{n\log(1/\eps)}\log(d))$, and so
$\ket\Psi$ can be prepared with error $\eps$ using this amount of
communication~\cite{LP99}.   Next, we argue that $\ket{\Theta}$ has a
large amount of spread.  The part attached to the $\ket{1,1}$ register
has entanglement roughly equal to  $n(H(B)_{\rho_1} + H(R)_{\rho_2})$
and the part attached to the $\ket{2,2}$ register has entanglement
roughly equal to $nH(E)_{\rho_2}$.  Since these registers are
combinations of i.i.d. states, one can prove (c.f. Theorem 13 of
\cite{HW02}, or Proposition 6 of \cite{HL02}) that for any
$\eps<1/2$, $\Delta_\eps(\Theta^A) \geq n(H(B)_{\rho_1} +
H(R)_{\rho_2} - H(E)_{\rho_2}) - O(\sqrt{n})$.  We conclude from
\thmref{spread} that simulating $U_\cN^{\ot n}$ to error lower than a
sufficiently small constant (such as $10^{-4}$) using unlimited ebits requires communication $\geq n(C_E(\cN)+\deficit(\cN)-o(1))$.
We suspect, but do not prove, that a tighter analysis could prove this
lower bound for all $\eps<1/2$. Note that our statement is not a
strong converse in the usual sense (which would demand a proof of
our bound for all $\eps<1$) but that it still establishes a
jump of the error at the optimal rate.

\medskip
When we consider simulations without feedback, we no longer have
additivity, and we are able only to establish
regularized coding theorems and weak converses.
The zero-entanglement limit is discussed
in \cite{Hayashi:EoP} and the low-entanglement regime (part (b) of
\thmref{qrst}) follows similar
lines.  The main idea is that if only coherent resources (such as
qubits and ebits) are used, then the state of the environment is
entirely comprised of what Alice and Bob discard.  Let $E_A$
(resp.~$E_B$) denote the system that Alice (resp.~Bob) discards.

Let $\cP$ denote the simulation of $\cN^{\ot n}$ constructed by the
protocol.  By the above arguments, $\cP$ has an isometric extension
$U_\cP^{A^n\ra B^nE_AE_B}$, just as $\cN^{A\ra B}$ has isometric
extension $U_\cN^{A\ra BE}$.  The fact that the simulation is
successful means that $\|\cP - \cN^{\ot n}\|_\diamond \leq \eps$.
We now make use of a generalization of Uhlmann's theorem~\cite{KSW06}
to show that
\be \|U_P - V^{E^n\ra E_AE_B}\circ U_\cN^{\ot n}\|_\diamond \leq
\sqrt\eps
\label{eq:iso-close}\ee
for some isometry $V$.

For part (b) of \thmref{qrst}, this allows us to reduce the converse
to that for part (a).  We obtain \eqs{q-low-iid}{qe-low-iid} from
Fannes' inequality.
Before discarding $E_B$, Bob's total state $B^n
E_B$ is within $\eps$ of a state on $q+e$ qubits, and thus has
$H(B^nE_B) \leq q+e  +O(n\eps)$.  Similarly, Bob has received only $q$
qubits, so we must have $\frac{1}{2}I(R^n;B^nE_B) \leq q + O(n\eps)$.

For part (e), \eq{iso-close} allows us to reduce the converse to the
converse for part (d).  Again this is because the ability to simulate
$\cN^{\ot n}$ without preserving $E$ is equivalent to the ability to
simulate $V\cdot U_\cN^{\ot n}$ for {\em some} choice of $V^{E^n\ra E_AE_B}$.

\subsubsection{The clueless Eve channel}\label{sec:clueless}
We conclude our discussion of converses with an explicit example of
a channel that requires more communication to simulate with ebits than with embezzling states [part (e) of \thmref{qrst}]. This
channel is designed so that different inputs create different
amounts of entropy for the receiver, but without leaking information
about this to the environment. Hence, we call it the ``clueless Eve
channel.''

The channel $\cN_d$ maps $d+1 \ra d+1$ dimensions.  We define it in
terms of its isometric extension as follows: \be U_{\cN_d}  =
\ket{\Phi_d}^{BE}\bra{0}^A + \sum_{i=1}^d \ket{0}^B\ket{i}^E
\bra{i}^A ,\label{eq:clueless}\ee where $\ket{\Phi_d} =
\frac{1}{\sqrt{d}}\sum_{i=1}^d \ket{i,i}$. In other words $\ket{0}$
is mapped to the maximally mixed state (over dimensions $1,\ldots,
d$) for Bob, while $\ket{i}$ is mapped to the $\ket{0}$ state for
$1\leq i \leq d$. One can show that $C_E(\cN_d)=2Q_E(\cN_d)=1$
independent of $d$ using convexity and symmetry
arguments\footnote{$C_E$ is given by the maximum of $I(A;B)$ over
  inputs $\rho$.  Due to the structure of the channel, it is invariant
under the map $\rho\ra U\rho U^\dag$ for any $U$ satisfying $U\ket 0 =
\ket 0$.  Since $I(A:B)$ is concave in the input density matrix
$\rho$, it follows that it can be maximized by $\rho$ that commutes
with all such $U$.  The resulting states have the form $p\proj 0 +
(1-p)(\sum_{i=1}^d \proj i / d)$ and the resulting one-parameter
maximization problem is easily seen to be equivalent to determining
the entanglement-assisted capacity of a noiseless classical bit
(i.e. totally dephasing) channel.} along the lines of \cite{cortese-covariant,Holevo-covariant}.  However,
the following argument will show that on some (non-tensor-power)
inputs, the channel's ebit-assisted simulation cost, even for a
non-feedback simulation, strictly exceeds its entanglement-assisted
capacity. (This may be contrasted with the case of the amplitude
damping channel considered earlier in \fig{ampl}, where the gap
between ebit-assisted simulation cost and $C_E$ is present only for
feedback simulation). To see qualitatively why standard ebits
are an insufficient entanglement resource to efficiently simulate
this channel, consider the purified non-tensor-power input \be
\ket{\Psi}^{RA^n}:= \frac{ \ket{0^n}^R\ket{0^n}^{A^n} +
\ket{\Phi_{d^n}}^{RA^n}}{\sqrt 2} \label{eq:tough-input}\ee to $n$
uses of the channel.  As usual, $R$ is a reference system and $A^n$
is sent to $B^nE^n$ by $U_\cN^{\ot n}$.
\begin{figure}[htbp]
\includegraphics[width=3.5in]{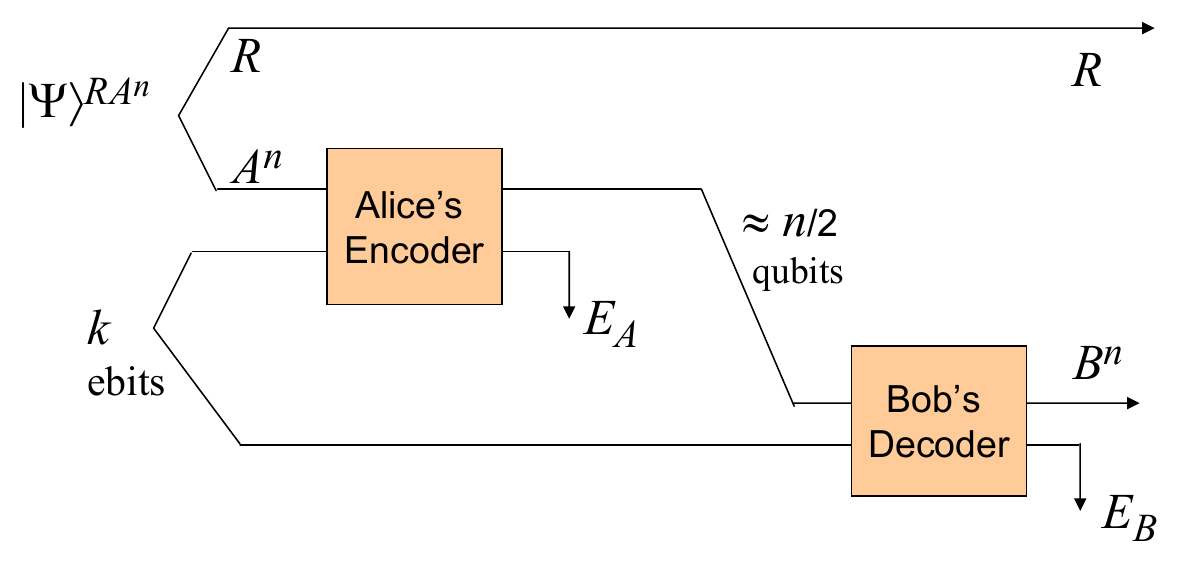}
\caption{A would-be simulation of the clueless Eve channel $\cN_d$
on the non-tensor-power source $\ket{\Psi^{RA^n}}$, using around
$nQ_E=n/2$ qubits of forward communication and $k$ ordinary ebits as
the entanglement resource, deposits different amounts of entropy in
Alice's local environment $E_A$ depending on which term in
$\Psi^{RA^n}$ is acted upon, thereby decohering the superposition
and spoiling the simulation.} \label{fig:clueless}
\end{figure}

In \fig{clueless} Alice's encoder, assisted by some number $k$ of
ebits, transforms the $A^n$ part of this input into a supposedly
small ($\approx n/2$ qubit) quantum message sent to Bob and a
residual environment system $E_A$ retained by Alice. By conservation
of entropy, if $A^n\!=\!0$, then Alice's environment $E_A$ will have
entropy at most $k+n/2$, whereas if $A^n\!\neq\!0$ it will be left
in a different state with entropy at least $k+n\log(d) -n/2$.
Because (as will be shown in the following theorem) these two states
become close to orthogonal for large $d$, Alice's
environment will gain information about which term in the
superposition $\Psi^{RA^n}$ was present, and consequently will
decohere the superposition, which a faithful simulation of the
channel would not have done.  Carrying through this argument more
precisely, we have:

\begin{theorem}
\label{thm:clueless-eve-spread} Let $\cP$ be a protocol using $q$
qubits of communication (total, in either direction) and $k$ ebits. If $\|\cP-\cN_d^{\ot n}\|_\diamond \leq \eps$ then $q \geq
\frac{1}{4}n\log d -1-8\sqrt{\eps}n\log d $.
\end{theorem}

\begin{IEEEproof}
We begin by introducing some notation.  The proof depends on the
assumed closeness between $\cP$ and $\cN_d^{\ot n}$ on the
non-tensor-power source $\Psi^{RA^n}$.  Applying $U_\cN^{\ot n}$ to
the $A^n$ part of $\Psi^{RA^n}$ and using the fact that
$\ket{\Phi_d}^{\ot n} = \ket{\Phi_{d^n}}$ we obtain the (ideal)
state \be \ket\Theta^{RB^nE^n} :=
\frac{\ket{0}^{R}\ket{\Phi_{d^n}}^{B^nE^n}
 +\ket{0^n}^{B^n}\ket{\Phi_{d^n}}^{RE^n}}{\sqrt 2}\ee
that would result from the operation of the channel. We now purify
the simulating protocol $\cP$ in a canonical way: all non-unitary
operations are replaced by isometries and discarding subsystems, and
each subsystem that Alice (resp. Bob) discards is added to a
register called $E_A$ (resp. $E_B$).  We can think of $E_A, E_B$ as
local environments. Let $U_\cP$ denote the resulting purification
and define the actual protocol output to be \be
\ket{\tilde\Theta}^{RB^nE_AE_B} := (I \ot U_\cP)\ket{\Psi} .\ee To
simulate $\cN_d^{\ot n}$ on input $\ket\Psi$ we do not need to
approximate $\ket\Theta$, but it suffices to approximate the $RB^n$
part of the state.  By Uhlmann's theorem, this is equivalent to the
claim that there exists an isometry $V:E^n\ra E_AE_B$ such that
\be\left|\braket{\tilde\Theta}{\Theta_V}\right|^2\geq 1-\eps,
\label{eq:tilde-fidelity}\ee
 where $\ket{\Theta_V}^{RB^nE_AE_B} := (I^{RB^n} \ot V^{E^n\ra E_AE_B})\ket{\Theta}$.

To prove the lower bound, we will argue that either $\ket{\Theta_V}$ has high spread for any $V$, or it has high mutual information between $R$ and $B^nE_B$.  Either way, we obtain a lower bound on the communication required to approximately create it.  We first sketch the idea of why this should be true.
Let $V\ket{\Phi_{d^n}}^{XE^n} := \ket{\varphi_V}^{XE_AE_B}$, where $X$
can be either $R$ or $B^n$.  Then
\ba \ket{\Theta_V}^{RB^nE_AE_B} & =
\frac{\ket{0^n}^R \ket{\varphi_V}^{B^nE_AE_B} + \ket{0^n}^{B^n}\ket{\varphi_V}^{RE_AE_B}}{\sqrt{2}}
\nonumber\\\nonumber
\Theta_V^{RE_A} &= \frac{\proj{0}^{R} \ot \varphi_V^{E_A} +
 \varphi_V^{RE_A}}{2}.
 \ea
To understand the spectrum of $\Theta_V^{RE_A}$, observe that
$\proj{0}\ot \varphi_V^{E_A}$ and $\varphi_V^{RE_A}$ have orthogonal
support.   Therefore, if $\varphi_V^{E_A}$ and $\varphi_V^{RE_A}$ have
spectrum
$\alpha=(\alpha_1\ldots,\alpha_a)$ and $\beta=(\beta_1,\ldots,\beta_b)$
respectively, then $\Theta_V^{RE_A}$ has eigenvalues
$$\frac{\alpha_1}{2},\ldots,\frac{\alpha_a}{2},
\frac{\beta_1}{2},\ldots,\frac{\beta_b}{2}.$$
Note also that $\varphi_V^{RE_A}$ has the same spectrum as $\varphi_V^{E_B}$. Thus
\begin{subequations}\label{eq:theta-varphi}
\begin{multline}
1+ \max(H_{0,\eps}(\varphi_V^{E_A}),H_{0,\eps}(\varphi_V^{E_B}))
\\ \geq H_{0,\eps}(\Theta_V^{RE_A})
\\ \geq \max(H_{0,2\eps}(\varphi_V^{E_A}),H_{0,2\eps}(\varphi_V^{E_B}))
 \label{eq:theta-varphi-0}
 \end{multline}
\begin{multline}
1+ \min(H_{\infty,\eps}(\varphi_V^{E_A}),H_{\infty,\eps}(\varphi_V^{E_B}))
\\ \leq H_{\infty,\eps}(\Theta_V^{RE_A}) \\
\leq 1+ \min(H_{\infty,2\eps}(\varphi_V^{E_A}),H_{\infty,2\eps}(\varphi_V^{E_B}))\label{eq:theta-varphi-inf}.
 \end{multline}
\end{subequations}

Pretend for a moment that $\eps=0$.  In that case
$H_0(\Theta_V^{RE_A}) \geq H_0(\varphi_V^{E_A}) \geq
H(\varphi_V^{E_A}) = H(\tilde\Theta^{E_A})$. Combining this with the
fact that $H_\infty \leq S$, we have that $\Delta_0(\Theta_V^{RE_A})
\geq H(\tilde\Theta^{E_A}) - H(\tilde\Theta^{RE_A}) =
-H(R|E_A)_{\tilde\Theta} = H(R)_{\tilde\Theta} -
I(R;B^nE_B)_{\tilde\Theta}$.  Rearranging we have that the sum of
the spread ($\Delta_0(\Theta_V^{RE_A})$) and the mutual information
($I(R;B^nE_B)_{\tilde\Theta}$) is at least $H(R)_{\tilde\Theta} = 1
+ \frac{1}{2}n\log d$.  Since spread and mutual information are both
$\leq 2q$, we obtain the desired result.

The difficulty in extending this argument to the $\eps>0$ case is (a)
that spread can vary dramatically under small perturbations in the
state (as observed even in situations as simple as entanglement
dilution~\cite{LP99, HW02, HL02}), and (b) that the dimensions of $E_A,E_B$ are unbounded, and so Fannes' inequality is difficult to apply.  The second difficulty is easiest to address: we will use a variant of Fannes' inequality known as the Alicki-Fannes inequality, which bounds the variation of $H(R|E_A)$ using only $|R|$ and not $|E_A|$.
\begin{lemma}[Alicki-Fannes inequality~\cite{AF04}]\label{lem:AF}
If $\eps:=\frac{1}{2}\|\rho^{XY}-\sigma^{XY}\|_1 < 1/2$ then
\be |H(X|Y)_\rho - H(X|Y)_\sigma| \leq 8\eps\log|X|  + 2H_2(2\eps)
\label{eq:AF-ineq}\ee
\end{lemma}
To address the unbounded Lipschitz constant of $\Delta_0$, we will need to look more carefully at how $\Theta_V$ and $\tilde\Theta$ are related.
%First, we use \eq{theta-varphi} to obtain \be \Delta_\eps(\Theta_V^{RE_A}) \geq H_{0,2\eps}(\varphi_V^{E_A}) - H_{\infty,2\eps}(\varphi_V^{E_B}) \ee
First, we replace $\varphi_V$ with a low-spread approximation.   From \eq{S-opt}, we obtain nonnegative operators $M_{A,0}, M_{A,\infty}, M_{B_0}, M_{B,\infty}$ whose largest eigenvalues are $\leq 1$ and that satisfy
\bit
\item $1-2\eps = \tr M_{A,0}\varphi^{E_A} = \tr M_{A,\infty}\varphi^{E_A}$
\item $1-2\eps = \tr M_{B,0}\varphi^{E_B} = \tr M_{B,\infty}\varphi^{E_B}$
\item $M_{A,0}, M_{A,\infty}, \varphi^{E_A}$ all commute
\item $M_{B,0}, M_{B,\infty}, \varphi^{E_B}$ all commute
\item $H_0(M_{A,0}\varphi^{E_A}) = H_{0,2\eps}(\varphi^{E_A})$
\item $H_\infty(M_{A,\infty}\varphi^{E_A}) = H_{\infty,2\eps}(\varphi^{E_A})$
\item $H_0(M_{B,0}\varphi^{E_B}) = H_{0,2\eps}(\varphi^{E_B})$
\item $H_\infty(M_{B,\infty}\varphi^{E_B}) = H_{\infty,2\eps}(\varphi^{E_B})$
\eit
We can now define
\be \ket{\hat\varphi} := \gamma^{-1/2}(I \ot \sqrt{M_{A,0}M_{A_\infty}} \ot
\sqrt{M_{B,0}M_{B_\infty}})\ket{\varphi_V}, \ee
where $\gamma\geq 1-8\eps$ is a normalizing constant, chosen so that $\braket{\hat\varphi}{\hat\varphi}=1$.
For ease of calculations, we will choose $\eps\leq 1/16$, so that $\log(1-8\eps) \geq -1$.
    Then
\begin{subequations}\label{eq:varphi-hat-0}
\ba H_0(\hat\varphi^{E_A}) &\leq H_{0,2\eps}(\varphi_V^{E_A}) \\
H_0(\hat\varphi^{E_B}) &\leq H_{0,2\eps}(\varphi_V^{E_B})
\ea\end{subequations}\vspace{-5mm}
\begin{subequations}\label{eq:varphi-hat-inf}\ba
H_{\infty}(\hat\varphi^{E_A}) &\geq H_{\infty,2\eps}(\varphi_V^{E_A}) + 1\\
H_{\infty}(\hat\varphi^{E_B}) &\geq H_{\infty,2\eps}(\varphi_V^{E_B}) + 1
\ea
\end{subequations}

Now we use $\ket{\hat\varphi}$ to define
\be \ket{\hat\Theta}^{RB^nE_AE_B} := \frac{\ket{0^n}^R\ket{\hat\varphi}^{B^nE_AE_B} +
\ket{0^n}^{B^n}\ket{\hat\varphi}^{RE_AE_B}}{\sqrt{2}}.\ee
Observe that $\braket{\Theta_V}{\hat\Theta} = \braket{\varphi_V}{\hat\varphi}=\sqrt{\gamma}\geq \sqrt{1-8\eps}$, implying
$\frac{1}{2}\|\Theta_V - \hat\Theta\|_1 \leq \sqrt{8\eps}$.  Combined with \eq{tilde-fidelity}, we obtain
\be \frac{1}{2}\|\tilde\Theta - \hat\Theta\|_1 \leq 4\sqrt\eps.
\label{eq:theta-dbl-approx}\ee
The advantage of $\hat\Theta$ is that it has exactly the same structure as $\Theta_V$, but with $\ket{\varphi_V}$ replaced with $\ket{\hat\varphi}$. Thus it similarly satisfies
\be \hat\Theta^{E_A} = \hat\varphi^{E_A} \quad\text{and}\quad
\hat\Theta^{E_B} = \hat\varphi^{E_B}
\label{eq:hat-theta-varphi},\ee
and thus \eq{theta-varphi} still holds when $\Theta_V$ is replaced with $\hat\Theta$ and $\varphi_V$ is replaced with $\hat\varphi$.

We now conclude with a traditional chain of entropic inequalities, with each step labeled by its justification:
\begin{subequations}
\bas
2q  \geq &\Delta_0(\tilde\Theta^{RE_A}) \\
 \geq& \Delta_\eps(\Theta^{RE_A})) & \text{\lemref{spread-normalized}} \\
 \geq& H_{0,\eps}(\Theta^{RE_A}) - H_{\infty,\eps}(\Theta^{RE_A}) & \text{\eq{spread-vs-S}}\\
 \geq &\max(H_{0,2\eps}(\varphi_V^{E_A}),H_{0,2\eps}(\varphi_V^{E_B}))
 & \text{\eq{theta-varphi-0}} \\
 & -\min(H_{\infty,2\eps}(\varphi_V^{E_A}),H_{\infty,2\eps}(\varphi_V^{E_B})) -1
 & \text{\eq{theta-varphi-inf}} \\
 \geq &\max(H_{0}(\hat\varphi_V^{E_A}),H_{0}(\hat\varphi_V^{E_B}))
 & \text{\eq{varphi-hat-0}} \\
 & -\min(H_{\infty}(\hat\varphi_V^{E_A}),H_{\infty}(\hat\varphi_V^{E_B})) - 3
 & \text{\eq{varphi-hat-inf}} \\
 \geq& H_{0}(\hat\varphi^{E_A}) - H_{\infty}(\hat\Theta^{RE_A}) -3& \text{\eq{theta-varphi-inf}}\\
 =& H_{0}(\hat\Theta^{E_A}) - H_{\infty}(\hat\Theta^{RE_A}) -3& \text{\eq{hat-theta-varphi}}\\
\geq& H(\hat\Theta^{E_A}) - H(\hat\Theta^{RE_A}) -3& \hspace{-2mm}H_0\geq S \geq H_\infty \\
=& -H(R|E_A)_{\hat\Theta} -3\\
\geq& -H(R|E_A)_{\tilde\Theta} - 32\sqrt\eps n\log d - 5 & \text{\lemref{AF}} \\
=& -H(R|E_A)_{\tilde\Theta} - \delta & \hspace{-1cm}\delta:=32\sqrt\eps n\log d + 5 \\
=& H(R)_{\tilde\Theta} - I(R;E_A)_{\tilde\Theta} - \delta \\
=& \L(1+\frac{1}{2}n\log d\R) - I(R;E_A)_{\tilde\Theta} - \delta \\
\geq& \L(1+\frac{1}{2}n\log d\R) - 2q-\delta \eas
\end{subequations}

\end{IEEEproof}

\section{Conclusion}
We conclude by summarizing the operational and technical consequences
of our work, as well as some open problems.

Operationally, we establish necessary and sufficient amounts of standard noiseless resources
for simulation of discrete memoryless quantum channels, including classical DMCs as a special
case.  As is usual in Shannon theory, simulations become efficient and faithful only in the
limit of large block size, even in cases where the simulation capacity is given by a
single-letter formula. We consider both ordinary and feedback simulations, a feedback
simulation being one in which the simulating sender coherently retains what the simulated
channel would have discarded into its environment. We consider simulations on both tensor
power sources (the quantum generalization of classical IID sources) and general sources,
which may be correlated or entangled over the multiple inputs, a distinction that becomes
important for quantum channels.  We also establish conditions for asymptotic equivalence
among channels, that is conditions under which channels can simulate one another
efficiently and reversibly, so that the capacity for channel $\cM$ to simulate $\cN$ is
the reciprocal of that for performing the simulation in the opposite direction.  Such
equivalences generally hold only in the presence of some combination of auxiliary resources,
which by themselves would have no capacity for channel simulation.  In each case, an unlimited
supply of the auxiliary resources enables asymptotically reversible cross-simulation. For cross-simulations among classical channels, shared randomness is a
necessary and sufficient auxiliary resource. For quantum channels on tensor power sources,
ordinary shared entanglement is necessary and sufficient. For quantum channels on general
sources, more general entangled states (``entanglement-embezzling states'') or combinations
of resources, such as entanglement and classical back-communication, are required.  Finally,
in many cases of interest, we quantify the loss of efficiency and reversibility when an auxiliary resource
is insufficient or absent. In this respect, we feel that our
Theorem~\ref{thm:clueless-eve-spread}
is not giving a tight bound, due to an imperfect proof technique.
One problem is that mutual information and spread are not placed on a common
footing, as they are in the case when simulating an isometry (feedback case).

On the technical side, we can now understand quantum simulations of quantum channels in
terms of three key ingredients:
\benum
\item {\em State splitting} (also known as the reverse of state
  merging~\cite{HOW05,ADHW06}) in which a {\em known} tripartite state
  $\Psi^{ABC}$ begins with $C$ held by Alice and ends with $C$ held by
  Bob.  Note that this is a coherent version of measurement
  compression~\cite{Winter:POVM}, upon which early QRST proofs were based.
\item {\em Entanglement spread}, which measures how far a state is
  from maximally entangled on some subspace~\cite{Har-spread}, and turns out to be
  necessary when protocols requiring different numbers of ebits
  need to be executed in superposition.
\item {\em Dividing the environment between Alice and Bob}, which
  starts with the
  ``Church of the Larger Hilbert Space'' principle that mixed states
  have purifications, and proceeds to the observation that in a
  protocol using only noiseless resources any simulated environment
  must WLOG be divided between the sender and receiver.  This form of
  the idea first appeared in \cite{purification} and is necessary to
  understand the low-entanglement versions of the QRST.
\eenum
These concepts were known to the quantum information theory community separately in
various contexts, but find their common use in the QRST.

A number of interesting open questions remain.  On the
technical side, we observe that for
classical channels, \lemref{flat-crst} gives low error in the worst
case, but for quantum channels, \lemref{flat-qrst} only gives
average-case bounds.  While this can be addressed by using shared
randomness catalytically (and thereby without increasing the overall
cost of the protocol), a more direct proof would be preferable.

More ambitiously, we observe that the classical and quantum reverse
Shannon theorems are incomparable because the assistance of shared
entanglement is stronger than the assistance of shared randomness even
for purely classical channels (cf the discussion at the end of
\secref{qrst-statement}).  This is in contrast to the fact that
Shannon's noisy coding theorem can be viewed as a special case of the
entanglement-assisted capacity theorem.  It would be desirable to have
a single theorem that stated the cost of simulating a channel given
the assistance of an arbitrary rate of randomness and entanglement.
Some encouraging progress in this direction is given by \cite{WHBH12, BRW13},
which shows that for QC channels (i.e. quantum input, classical
output) shared randomness can be used in place of shared
entanglement.  Another direction for generalization is to consider
simulations that use side information, along the lines of \cite{YD09,WDHW13}.

There are also new questions about additivity and regularization that
arise when considering low-entanglement simulations of quantum
channels, most of which are completely open.  For example, the
zero-entanglement point on the tradeoff curve corresponds to the
entanglement of purification~\cite{purification} whose additivity
properties are still open (but see \cite{EoP-doubt} for recent work
suggesting that it is not additive).

\section{Acknowledgments} We wish to acknowledge helpful discussions
with Paul Cuff, Patrick Hayden, Jonathan Oppenheim, Graeme Smith, John
Smolin and Mark Wilde.

\appendix
%\section{Proofs from \secref{schur-typical}}\label{app:typical-proofs}

In this appendix, we prove \lemref{joint-schur} and \lemref{spread-normalized}.

First, we prove \lemref{joint-schur}, restated below for
convenience.  We follow the proof of Section 6.4.3 of \cite{Har05},
but simplify and streamline the arguments at the cost of proving a
less general claim.

\begin{replemma}{lem:joint-schur}
Let $d = \max(d_A, d_B, d_E)$.
For any state $\ket{\varphi}^{R^nA^n}$ with $\ket{\Psi} = (I \ot
U_\cN)^{\ot n}\ket\varphi$,
\bmu \L\|\, \ket\Psi -
\sum_{(\lambda_A,\lambda_B,\lambda_E)\in T_{\cN,\delta}^n}
I\ot ((\Pi_{\lambda_B} \ot \Pi_{\lambda_E})U_\cN^{\ot n}
\Pi_{\lambda_A})\ket\varphi\R\|_1
\\\leq n^{O(d^2)}\exp\L(-n\frac{\delta^2}{8\log^2(d)}\R).
\label{eq:typicality}\emu
\end{replemma}

\begin{IEEEproof}
By the triangle inequality, the LHS of \eq{typicality} is
$$\leq \sum_{(\lambda_A,\lambda_B,\lambda_E)\not\in T_{\cN,\delta}^n}
\| (I\ot (\Pi_{\lambda_B} \ot \Pi_{\lambda_E})U_\cN^{\ot n}
\Pi_{\lambda_A})\ket\varphi\|_1 .$$
We now consider a particular triple $(\lambda_A, \lambda_B, \lambda_E)\not\in
T_{\cN,\delta}^n$.  Let
\ba \eps & = \| (I\ot (\Pi_{\lambda_B} \ot
\Pi_{\lambda_E})U_\cN^{\ot n} \Pi_{\lambda_A})\ket\varphi\|_1 \\
& = \tr (\Pi_{\lambda_B} \ot \Pi_{\lambda_E})
U_\cN^{\ot n} \Pi_{\lambda_A} \varphi^A \Pi_{\lambda_A}
(U_\cN^\dag)^{\ot n}
\label{eq:typ-err-1}
\ea

In this last step, we observe that all of the terms commute with
collective permutations except for $\varphi^A$.  Thus, \eq{typ-err-1}
is unchanged if we replace $\varphi^A$ with its symmetrized version,
$\tilde\varphi^A := \frac{1}{n!} \sum_{\pi\in \cS_n} \pi \varphi^A
\pi^{-1}$.  Next, observe that
$$\Pi_{\lambda_A} \tilde\varphi^A \Pi_{\lambda}
= %w_{\lambda_A}
\proj{\lambda_A} \ot \sigma \ot
\frac{I_{\cP_{\lambda_A}}}{\dim \cP_{\lambda_A}},$$
where %$w_\lambda = \tr \Pi_{\lambda_A}\varphi^A$ and
$\sigma$ is some
(subnormalized) density matrix on $\cQ_{\lambda_A}^{d_A}$.  This implies
that
\be \Pi_{\lambda_A} \tilde\varphi^A \Pi_{\lambda}
\leq %w_{\lambda_A} %\dim \cQ_{\lambda_A}^{d_A}
\proj{\lambda_A} \ot I_{\cQ_{\lambda_A}^{d_A}}
%\frac{I_{\cQ_{\lambda_A}^{d_A}}}{\dim \cQ_{\lambda_A}^{d_A}}
\ot \frac{I_{\cP_{\lambda_A}}}{\dim \cP_{\lambda_A}}
= \frac{\Pi_{\lambda_A}}{\dim \cP_{\lambda_A}}
\label{eq:typ-oper-ineq}\ee

Next, define the single-system density matrix $\rho = \sum_{i=1}^{d_A}
\bar\lambda_{A,i} \proj i$.  By \eq{schur-proj}, we have
$$\tr \Pi_{\lambda_A} \rho^{\ot n} \geq
(n+d)^{-d(d+1)/2}.$$
Thus, if we twirl $\rho^{\ot n}$, we find that
\bmu \frac{\Pi_{\lambda_A}}{\dim \cP_{\lambda_A}}
\leq \bbE_{U\in \cU_{d_A}}[(U\rho U^\dag)^{\ot n}]
\cdot (n+d)^{d(d+1)/2} \dim \cQ_{\lambda_A}^{d_A}
\\ \leq
\bbE_{U\in \cU_{d_A}}[(U\rho U^\dag)^{\ot n}](n+d)^{d^2},
\emu
where in the second step we have used \eq{cQ-bound}.
Combining this equation with \eq{typ-oper-ineq}, we obtain the operator
inequality
$$ \Pi_{\lambda_A} \tilde\varphi^A \Pi_{\lambda} \leq
%w_{\lambda_A}
\bbE_{U\in \cU_{d_A}}[(U\rho U^\dag)^{\ot n}](n+d)^{d^2}.$$
Let $r_B, r_E$ be the spectra respectively of the $B$ and $E$ parts of
$U_\cN U \rho U^\dag U_\cN^\dag$.  Then by the definition of
$T_{\cN,\delta}^n$ we have that $\|r_B-\bl_B\|_1 + \|r_E -
\bl_E\|_1 > \delta/\log(d)$.  Thus, at least one of these distances
must be $>\delta/2\log(d)$. By Pinsker's inequality it follows that
 either $D(\bl_B\|r_B)\geq \delta^2/8\log^2(d)$ or
$D(\bl_E\|r_E)\geq \delta^2/8\log^2(d)$.
This in turn means we can bound
\ba \eps &\leq %w_{\lambda_A}
 (n+d)^{d^2}
\tr (\Pi_{\lambda_B} \ot \Pi_{\lambda_E})
(U_\cN U\rho U^\dag U_\cN^\dag)^{\ot n}
\non\\
& \leq %w_{\lambda_A}
(n+d)^{d(3d-1)/2}\exp(-n
\max(D(\bl_B\|r_B),D(\bl_E\|r_E)))
\non\\ & \leq %w_{\lambda_A}
(n+d)^{d(3d-1)/2}\exp\L(-n\frac{\delta^2}{8\log^2(d)}\R)
\ea
Finally, we sum over all $(\lambda_A,\lambda_B,\lambda_E)\not\in
T_{\cN,\delta}^n$ to upper-bound the LHS of \eq{typicality} by
$|\cI_{d,n}|^3(n+d)^{d(3d-1)/2}\exp(-n\delta^2/8\log^2(d))
\leq (n+d)^{\frac{d(3d+5)}{2}}\exp(-n\delta^2/8\log^2(d))$.
\end{IEEEproof}

\begin{replemma}{lem:spread-normalized}
\begin{multline} \max(0,\Delta_\eps(\rho))
\\= \min\{\Delta_0(\sigma) : \frac{1}{2}\|\rho-\sigma\|_1\leq \eps, 0\leq \sigma, \tr\sigma=1\}
\label{eq:spread-equiv}
\end{multline}
\end{replemma}
\begin{IEEEproof}
If $\Delta_\eps(\rho)=\delta$ then by definition there exists $\sigma$ satisfying $\Delta_0(\sigma)=\delta$, $0\leq\sigma\leq \rho$, $\rho\sigma=\sigma\rho$ and $\tr\sigma=1-\eps$, which implies that $\|\rho-\sigma\|_1 = \tr(\rho-\sigma)=\eps$.  Let the  nonzero eigenvalues of $\sigma$ be $s_1\geq \cdots \geq s_d >0$.  Then $ds_1=2^\delta$ and $\sum_{i=1}^d s_i=1-\eps$.  We can add up to $2^\delta-(1-\eps)$ weight to these eigenvalues while keeping them all $\leq s_1$.   Thus, if
\be \eps \leq 2^\delta-(1-\eps),
\label{eq:eps-condition}\ee
 then we can add $\eps$ weight to $\sigma$, thus obtaining a normalized state, without increasing its $\Delta_0$.  Call the resulting density matrix $\omega$.  Then $\Delta_0(\omega)=\delta$ and $\|\omega-\rho\|_1 \leq 2\eps$ by the triangle inequality.   This is possible whenever \eq{eps-condition} holds, or equivalently, whenever $\delta\geq 0$.

 If $\delta<0$, then we cannot create a normalized state without increasing the spread, since any normalized state has $\Delta_0 \geq 0$.  Instead we will take $\omega$ to be the maximally mixed state on $\supp\sigma$.  Note that $\sigma\leq \omega$, since $s_1 = 2^\delta / d < 1/d$.  Thus $\|\omega-\sigma\|_1 = \tr(\omega-\sigma)=\eps$ and we again have $\|\omega-\rho\|_1 \leq 2\eps$.

 This establishes that the RHS of \eq{spread-equiv} is $\leq$ the LHS.  To show the other direction, suppose that there exists a normalized $\omega$ satisfying $\frac{1}{2}\|\rho-\omega\|_1\leq \eps$ and $\Delta_0(\omega)=\delta$.  Then we can write $\rho-\omega=A-B$ where $A,B\geq 0$ and $\tr A=\tr B=\eps$.  Define $\sigma=\rho-A=\omega-B$.  Then $\tr\sigma=1-\eps$, $0\leq \sigma\leq \rho$ and $\sigma\leq \omega$, implying $\Delta_0(\sigma)\leq \Delta_0(\omega)=\delta$.
\end{IEEEproof}


\begin{thebibliography}{10}

\bibitem{ADHW06}
A.~Abeyesinghe, I.~Devetak, P.~Hayden, and A.~Winter.
\newblock The mother of all protocols: Restructuring quantum information's
  family tree.
\newblock {\em Proc. Roc. Soc. A}, 465(2108):2537--2563, 2009,
  \href{http://arxiv.org/abs/quant-ph/0606225}{{\ttfamily
  arXiv:quant-ph/0606225}}.

\bibitem{AC97}
C.~Adami and N.~J. Cerf.
\newblock von {Neumann} capacity of noisy quantum channels.
\newblock {\em Phys. Rev. A}, 56:3470--3483, Nov 1997,
  \href{http://arxiv.org/abs/quant-ph/9609024}{{\ttfamily
  arXiv:quant-ph/9609024}}.

\bibitem{AhlswedeCsiszar1}
R.~Ahlswede and I.~Csisz\'{a}r.
\newblock {Common Randomness in Information Theory and Cryptography Part I:
  Secret Sharing}.
\newblock {\em IEEE Trans. Inf. Theory}, 39(4):1121--1132, 1993.

\bibitem{AhlswedeCsiszar2}
R.~Ahlswede and I.~Csisz\'{a}r.
\newblock {Common Randomness in Information Theory and Cryptography Part II: CR
  capacity}.
\newblock {\em IEEE Trans. Inf. Theory}, 44(1):225--240, 1998.

\bibitem{AW02}
R.~Ahlswede and A.~Winter.
\newblock Strong converse for identification via quantum channels.
\newblock {\em IEEE Trans. Inf. Theory}, 48(3):569--579, 2002,
  \href{http://arxiv.org/abs/quant-ph/0012127}{{\ttfamily
  arXiv:quant-ph/0012127}}.

\bibitem{AF04}
R.~Alicki and M.~Fannes.
\newblock Continuity of quantum conditional information.
\newblock {\em J. Phys. A}, 37:L55--L57, 2004,
  \href{http://arxiv.org/abs/quant-ph/0312081}{{\ttfamily
  arXiv:quant-ph/0312081}}.

\bibitem{Aud-ineq}
K.~M.~R. Audenaert.
\newblock A sharp continuity estimate for the von {Neumann} entropy.
\newblock {\em J. Phys. A}, 40(28):8127, 2007.
\newblock quant-ph/0610146.

\bibitem{BBPS96}
C.~H. Bennett, H.~J. Bernstein, S.~Popescu, and B.~Schumacher.
\newblock Concentrating partial entanglement by local operations.
\newblock {\em Phys. Rev. A}, 53:2046--2052, 1996,
  \href{http://arxiv.org/abs/quant-ph/9511030}{{\ttfamily
  arXiv:quant-ph/9511030}}.

\bibitem{BBCJPW98}
C.~H. Bennett, G.~Brassard, C.~Cr\'{e}peau, R.~Jozsa, A.~Peres, and W.~K.
  Wootters.
\newblock Teleporting an unknown quantum state via dual classical and
  {Einstein-Podolsky-Rosen} channels.
\newblock {\em Phys. Rev. Lett.}, 70:1895--1899, 1993.

\bibitem{BDSS04}
C.~H. Bennett, I.~Devetak, P.~W. Shor, and J.~A. Smolin.
\newblock Inequalities and separations among assisted capacities of quantum
  channels.
\newblock {\em Phys. Rev. Lett.}, 96:150502, 2006,
  \href{http://arxiv.org/abs/quant-ph/0406086}{{\ttfamily
  arXiv:quant-ph/0406086}}.

\bibitem{BHL06}
C.~H. Bennett, A.~W. Harrow, and S.~Lloyd.
\newblock Universal quantum data compression via gentle tomography.
\newblock {\em Phys. Rev. A}, 73:032336, 2006,
  \href{http://arxiv.org/abs/quant-ph/0403078}{{\ttfamily
  arXiv:quant-ph/0403078}}.

\bibitem{BHLSW03}
C.~H. Bennett, P.~Hayden, D.~W. Leung, P.~W. Shor, and A.~J. Winter.
\newblock Remote preparation of quantum states.
\newblock {\em IEEE Trans. Inf. Theory}, 51(1):56--74, 2005,
  \href{http://arxiv.org/abs/quant-ph/0307100}{{\ttfamily
  arXiv:quant-ph/0307100}}.

\bibitem{BSST99}
C.~H. Bennett, P.~W. Shor, J.~A. Smolin, and A.~Thapliyal.
\newblock Entanglement-assisted classical capacity of noisy quantum channels.
\newblock {\em Phys. Rev. Lett.}, 83:3081--3084, 1999,
  \href{http://arxiv.org/abs/quant-ph/9904023}{{\ttfamily
  arXiv:quant-ph/9904023}}.

\bibitem{BSST01}
C.~H. Bennett, P.~W. Shor, J.~A. Smolin, and A.~Thapliyal.
\newblock Entanglement-assisted capacity of a quantum channel and the reverse
  {S}hannon theorem.
\newblock {\em IEEE Trans. Inf. Theory}, 48:2637--2655, 2002,
  \href{http://arxiv.org/abs/quant-ph/0106052}{{\ttfamily
  arXiv:quant-ph/0106052}}.

\bibitem{BW92}
C.~H. Bennett and S.~J. Wiesner.
\newblock Communication via one- and two-particle operators on
  {E}instein-{P}odolsky-{R}osen states.
\newblock {\em Phys. Rev. Lett.}, 69:2881--2884, 1992.

\bibitem{BCR11}
M.~Berta, M.~Christandl, and R.~Renner.
\newblock A conceptually simple proof of the quantum reverse {Shannon} theorem.
\newblock {\em Comm. Math. Phys.}, 306(3):579--615, 2011,
  \href{http://arxiv.org/abs/0912.3805}{{\ttfamily arXiv:0912.3805}}.

\bibitem{BRW13}
M.~Berta, J.~M. Renes, and M.~M. Wilde.
\newblock Identifying the information gain of a quantum measurement, 2013,
  \href{http://arxiv.org/abs/1301.1594}{{\ttfamily arXiv:1301.1594}}.

\bibitem{Bowen-feedback}
G.~Bowen.
\newblock Quantum feedback channels.
\newblock {\em IEEE Trans. Inf. Theory}, 50(10):2429--2434, 2004,
  \href{http://arxiv.org/abs/quant-ph/0209076}{{\ttfamily
  arXiv:quant-ph/0209076}}.

\bibitem{Bowen-feedbacks}
G.~Bowen.
\newblock Feedback in quantum communication.
\newblock {\em Int. J. Quant. Info.}, 3(01):123--127, 2005,
  \href{http://arxiv.org/abs/quant-ph/0410191}{{\ttfamily
  arXiv:quant-ph/0410191}}.

\bibitem{EoP-doubt}
J.~Chen and A.~Winter.
\newblock Non-additivity of the entanglement of purification (beyond reasonable
  doubt), 2012,  \href{http://arxiv.org/abs/1206.1307}{{\ttfamily
  arXiv:1206.1307}}.

\bibitem{matthias}
M.~Christandl.
\newblock {\em The structure of bipartite quantum states: Insights from group
  theory and cryptography}.
\newblock PhD thesis, University of Cambridge, 2006,
  \href{http://arxiv.org/abs/quant-ph/0604183}{{\ttfamily
  arXiv:quant-ph/0604183}}.

\bibitem{CKR09}
M.~Christandl, R.~Koenig, and R.~Renner.
\newblock Post-selection technique for quantum channels with applications to
  quantum cryptography.
\newblock {\em Phys. Rev. Lett.}, 102:020504, 2009,
  \href{http://arxiv.org/abs/0809.3019}{{\ttfamily arXiv:0809.3019}}.

\bibitem{CM04}
M.~Christandl and G.~Mitchison.
\newblock The spectra of density operators and the {Kronecker} coefficients of
  the symmetric group.
\newblock {\em Commun. Math. Phys.}, 261(3):789--797, 2006,
  \href{http://arxiv.org/abs/quant-ph/0409016}{{\ttfamily
  arXiv:quant-ph/0409016}}.

\bibitem{CW05}
M.~Christandl and A.~Winter.
\newblock Uncertainty, monogamy, and locking of quantum correlations.
\newblock {\em Information Theory, IEEE Transactions on}, 51(9):3159 -- 3165,
  sept. 2005,  \href{http://arxiv.org/abs/quant-ph/0501090}{{\ttfamily
  arXiv:quant-ph/0501090}}.

\bibitem{cortese-covariant}
J.~Cortese.
\newblock {Holevo-Schumacher-Westmoreland} channel capacity for a class of
  qudit unital channels.
\newblock {\em Phys. Rev. A}, 69:022302, Feb 2004,
  \href{http://arxiv.org/abs/quant-ph/0211093}{{\ttfamily
  arXiv:quant-ph/0211093}}.

\bibitem{CT91}
T.~M. Cover and J.~A. Thomas.
\newblock {\em Elements of Information Theory}.
\newblock Series in Telecommunication. John Wiley and Sons, New York, 1991.

\bibitem{CLMW11}
T.~S. Cubitt, D.~Leung, W.~Matthews, and A.~Winter.
\newblock Zero-error channel capacity and simulation assisted by non-local
  correlations.
\newblock {\em IEEE Trans. Inf. Theory}, 57(8):5509--5523, Aug 2011,
  \href{http://arxiv.org/abs/1003.3195}{{\ttfamily arXiv:1003.3195}}.

\bibitem{Cuff08}
P.~Cuff.
\newblock Communication requirements for generating correlated random
  variables.
\newblock {\em Proc. IEEE Symp. on Info. Th.}, 2008,
  \href{http://arxiv.org/abs/0805.0065}{{\ttfamily arXiv:0805.0065}}.

\bibitem{coordination}
P.~W. Cuff, H.~H. Permuter, and T.~M. Cover.
\newblock Coordination capacity.
\newblock {\em Information Theory, IEEE Transactions on}, 56(9):4181--4206,
  2010,  \href{http://arxiv.org/abs/0909.2408}{{\ttfamily arXiv:0909.2408}}.

\bibitem{DH-schubert}
S.~Daftuar and P.~Hayden.
\newblock Quantum state transformations and the {Schubert} calculus.
\newblock {\em Annals of Physics}, 315(1):80--122, 2005,
  \href{http://arxiv.org/abs/quant-ph/0410052}{{\ttfamily
  arXiv:quant-ph/0410052}}.

\bibitem{Datta-distort}
N.~Datta, M.-H. Hsieh, and M.~Wilde.
\newblock Quantum rate distortion, reverse shannon theorems, and source-channel
  separation.
\newblock {\em IEEE Trans. Inf. Theory}, 59(1):615--630, 2013,
  \href{http://arxiv.org/abs/1108.4940}{{\ttfamily arXiv:1108.4940}}.

\bibitem{devetak-triangle}
I.~Devetak.
\newblock Triangle of dualities between quantum communication protocols.
\newblock {\em Phys. Rev. Lett.}, 97(14):140503, Oct 2006,
  \href{http://arxiv.org/abs/quant-ph/0505138}{{\ttfamily
  arXiv:quant-ph/0505138}}.

\bibitem{DHW05}
I.~Devetak, A.~Harrow, and A.~Winter.
\newblock A resource framework for quantum {S}hannon theory.
\newblock {\em IEEE Trans. Inf. Theory}, 54(10):4587--4618, Oct 2008,
  \href{http://arxiv.org/abs/quant-ph/0512015}{{\ttfamily
  arXiv:quant-ph/0512015}}.

\bibitem{DJKR06}
I.~Devetak, M.~Junge, C.~King, and M.~B. Ruskai.
\newblock Multiplicativity of completely bounded p-norms implies a new
  additivity result.
\newblock {\em Commun. Math. Phys}, 266:37--63, 2006,
  \href{http://arxiv.org/abs/quant-ph/0506196}{{\ttfamily
  arXiv:quant-ph/0506196}}.

\bibitem{DY08}
I.~Devetak and J.~Yard.
\newblock Exact cost of redistributing multipartite quantum states.
\newblock {\em Phys. Rev. Lett.}, 100:230501, 2008,
  \href{http://arxiv.org/abs/quant-ph/0612050}{{\ttfamily
  arXiv:quant-ph/0612050}}.

\bibitem{Fannes73}
M.~Fannes.
\newblock A continuity property of the entropy density for spin lattices.
\newblock {\em Commun.~Math.~Phys.}, 31:291--294, 1973.

\bibitem{GW98}
R.~Goodman and N.~Wallach.
\newblock {\em Representations and Invariants of the Classical Groups}.
\newblock Cambridge University Press, 1998.

\bibitem{GW13}
M.~K. Gupta and M.~M. Wilde.
\newblock Multiplicativity of completely bounded p-norms implies a strong
  converse for entanglement-assisted capacity, 2013,
  \href{http://arxiv.org/abs/1310.7028}{{\ttfamily arXiv:1310.7028}}.

\bibitem{HV93}
T.~S. Han and S.~Verd{\'u}.
\newblock Approximation theory of output statistics.
\newblock {\em IEEE Trans. Inf. Theory}, 39(3):752--752, 1993.

\bibitem{Har03}
A.~W. Harrow.
\newblock Coherent communication of classical messages.
\newblock {\em Phys. Rev. Lett.}, 92:097902, 2004,
  \href{http://arxiv.org/abs/quant-ph/0307091}{{\ttfamily
  arXiv:quant-ph/0307091}}.

\bibitem{Har05}
A.~W. Harrow.
\newblock {\em Applications of coherent classical communication and {S}chur
  duality to quantum information theory}.
\newblock PhD thesis, M.I.T., Cambridge, MA, 2005,
  \href{http://arxiv.org/abs/quant-ph/0512255}{{\ttfamily
  arXiv:quant-ph/0512255}}.

\bibitem{Har-spread}
A.~W. Harrow.
\newblock Entanglement spread and clean resource inequalities.
\newblock In {\em Proc. 16th Intl. Cong. Math. Phys.}, pages 536--540, 2009,
  \href{http://arxiv.org/abs/0909.1557}{{\ttfamily arXiv:0909.1557}}.

\bibitem{HHL03}
A.~W. Harrow, P.~Hayden, and D.~W. Leung.
\newblock Superdense coding of quantum states.
\newblock {\em Phys. Rev. Lett.}, 92:187901, 2004,
  \href{http://arxiv.org/abs/quant-ph/0307221}{{\ttfamily
  arXiv:quant-ph/0307221}}.

\bibitem{HL02}
A.~W. Harrow and H.-K. Lo.
\newblock A tight lower bound on the classical communication cost of
  entanglement dilution.
\newblock {\em IEEE Trans. Inf. Theory}, 50(2):319--327, 2004,
  \href{http://arxiv.org/abs/quant-ph/0204096}{{\ttfamily
  arXiv:quant-ph/0204096}}.

\bibitem{Hayashi:02e}
M.~Hayashi.
\newblock Exponents of quantum fixed-length pure state source coding.
\newblock {\em Phys. Rev. A}, 66:032321, 2002,
  \href{http://arxiv.org/abs/quant-ph/0202002}{{\ttfamily
  arXiv:quant-ph/0202002}}.

\bibitem{Hayashi:EoP}
M.~Hayashi.
\newblock Optimal visible compression rate for mixed states is determined by
  entanglement of purificatio.
\newblock {\em Phys. Rev. A}, 73:060301(R), 2006,
  \href{http://arxiv.org/abs/quant-ph/0511267}{{\ttfamily
  arXiv:quant-ph/0511267}}.

\bibitem{Hayashi-post}
M.~Hayashi.
\newblock Universal approximation of multi-copy states and universal quantum
  lossless data compression.
\newblock {\em Comm. Math. Phys.}, 293(1):171--183, 2010,
  \href{http://arxiv.org/abs/0806.1091}{{\ttfamily arXiv:0806.1091}}.

\bibitem{Hayashi:01b}
M.~Hayashi and K.~Matsumoto.
\newblock Variable length universal entanglement concentration by local
  operations and its application to teleportation and dense coding, 2001,
  \href{http://arxiv.org/abs/quant-ph/0109028}{{\ttfamily
  arXiv:quant-ph/0109028}}.

\bibitem{Hayashi:02b}
M.~Hayashi and K.~Matsumoto.
\newblock Quantum universal variable-length source coding.
\newblock {\em Phys. Rev. A}, 66(2):022311, 2002,
  \href{http://arxiv.org/abs/quant-ph/0202001}{{\ttfamily
  arXiv:quant-ph/0202001}}.

\bibitem{Hayashi:02c}
M.~Hayashi and K.~Matsumoto.
\newblock Simple construction of quantum universal variable-length source
  coding.
\newblock {\em Quantum Inf. Comput.}, 2:519--529, 2002,
  \href{http://arxiv.org/abs/quant-ph/0209124}{{\ttfamily
  arXiv:quant-ph/0209124}}.

\bibitem{Hayashi:02a}
M.~Hayashi and K.~Matsumoto.
\newblock Universal distortion-free entanglement concentration.
\newblock {\em Phys. Rev. A}, 75:062338, 2007,
  \href{http://arxiv.org/abs/quant-ph/0209030}{{\ttfamily
  arXiv:quant-ph/0209030}}.

\bibitem{HW02}
P.~Hayden and A.~Winter.
\newblock On the communication cost of entanglement transformations.
\newblock {\em Phys. Rev. A}, 67:012306, 2003,
  \href{http://arxiv.org/abs/quant-ph/0204092}{{\ttfamily
  arXiv:quant-ph/0204092}}.

\bibitem{Hoeffding63}
W.~Hoeffding.
\newblock Probability inequalities for sums of bounded random variables.
\newblock {\em Journal of the {American} Statistical Association},
  58(1):13--30, March 1963.

\bibitem{Holevo02}
A.~S. Holevo.
\newblock On entanglement assisted classical capacity.
\newblock {\em J. Math. Phys.}, 43(9):4326--4333, 2002,
  \href{http://arxiv.org/abs/quant-ph/0106075}{{\ttfamily
  arXiv:quant-ph/0106075}}.

\bibitem{Holevo-covariant}
A.~S. Holevo.
\newblock Remarks on the classical capacity of quantum channel, 2002,
  \href{http://arxiv.org/abs/quant-ph/0212025}{{\ttfamily
  arXiv:quant-ph/0212025}}.

\bibitem{Holevo07}
A.~S. Holevo.
\newblock Complementary channels and the additivity problem.
\newblock {\em Theory of Probability \& Its Applications}, 51(1):92--100, 2007.
\newblock quant-ph/0509101.

\bibitem{HOW05}
M.~Horodecki, J.~Oppenheim, and A.~Winter.
\newblock Quantum information can be negative.
\newblock {\em Nature}, 436:673---676, 2005,
  \href{http://arxiv.org/abs/quant-ph/0505062}{{\ttfamily
  arXiv:quant-ph/0505062}}.

\bibitem{JP03}
R.~Jozsa and S.~Presnell.
\newblock Universal quantum information compression and degrees of prior
  knowledge.
\newblock {\em Proc. Roy. Soc. London Ser. A}, 459:3061--3077, October 2003,
  \href{http://arxiv.org/abs/quant-ph/0210196}{{\ttfamily
  arXiv:quant-ph/0210196}}.

\bibitem{keyl06}
M.~Keyl.
\newblock Quantum state estimation and large deviations.
\newblock {\em Rev. Mod. Phys.}, 18(1):19--60, 2006,
  \href{http://arxiv.org/abs/quant-ph/0412053}{{\ttfamily
  arXiv:quant-ph/0412053}}.

\bibitem{Kitaev:02a}
A.~Y. Kitaev, A.~H. Shen, and M.~N. Vyalyi.
\newblock {\em Classical and Quantum Computation}, volume~47 of {\em Graduate
  Studies in Mathematics}.
\newblock AMS, 2002.

\bibitem{converse09}
R.~Koenig and S.~Wehner.
\newblock A strong converse for classical channel coding using entangled
  inputs.
\newblock {\em Phys. Rev. Lett.}, 103:070504, 2009,
  \href{http://arxiv.org/abs/0903.2838}{{\ttfamily arXiv:0903.2838}}.

\bibitem{KSW06}
D.~Kretschmann, D.~Schlingemann, and R.~F. Werner.
\newblock The information-disturbance tradeoff and the continuity of
  {Stinespring's} representation.
\newblock {\em IEEE Trans. Inf. Theory}, 54(4):1708--1717, April 2006,
  \href{http://arxiv.org/abs/quant-ph/0605009}{{\ttfamily
  arXiv:quant-ph/0605009}}.

\bibitem{KW03}
D.~Kretschmann and R.~F. Werner.
\newblock {\em Tema Con Variazioni:} quantum channel capacity.
\newblock {\em New J.~Phys.}, 6:26, 2004,
  \href{http://arxiv.org/abs/quant-ph/0311037}{{\ttfamily
  arXiv:quant-ph/0311037}}.

\bibitem{LP99}
H.-K. Lo and S.~Popescu.
\newblock The classical communication cost of entanglement manipulation: Is
  entanglement an inter-convertible resource?
\newblock {\em Phys. Rev. Lett.}, 83:1459--1462, 1999,
  \href{http://arxiv.org/abs/quant-ph/9902045}{{\ttfamily
  arXiv:quant-ph/9902045}}.

\bibitem{MP00}
S.~Massar and S.~Popescu.
\newblock Amount of information obtained by a quantum measurement.
\newblock {\em Phys. Rev. A}, 61:062303, 2000,
  \href{http://arxiv.org/abs/quant-ph/9907066}{{\ttfamily
  arXiv:quant-ph/9907066}}.

\bibitem{MW01}
S.~Massar and A.~Winter.
\newblock Compression of quantum measurement operations.
\newblock {\em Phys. Rev. A}, 64(012311), 2003,
  \href{http://arxiv.org/abs/quant-ph/0012128}{{\ttfamily
  arXiv:quant-ph/0012128}}.

\bibitem{Nielsen99a}
M.~A. Nielsen.
\newblock Conditions for a class of entanglement transformations.
\newblock {\em Phys. Rev. Lett.}, 83:436--439, 1999,
  \href{http://arxiv.org/abs/quant-ph/9811053}{{\ttfamily
  arXiv:quant-ph/9811053}}.

\bibitem{ON99}
T.~Ogawa and H.~Nagaoka.
\newblock Strong converse to the quantum channel coding theorem.
\newblock {\em IEEE Trans. Inf. Theory}, 45(7):2486--2489, 1999,
  \href{http://arxiv.org/abs/quant-ph/9808063}{{\ttfamily
  arXiv:quant-ph/9808063}}.

\bibitem{CB-book}
V.~Paulsen.
\newblock {\em Completely Bounded Maps and Operator Algebras}.
\newblock Cambridge University Press, 2003.

\bibitem{Pinsker64}
M.~Pinsker.
\newblock {\em Information and Information Stability of Random Variables and
  Processes}.
\newblock Holden-Day, San Francisco, 1964.

\bibitem{Shannon48}
C.~E. Shannon.
\newblock A mathematical theory of communication.
\newblock {\em Bell System Tech.~Jnl.}, 27:379--423, 623--656, 1948.

\bibitem{SS09}
G.~Smith and J.~A. Smolin.
\newblock Extensive nonadditivity of privacy.
\newblock {\em Phys. Rev. Lett.}, 103:120503, 2009,
  \href{http://arxiv.org/abs/0904.4050}{{\ttfamily arXiv:0904.4050}}.

\bibitem{SS96}
Y.~Steinberg and S.~Verd{\'u}.
\newblock Simulation of random processes and rate-distortion theory.
\newblock {\em IEEE Trans. Inf. Theory}, 42(1):63--86, Jan 1996.

\bibitem{TGW13}
M.~Takeoka, S.~Guha, and M.~M. Wilde.
\newblock The squashed entanglement of a quantum channel, 2013,
  \href{http://arxiv.org/abs/1310.0129}{{\ttfamily arXiv:1310.0129}}.

\bibitem{purification}
B.~M. Terhal, M.~Horodecki, D.~W. Leung, and D.~P. DiVincenzo.
\newblock The entanglement of purification.
\newblock {\em J. Math. Phys.}, 43(9):4286--4298, 2002,
  \href{http://arxiv.org/abs/quant-ph/0202044}{{\ttfamily
  arXiv:quant-ph/0202044}}.

\bibitem{Tropp-LD}
J.~A. Tropp.
\newblock User-friendly tail bounds for sums of random matrices, 2010,
  \href{http://arxiv.org/abs/1004.4389}{{\ttfamily arXiv:1004.4389}}.

\bibitem{Uhlmann76}
A.~Uhlmann.
\newblock The 'transition probability' in the state space of a $*$-algebra.
\newblock {\em Rep.~Math.~Phys.}, 9:273--279, 1976.

\bibitem{vDH03}
W.~van Dam and P.~Hayden.
\newblock Universal entanglement transformations without communication.
\newblock {\em Phys. Rev. A}, 67(6):060302(R), 2003,
  \href{http://arxiv.org/abs/quant-ph/0201041}{{\ttfamily
  arXiv:quant-ph/0201041}}.

\bibitem{WH02}
R.~F. Werner and A.~S. Holevo.
\newblock Counterexample to an additivity conjecture for output purity of
  quantum channels.
\newblock {\em Journal of Mathematical Physics}, 43(9):4353--4357, 2002,
  \href{http://arxiv.org/abs/quant-ph/0203003}{{\ttfamily
  arXiv:quant-ph/0203003}}.

\bibitem{WDHW13}
M.~Wilde, N.~Datta, M.-H. Hsieh, and A.~Winter.
\newblock Quantum rate-distortion coding with auxiliary resources.
\newblock {\em IEEE Trans. Inf. Theory}, 59(10):6755--6773, Oct 2013,
  \href{http://arxiv.org/abs/1212.5316}{{\ttfamily arXiv:1212.5316}}.

\bibitem{WHBH12}
M.~M. Wilde, P.~Hayden, F.~Buscemi, and M.-H. Hsieh.
\newblock The information-theoretic costs of simulating quantum measurements.
\newblock {\em Journal of Physics A: Mathematical and Theoretical},
  45(45):453001, 2012,  \href{http://arxiv.org/abs/1206.4121}{{\ttfamily
  arXiv:1206.4121}}.

\bibitem{WWY13}
M.~M. Wilde, A.~Winter, and D.~Yang.
\newblock Strong converse for the classical capacity of entanglement-breaking
  and hadamard channels.
\newblock 2013,  \href{http://arxiv.org/abs/1306.1586}{{\ttfamily
  arXiv:1306.1586}}.

\bibitem{Winter99}
A.~Winter.
\newblock Coding theorem and strong converse for quantum channels.
\newblock {\em IEEE Trans. Inf. Theory}, 45(7):2481--2485, 1999.

\bibitem{Winter:RST}
A.~Winter.
\newblock Compression of sources of probability distributions and density
  operators, 2002,  \href{http://arxiv.org/abs/quant-ph/0208131}{{\ttfamily
  arXiv:quant-ph/0208131}}.

\bibitem{Winter:POVM}
A.~Winter.
\newblock ``{E}xtrinsic'' and ``intrinsic'' data in quantum measurements:
  asymptotic convex decomposition of positive operator valued measures.
\newblock {\em Comm. Math. Phys.}, 244(1):157--185, 2004,
  \href{http://arxiv.org/abs/quant-ph/0109050}{{\ttfamily
  arXiv:quant-ph/0109050}}.

\bibitem{Winter-triples}
A.~Winter.
\newblock Secret, public and quantum correlation cost of triples of random
  variables.
\newblock In {\em 2005 IEEE International Symposium on Information Theory},
  pages 2270---2274, 2005.

\bibitem{Winter-ident}
A.~Winter.
\newblock Identification via quantum channels in the presence of prior
  correlation and feedback.
\newblock In {\em General Theory of Information Transfer and Combinatorics},
  volume 4123 of {\em Lecture Notes in Computer Science}, pages 486--504.
  Springer Berlin Heidelberg, 2006,
  \href{http://arxiv.org/abs/quant-ph/0403203}{{\ttfamily
  arXiv:quant-ph/0403203}}.

\bibitem{WZ82}
W.~Wootters and W.~Zurek.
\newblock A single quantum cannot be cloned.
\newblock {\em Nature}, 299:802--803, 1982.

\bibitem{Wyner75}
A.~Wyner.
\newblock The common information of two dependent random variables.
\newblock {\em IEEE Trans. Inf. Theory}, 21(2):163--179, 1975.

\bibitem{YD09}
J.~Yard and I.~Devetak.
\newblock Optimal quantum source coding with quantum information at the encoder
  and decoder.
\newblock {\em IEEE Trans. Inf. Theory}, 55(11):5339--5351, Nov 2009,
  \href{http://arxiv.org/abs/0706.2907}{{\ttfamily arXiv:0706.2907}}.

\end{thebibliography}
\end{document}